%% file: properties.tex
\documentclass[11pt,aps,pra,notitlepage,nofootinbib,tightenlines]{revtex4-1}

\usepackage[T1]{fontenc}
\usepackage[english]{babel}
\usepackage[utf8]{inputenc}

\usepackage[dvipsnames]{xcolor}
\usepackage{amsmath}
\usepackage{amssymb}
\usepackage{amsthm}
\usepackage{mathtools}
\usepackage{tikz}
\usepackage[colorlinks=true,linkcolor=blue,citecolor=blue]{hyperref}
\usepackage{pgfplots}
\usepackage{enumitem}
\setlist[enumerate]{noitemsep,partopsep=0pt,parsep=0pt}
\setlist[itemize]{noitemsep,partopsep=0pt,parsep=0pt}
\usepackage{booktabs}
\newtheorem{thm}{Theorem}
\newtheorem{lem}{Lemma}
\newtheorem{prop}{Proposition}
\theoremstyle{remark}
\newtheorem{rem}{Remark}

\numberwithin{equation}{section}

\DeclareMathOperator{\tr}{tr}
\DeclareMathOperator{\rank}{rank}
\DeclareMathOperator{\spec}{spec}

\DeclareMathOperator{\supp}{supp}
\DeclareMathOperator{\CPTP}{CPTP}
\DeclareMathOperator*{\argmin}{\arg\min}
\DeclareMathOperator*{\argmax}{\arg\max}
\DeclareMathOperator{\sym}{sym}
\newcommand{\ket}[1]{|#1\rangle}
\newcommand{\bra}[1]{\langle#1|}
\newcommand{\ketbra}[2]{\ket{#1}\!\bra{#2}}
\newcommand{\proj}[1]{\ketbra{#1}{#1}}
\newcommand{\braket}[2]{\langle#1|#2\rangle}
\newcommand{\brak}[1]{\braket{#1}{#1}}
\allowdisplaybreaks

\begin{document}

\title{Doubly minimized Petz and sandwiched R\'enyi mutual information: Properties}
\author{Laura Burri}
\affiliation{Institute for Theoretical Physics, ETH Zurich, Zurich, Switzerland}

\begin{abstract}
The doubly minimized Petz R\'enyi mutual information of order $\alpha$ is defined as the minimization of the Petz divergence of order $\alpha$ of a fixed bipartite quantum state relative to any product state. The doubly minimized sandwiched R\'enyi mutual information is defined analogously using the sandwiched divergence in place of the Petz divergence. 
In this work, we establish several properties of these two types of R\'enyi mutual information. 
In particular, for the Petz case, we prove additivity for $\alpha\in [1/2,2]$. 
For the sandwiched case, we establish a novel duality relation for $\alpha\in [2/3,\infty]$ via Sion's minimax theorem, and we subsequently use this duality relation to prove additivity for the same range of $\alpha$. Previously, additivity for the sandwiched case was known only for $\alpha\in [1,\infty]$, but it had been conjectured to hold for $\alpha\in [1/2,\infty]$.
\end{abstract}

\maketitle
%\tableofcontents

\section{Introduction}\label{sec:1}

\subsection{Motivation and background}
The relationship between operational tasks and information measures is a central topic of study in information theory. 
This paper is concerned with R\'enyi generalizations of the mutual information. 
We will now elucidate their relation to error exponents of certain binary discrimination tasks in classical information theory.

\paragraph*{Classical mutual information.} The classical mutual information can be expressed in terms of the Kullback-Leibler divergence $D$ in several ways~\cite[Proposition~8]{lapidoth2019two}:
\begin{equation}\label{eq:shannon-info}
I(X:Y)_P
=D(P_{XY}\|P_XP_Y)
=\inf_{R_Y} D(P_{XY}\|P_X R_Y)
=\inf_{Q_X,R_Y} D(P_{XY}\|Q_X R_Y).
\end{equation}
Here, $X$ and $Y$ are random variables over finite alphabets $\mathcal{X}$ and $\mathcal{Y}$, $P_{XY}$ is a joint probability mass function (PMF) with marginals $P_X$ and $P_Y$, 
and the minimizations are over PMFs $Q_X$ and $R_Y$. 
Based on the R\'enyi divergence $D_\alpha$, these expressions induce the following types of R\'enyi mutual information (RMI) for $\alpha\in [0,\infty]$.
\begin{align}
I_\alpha^{\uparrow\uparrow}(X:Y)_P&\coloneqq D_\alpha (P_{XY}\| P_X P_Y)
\label{eq:i-classical-0}
\\
I_\alpha^{\uparrow\downarrow}(X:Y)_P&\coloneqq \inf_{R_Y}D_\alpha (P_{XY}\| P_X R_Y)
\label{eq:i-classical-1}
\\
I_\alpha^{\downarrow\downarrow}(X:Y)_P&\coloneqq \inf_{Q_X,R_Y}D_\alpha(P_{XY}\| Q_X R_Y)
\label{eq:i-classical-2}
\end{align}
We call them the \emph{non-minimized RMI}, the \emph{singly minimized RMI}, and the \emph{doubly minimized RMI}, respectively. 
For $\alpha=1$, all three RMIs are identical to the classical mutual information~\cite{lapidoth2019two}. 
The singly minimized RMI was originally introduced in~\cite{sibson1969information} and its properties were subsequently studied in~\cite{csiszar1995generalized,ho2015convexity,verdu2015alpha,verdu2021error,esposito2022sibsons,esposito2024sibsons}. 
The doubly minimized RMI was introduced in~\cite{tomamichel2018operational,lapidoth2019two}. 
All three RMIs possess an operational interpretation in binary quantum state discrimination. 
More specifically, the direct exponents of certain binary discrimination problems are determined by the non-minimized RMIs~\cite{hoeffding1965asymptotically,hoeffding1965probabilities,csiszar1971error,audenaert2008asymptotic,blahut1974hypothesis} of order $\alpha\in [0,1]$, 
the singly minimized RMIs of order $\alpha\in [0,1]$~\cite{tomamichel2018operational}, 
and the doubly minimized RMIs of order $\alpha\in [\frac{1}{2},1]$~\cite{tomamichel2018operational}, respectively. 
Similarly, the strong converse exponents for these discrimination problems are determined by the non-minimized RMIs~\cite{blahut1974hypothesis,han1989strong,nakagawa1993converse}, the singly minimized RMIs~\cite{tomamichel2018operational}, and the doubly minimized RMIs~\cite{tomamichel2018operational} of order $\alpha\in [1,\infty]$, respectively.

\paragraph*{Quantum mutual information.} Analogous to~\eqref{eq:shannon-info}, the quantum mutual information can be expressed in terms of the quantum relative entropy $D$ in several ways~\cite{gupta2014multiplicativity,hayashi2016correlation,mckinlay2020decomposition}:
\begin{equation}
I(A:B)_\rho
=D (\rho_{AB}\| \rho_A\otimes \rho_B)
=\inf_{\tau_B}D (\rho_{AB}\| \rho_A\otimes \tau_B)
=\inf_{\sigma_A,\tau_B} D (\rho_{AB}\| \sigma_A\otimes \tau_B).
\end{equation}
Here, $\rho_{AB}$ is a bipartite quantum state on finite-dimensional Hilbert spaces $A$ and $B$ with marginal states $\rho_A$ and $\rho_B$, and the minimizations are over quantum states $\sigma_A$ and $\tau_B$. 
From the perspective of quantum information theory, it is desirable to extend the RMIs in~\eqref{eq:i-classical-0}--\eqref{eq:i-classical-2} from the classical to the quantum domain. 
However, there exist several quantum generalizations of the R\'enyi divergence. 
Commonly used are, for example, the Petz divergence~\cite{petz1986quasi} and the sandwiched divergence~\cite{wilde2014strong,mueller2013quantum}. 
Consequently, various generalizations of the classical RMIs to the quantum setting are conceivable and discerning which of these generalizations are operationally relevant is not straightforward. 

This paper examines generalizations based on the Petz and the sandwiched divergence, as these choices of divergence turn out to be suitable for generalizing the aforementioned results in binary quantum state discrimination from the classical to the quantum setting. 
This generalization has been established in previous work for the non-minimized~\cite{hayashi2007error,nagaoka2006converse,audenaert2008asymptotic,mosonyi2014quantum,mosonyi2015two} and the singly minimized~\cite{hayashi2016correlation} case, and we will show the corresponding result for the doubly minimized case in forthcoming work~\cite{burri2025prmisrmi2}. 

Using the Petz divergence $D_\alpha$, we define the following types of Petz R\'enyi mutual information (PRMI) for $\alpha\in [0,\infty)$.
\begin{align}
I_\alpha^{\uparrow\uparrow}(A:B)_\rho &\coloneqq D_\alpha (\rho_{AB}\| \rho_A\otimes \rho_B)
\label{eq:i-quantum-0}\\
I_\alpha^{\uparrow\downarrow}(A:B)_\rho &\coloneqq \inf_{\tau_B} D_\alpha (\rho_{AB}\| \rho_A\otimes \tau_B)
\\
I_\alpha^{\downarrow\downarrow}(A:B)_\rho &\coloneqq \inf_{\sigma_A,\tau_B} D_\alpha (\rho_{AB}\| \sigma_A\otimes \tau_B)
\label{eq:i-quantum-2}
\end{align}
We call them the \emph{non-minimized PRMI}, the \emph{singly minimized PRMI}, and the \emph{doubly minimized PRMI}, respectively. 
For $\alpha=1$, all three PRMIs are identical to the quantum mutual information~\cite{gupta2014multiplicativity,hayashi2016correlation}. 
The non-minimized PRMI has previously been examined in the context of quantum field theory~\cite{kudlerflam2023renyi,kudlerflam2023renyi1}, 
and has been applied to classical-quantum channel coding~\cite{cheng2023simple}, 
convex splitting~\cite{cheng2023tight}, 
and quantum soft covering~\cite{cheng2024error}. 
The singly minimized PRMI has been studied with regard to its general properties in~\cite{gupta2014multiplicativity,hayashi2016correlation}, 
and has appeared in a random coding bound in~\cite{cheng2025errorexponentsquantumpacking}. 
The doubly minimized PRMI has been mentioned in previous work on binary quantum state discrimination~\cite{berta2021composite} and has been investigated for the case $\alpha=0$ in~\cite{zhai2023chain}. 
However, none of these works has studied its general properties in detail. 

Similarly, using the sandwiched divergence $\widetilde{D}_\alpha$, we define the following types of sandwiched R\'enyi mutual information (SRMI) for $\alpha\in (0,\infty]$.
\begin{align}
\widetilde{I}_\alpha^{\uparrow\uparrow}(A:B)_\rho &\coloneqq \widetilde{D}_\alpha (\rho_{AB}\| \rho_A\otimes \rho_B)
\label{eq:i-quantum-00}\\
\widetilde{I}_\alpha^{\uparrow\downarrow}(A:B)_\rho &\coloneqq \inf_{\tau_B} \widetilde{D}_\alpha (\rho_{AB}\| \rho_A\otimes \tau_B)
\\
\widetilde{I}_\alpha^{\downarrow\downarrow}(A:B)_\rho &\coloneqq \inf_{\sigma_A,\tau_B} \widetilde{D}_\alpha (\rho_{AB}\| \sigma_A\otimes \tau_B)
\label{eq:i-quantum-22}
\end{align}
We call them the \emph{non-minimized SRMI}, 
the \emph{singly minimized SRMI}, 
and the \emph{doubly minimized SRMI}, respectively. 
For $\alpha=1$, all three SRMIs are identical to the quantum mutual information~\cite{gupta2014multiplicativity,hayashi2016correlation}. 
For $\alpha=\infty$, (smoothed versions of) these SRMIs have been examined in~\cite{berta2011quantum,datta2013oneshot,ciganovic2014smooth,berta2014identifying,anshu2017quantum,scalet2021computablerenyi}.  
For general R\'enyi order $\alpha$, 
the singly minimized SRMI has been first studied in~\cite{beigi2013sandwiched} 
and has been applied in the context of 
state redistribution~\cite{leditzky2016strong}, 
binary quantum state discrimination~\cite{hayashi2016correlation},
channel coding~\cite{gupta2014multiplicativity,mosonyi2015coding,mosonyi2017strong}, 
quantum information decoupling~\cite{li2022reliability},  
convex splitting~\cite{cheng2023quantum}, and 
quantum soft covering~\cite{cheng2024error}. 
The doubly minimized SRMI has been 
examined with regard to its relation to conditional entropies~\cite{mckinlay2020decomposition} and continuity bounds~\cite{bluhm2023unified}, 
and has found applications in 
quantum information decoupling~\cite{li2024operational} and convex splitting~\cite{cheng2023tight}. 
Of the aforementioned works,~\cite{cheng2023tight} established general properties of the doubly minimized SRMI.
In particular,~\cite{cheng2023tight} proved the additivity of the doubly minimized SRMI for $\alpha\in (1,\infty)$.

\subsection{Overview of results}
In this paper, we initiate the study of properties of the doubly minimized PRMI and we deepen the study of properties of the doubly minimized SRMI. 
Our main findings for the doubly minimized PRMI are stated in Theorems~\ref{thm:convexity} and~\ref{thm:prmi2}, and those for the doubly minimized SRMI appear in Theorem~\ref{thm:srmi2}. 
We will now give an overview of these results. 

\emph{Properties of the doubly minimized PRMI.} 
In Theorem~\ref{thm:convexity}, we show that the minimization problem which defines the doubly minimized PRMI, see~\eqref{eq:i-quantum-2}, is jointly convex in $\sigma_A$ and $\tau_B$ for any $\alpha\in [\frac{1}{2},1)$. 
A crucial component of the proof of this theorem is an operator inequality (Lemma~\ref{lem:cs}), which follows from the subadditivity of the geometric operator mean.

In Theorem~\ref{thm:prmi2}~(a)--(v), we enumerate several properties pertaining to the doubly minimized PRMI. 
Of these, the following are particularly important:
\begin{enumerate}
\item[(d)] additivity for $\alpha\in [\frac{1}{2},2]$,
\item[(j)] uniqueness of the minimizer for $\alpha\in (\frac{1}{2},1]$,
\item[(k)] a fixed-point property of minimizers for $\alpha\in (\frac{1}{2},2]$, 
\item[(l)] asymptotic optimality of the universal permutation invariant state for $\alpha\in [0,2]$,  
\item[(p)] continuous differentiability in $\alpha$ of $I_\alpha^{\downarrow\downarrow}(A:B)_\rho$ on $\alpha\in (\frac{1}{2},2)$, and 
\item[(q)] convexity in $\alpha$ of $(\alpha-1)I_\alpha^{\downarrow\downarrow}(A:B)_\rho$ on $\alpha\in [0,2]$.
\end{enumerate}

Of these properties, we will prove~(j) first. 
According to~(j), the minimizer of the optimization problem in~\eqref{eq:i-quantum-2}, which defines the doubly minimized PRMI, is unique. 
That is, for any $\alpha\in (\frac{1}{2},1]$, there exists a unique pair of quantum states $(\sigma_A,\tau_B)$ such that 
$D_\alpha(\rho_{AB}\| \sigma_A\otimes \tau_B)=I_\alpha^{\downarrow\downarrow}(A:B)_\rho$. 
Previously, the uniqueness of the minimizer of the classical optimization problem that defines the doubly minimized RMI has been established in~\cite[Lemma~20]{lapidoth2019two} for any $\alpha\in (\frac{1}{2},\infty)$. 
Similar to their proof for the classical setting, we show that the uniqueness of the minimizer follows from the joint convexity of the optimization problem (Theorem~\ref{thm:convexity}). 

Then, we will prove~(k). 
The item~(k) characterizes minimizers $\sigma_A$ for the optimization problem~\eqref{eq:i-quantum-2} in terms of a fixed-point property. 
The proof of~(k) is based on an extension of a lemma from previous work~\cite[Lemma~22]{hayashi2016correlation} that asserts a general equivalence of optimizers and fixed-points (Lemma~\ref{lem:fixed_point}). 

Additivity on product states~(d) is then a direct consequence of~(k). 
The additivity of the doubly minimized PRMI has been previously suggested in~\cite[Section~3.3]{berta2021composite}, but no detailed proof was given. 

The central assertion in~(l) is that 
\begin{align}
I_\alpha^{\downarrow\downarrow}(A:B)_\rho=\lim_{n\rightarrow\infty}\frac{1}{n} D_\alpha(\rho_{AB}^{\otimes n}\| \omega_{A^n}^n\otimes\omega_{B^n}^n), 
\end{align}
where $\omega_{A^n}^n$ and $\omega_{B^n}^n$ denote certain universal permutation invariant states that are independent of $\alpha$ and will be defined in Section~\ref{ssec:permutation}. 
This equality shows that the minimization over product states on $AB$, which is present in the definition of $I_\alpha^{\downarrow\downarrow}(A:B)_\rho$ in~\eqref{eq:i-quantum-2}, can be circumvented by comparing multiple copies of $\rho_{AB}$ to the tensor product of two universal permutation invariant states. 
Thus,~(l) offers a qualitatively different perspective on the doubly minimized PRMI, which can be useful in applications. 
The proof of~(l) for $\alpha\in [\frac{1}{2},2]$ makes use of additivity~(d) and proceeds in a similar manner to an analogous proof~\cite{hayashi2016correlation} for the singly minimized SRMI. 
In the case where $\alpha\in [0,\frac{1}{2})$, a different proof method is employed that does not require additivity. 

The convexity property~(q) follows directly from~(l). 
Subsequently, we prove~(p) using~(j) and~(q). 

\emph{Properties of the doubly minimized SRMI.} 
In Theorem~\ref{thm:srmi2}~(a)--(s), we enumerate several properties of the doubly minimized SRMI. 
Of these, the following are particularly important: 
\begin{enumerate}
\item[(d)] additivity for $\alpha\in [\frac{2}{3},\infty]$,
\item[(e)] a novel duality relation for $\alpha\in [\frac{2}{3},\infty]$,
\item[(j)] asymptotic optimality of the universal permutation invariant state for $\alpha\in [\frac{2}{3},\infty]$,
\item[(n)] continuous differentiability in $\alpha$ of $\widetilde{I}_\alpha^{\downarrow\downarrow}(A:B)_\rho$ on $\alpha\in (1,\infty)$, and
\item[(o)] convexity in $\alpha$ of $(\alpha -1)\widetilde{I}_\alpha^{\downarrow\downarrow}(A:B)_\rho$ on $\alpha\in [\frac{2}{3},\infty)$.
\end{enumerate}

Of these properties, we prove the duality relation~(e) first by employing Sion's minimax theorem. 
We then show that additivity~(d) follows directly from~(e). 
Previously, a proof of additivity for $\alpha\in (1,\infty)$ has been given in~\cite{cheng2023tight}, and it has been conjectured~\cite{li2024operational} that additivity might hold for all $\alpha\in [\frac{1}{2},\infty]$. 
Our result in~(d) shows that this conjecture is at least partially true because~(d) extends the known range of additivity to $\alpha\in [\frac{2}{3},\infty]$. 
This extension is of interest for recent work on quantum information decoupling~\cite{li2024operational}, where the regularized doubly minimized SRMI of order $\alpha\in (\frac{1}{2},1)$ occurs. 
Our additivity result implies that the regularization can be omitted if $\alpha\in[\frac{2}{3},1)$, leading to a much simpler expression. 
It is worth noting that our proof of additivity for $\alpha\in [\frac{2}{3},\infty]$ is independent of the proof of additivity for $\alpha\in (1,\infty)$ in~\cite{cheng2023tight} as different proof methods are employed. 
Further clarification on the difference in the proof methods will be provided in Remark~\ref{rem:previous}.

The main assertion in~(j) is that for any $\alpha\in [\frac{2}{3},\infty]$
\begin{equation}\label{eq:j1}
\widetilde{I}_\alpha^{\downarrow\downarrow}(A:B)_\rho
= \lim_{n\rightarrow\infty}\frac{1}{n}\widetilde{D}_\alpha (\rho_{AB}^{\otimes n}\| \omega_{A^n}^n\otimes \omega_{B^n}^n)
= \lim_{n\rightarrow\infty}\frac{1}{n}D_\alpha (\mathcal{P}_{\omega_{A^n}^n\otimes \omega_{B^n}^n}(\rho_{AB}^{\otimes n})\| \omega_{A^n}^n\otimes \omega_{B^n}^n),
\end{equation}
where $\mathcal{P}$ denotes the pinching map. 
The first equality in~\eqref{eq:j1} shows that the universal permutation invariant states are asymptotically optimal for the minimization problem that defines the doubly minimized SRMI. 
According to the second equality in~\eqref{eq:j1}, the doubly minimized SRMI is asymptotically attainable by pinching with respect to the tensor product of two universal permutation invariant states. 
From a qualitative point of view,~\eqref{eq:j1} transforms the minimization problem inherent in the definition of the doubly minimized SRMI into an asymptotic limit, which can be useful for applications. 
The proof of~\eqref{eq:j1} is based on additivity~(d) and proceeds in a similar manner to the proof of an analogous assertion for the singly minimized SRMI in~\cite{hayashi2016correlation}.
The result in~\eqref{eq:j1} from~(j) directly implies the convexity property~(o), which in turn leads to the continuous differentiability property~(n). 

\paragraph*{Related work.} 
As we will show in a forthcoming paper~\cite{burri2025prmisrmi2}, both types of R\'enyi mutual information considered here turn out to be operationally meaningful. 
In particular, we show in~\cite{burri2025prmisrmi2} that the doubly minimized PRMI of order $\alpha\in [\frac{1}{2},1]$ determines the direct exponent, while the doubly minimized SRMI of order $\alpha\in [1,\infty]$ determines the strong converse exponent of a certain binary quantum state discrimination problem. 
The proofs of these statements crucially rely on several of the properties established in the present work.

\paragraph*{Outline.} In Section~\ref{sec:preliminaries} we collect some mathematical preliminaries. 
We first explain our general notation~(\ref{ssec:notation}). 
Subsequently, we state some definitions and properties related to 
permutation invariance~(\ref{ssec:permutation}), 
entropies and divergences~(\ref{ssec:divergence}), 
and several types of R\'enyi mutual information~(\ref{ssec:prmi}). 
In Section~\ref{sec:main} we present our main results (Theorems~\ref{thm:convexity},~\ref{thm:prmi2},~\ref{thm:srmi2}).

\section{Preliminaries}\label{sec:preliminaries}

\subsection{Notation}\label{ssec:notation}
We take ``$\log$'' to refer to the natural logarithm.
The set of natural numbers strictly less than $n\in \mathbb{N}$ is denoted by $[n]\coloneqq \{0,1,\dots, n-1\}$. 

Throughout this work, we restrict ourselves to finite-dimensional Hilbert spaces (over the field $\mathbb{C}$) for simplicity. 
The dimension of a Hilbert space $A$ is denoted by $d_A\coloneqq \dim(A)\in \mathbb{N}_{>0}$. 
The tensor product of two Hilbert spaces $A$ and $B$ is sometimes denoted by $AB$ instead of $A\otimes B$, and $A^{n}\coloneqq A^{\otimes n}$ for any $n\in \mathbb{N}_{>0}$. 
The set of linear maps from $A$ to $B$ is denoted by $\mathcal{L}(A,B)$, and the set of linear maps from $A$ to itself is denoted by $\mathcal{L}(A)\coloneqq \mathcal{L}(A,A)$. 
To facilitate concise notation, identities are occasionally omitted. 
For instance, for $X_A\in\mathcal{L}(A)$, ``$X_A$'' may be interpreted as $X_A\otimes 1_B\in \mathcal{L}(A\otimes B)$. 
The spectrum, kernel, and rank of $X\in \mathcal{L}(A)$ are denoted as $\spec(X),\ker(X),$ and $\rank(X)$, respectively. 
The support of $X\in \mathcal{L}(A)$ is defined as the orthogonal complement of the kernel of $X$, and is denoted by $\supp(X)$.
For $X,Y\in \mathcal{L}(A)$, $X\ll Y$ is true iff $\ker(Y)\subseteq \ker(X)$, and $X\not\ll Y$ is true iff $X\ll Y$ is false. 
For $X,Y\in \mathcal{L}(A)$, $X\perp Y$ is true iff $XY=YX=0$, and $X\not\perp Y$ is true iff $X\perp Y$ is false. 

The adjoint of $X\in \mathcal{L}(A,B)$ with respect to the inner products of $A$ and $B$ is denoted by $X^\dagger \in \mathcal{L}(B,A)$. 
For $X\in \mathcal{L}(A)$, $X\geq 0$ is true iff $X$ is positive semidefinite, 
and $X>0$ is true iff $X$ is positive definite. 
For two self-adjoint operators $X,Y\in \mathcal{L}(A)$, $X\geq Y$ is true iff $X-Y\geq 0$. 
For a positive semidefinite $X\in \mathcal{L}(A)$, $X^p$ is defined for $p\in \mathbb{R}$ by taking the power on the support of $X$. 
In the case where $p=1/2$, the square root symbol is sometimes employed, $\sqrt{X}\coloneqq X^{1/2}$. 
The operator absolute value of $X\in \mathcal{L}(A)$ is defined as $\lvert X\rvert \coloneqq (X^\dagger X)^{1/2}$.
The trace of $X\in \mathcal{L}(A)$ is denoted as $\tr[X]$, and the partial trace over $A$ is denoted as $\tr_A$. 
The Schatten $p$-norm of $X\in \mathcal{L}(A)$ is defined as 
$\|X \|_p\coloneqq \tr[\lvert X\rvert^p]^{1/p}$ for $p\in [1,\infty)$, and as 
$\lVert X\rVert_{\infty}\coloneqq \sqrt{\max (\spec(X^\dagger X))}$ for $p=\infty$. 
The Schatten $p$-quasi-norm is defined as $\|X \|_p\coloneqq \tr[\lvert X\rvert^p]^{1/p}$ for $p\in (0,1)$.

If $X\in \mathcal{L}(A)$ is self-adjoint, then 
the pinching map with respect to $X$ is denoted by 
$\mathcal{P}_{X}:\mathcal{L}(A)\rightarrow\mathcal{L}(A), Y\mapsto \sum_{\lambda\in \spec(X)}P_{\lambda}Y P_{\lambda}$, where $P_\lambda$ denotes the orthogonal projection onto the eigenspace associated with $\lambda$.

If $X,Y\in \mathcal{L}(A)$ are positive definite, then the geometric operator mean~\cite{pusz1975functional,ando1987some,kubo1980means} is defined as 
$X\#  Y\coloneqq X^{\frac{1}{2}}(X^{-\frac{1}{2}}YX^{-\frac{1}{2}})^{\frac{1}{2}}X^{\frac{1}{2}}$.
This definition is extended to the case where $X,Y\in \mathcal{L}(A)$ are positive semidefinite by setting 
$X\#  Y\coloneqq \lim_{\varepsilon\rightarrow 0}(X+\varepsilon 1)\#  (Y+\varepsilon 1)$, where ``$\lim$'' denotes the limit in the strong operator topology.
Since we are only working with finite-dimensional Hilbert spaces, the strong operator topology coincides with the weak operator topology and with the norm topology induced by the operator norm. 
``$\lim$'' can therefore be interpreted as an operator limit with respect to any of these topologies. 
The geometric operator mean is subadditive~\cite{kubo1980means}, i.e., 
$(X\# Y)+(X'\# Y')\leq (X+X')\# (Y+Y')$ 
for all positive semidefinite operators $X,X',Y,Y'\in \mathcal{L}(A)$.

The set of unitary operators on $A$ is denoted by $\mathcal{U}(A)\subseteq \mathcal{L}(A)$. 
The set of quantum states on $A$ is $\mathcal{S}(A)\coloneqq\{\rho\in \mathcal{L}(A):\rho\geq 0,\tr[\rho]=1\}$. 
The set of positive definite quantum states on $A$ is denoted by 
$\mathcal{S}_{>0}(A)\coloneqq \{\rho\in \mathcal{S}(A): \rho>0\}$. 
Moreover, we define the following constrained versions of $\mathcal{S}(A)$ with respect to a self-adjoint $X\in\mathcal{L}(A)$.
\begin{align}
\mathcal{S}_{\not\perp X}(A)
&\coloneqq \{\rho\in \mathcal{S}(A): \rho\not\perp X\}
\\
\mathcal{S}_{\ll X}(A)
&\coloneqq \{\rho\in \mathcal{S}(A): \rho\ll X\}
\\
\mathcal{S}_{\sim X}(A)
&\coloneqq\{\rho\in \mathcal{S}(A): \rho\ll X, X\ll \rho\}
\end{align}
The set of completely positive, trace-preserving linear maps from $\mathcal{L}(A)$ to $\mathcal{L}(B)$ is denoted by $\CPTP(A,B)$.

\subsection{Permutation invariance}\label{ssec:permutation}
The \emph{symmetric group of degree $n\in \mathbb{N}_{>0}$} is denoted by $S_n$.
The unitary operator $U(\pi)_{A^n}\in \mathcal{U}(A^{\otimes n})$ associated with $\pi \in S_n$ is defined by the requirement that
\begin{equation}
U(\pi)_{A^n}\ket{\psi_1}\otimes\dots\otimes\ket{\psi_n}
= \ket{\psi_{\pi^{-1}(1)}}\otimes\dots\otimes\ket{\psi_{\pi^{-1}(n)}}
\qquad\forall \ket{\psi_1},\dots,\ket{\psi_n}\in A.
\end{equation}
The set of \emph{permutation invariant states} is 
\begin{align}
\mathcal{S}_{\sym}(A^n)
\coloneqq \{\rho_{A^n}\in \mathcal{S}(A^{\otimes n}):
U(\pi)_{A^n} \rho_{A^n} U(\pi)_{A^n}^\dagger = \rho_{A^n}\,\forall \pi \in S_n\}.
\end{align}

We will make use of the following construction of a universal permutation invariant state. 
Let $A$ and $A'$ be isomorphic Hilbert spaces of dimension $d_A\in \mathbb{N}_{>0}$. 
The \emph{universal permutation invariant state} is defined for any $n\in \mathbb{N}_{>0}$ as~\cite{renner2006security,christandl2009postselection,hayashi2016correlation}
\begin{equation}
\omega_{A^n}^n \coloneqq \frac{1}{g_{n,d_A}}\tr_{{A'}^n}[(P^{n}_{\sym})_{A^n{A'}^n}],
\qquad\text{where}\qquad
g_{n,d_A}\coloneqq \binom{n+d_A^2-1}{n},
\end{equation}
and $(P^{n}_{\sym})_{A^n{A'}^n}$ denotes the orthogonal projection onto the symmetric subspace of $(AA')^{\otimes n}$. 
The following properties of $\omega_{A^n}^n$ have been established in previous work.

\begin{rem}[Universal permutation invariant state]\label{rem:universal-state}~\cite{renner2006security,christandl2009postselection,hayashi2016correlation}
Let $n\in \mathbb{N}_{>0}$.
Then all of the following hold.
\begin{enumerate}[label=(\alph*)]
\item $\omega_{A^n}^n\in \mathcal{S}_{\sym}(A^{\otimes n})$.
\item $\sigma_{A^n}\leq g_{n,d_A}\omega_{A^n}^n$  for all $\sigma_{A^n}\in \mathcal{S}_{\sym}(A^{\otimes n})$, and $1\leq g_{n,d_A}\leq (n+1)^{d_A^2-1}$. 

As a consequence, $\lim_{n\rightarrow\infty}\frac{1}{n^p}\log g_{n,d_A} = 0$ for any $p\in (0,\infty)$.
\item $\sigma_{A^n}\omega_{A^n}^n=\omega_{A^n}^n\sigma_{A^n}$ 
for all $\sigma_{A^n}\in \mathcal{S}_{\sym}(A^{\otimes n})$.
\item Let $\ket{0}_{AA'}\in AA'$ be an arbitrary but fixed unit vector.
Then
\begin{align}
(P_{\sym}^{n})_{A^n{A'}^n}
=g_{n,d_A}\int_{\mathcal{U}(AA')} \mathrm{d} \mu_H(U)\, (\proj{\sigma(U)}_{AA'})^{\otimes n},
\end{align}
where $\mathrm{d} \mu_H$ denotes the Haar measure on the unitary group $\mathcal{U}(AA')$
normalized so that $\int_{\mathcal{U}(AA')}\mathrm{d}\mu_H(U)=1$,  
and $\ket{\sigma(U)}_{AA'}\coloneqq U\ket{0}_{AA'}$.
As a consequence, 
\begin{align}
\omega_{A^n}^n
=\int_{\mathcal{U}(AA')} \mathrm{d} \mu_H(U)\, \sigma(U)_{A}^{\otimes n},
\end{align}
where $\sigma(U)_A\coloneqq \tr_{A'}[\proj{\sigma(U)}_{AA'}]$.
\item $\lvert \spec(\omega_{A^n}^n) \rvert\leq (n+1)^{d_A-1}$.
\end{enumerate}
\end{rem}

\subsection{Entropies and divergences}\label{ssec:divergence}
The \emph{von Neumann entropy} of $\rho\in \mathcal{S}(A)$ is defined as $H(A)_\rho\coloneqq -\tr[\rho\log \rho]$.
For $\rho\in \mathcal{S}(AB)$, 
the \emph{conditional entropy} of $A$ given $B$ is $H(A|B)_\rho \coloneqq H(AB)_\rho -H(B)_\rho$ and
the \emph{mutual information} between $A$ and $B$ is $I(A:B)_\rho \coloneqq H(A)_\rho+H(B)_\rho-H(AB)_\rho$.
The \emph{R\'enyi entropy (of order $\alpha$)} of $\rho\in \mathcal{S}(A)$ is defined as
$H_\alpha (A)_\rho
\coloneqq\frac{1}{1-\alpha}\log \tr[\rho^\alpha]$
for $\alpha\in (-\infty,1)\cup (1,\infty)$, 
and for $\alpha \in \{1,\infty\}$ as the corresponding limits. 

The \emph{quantum relative entropy} is defined for $\rho\in \mathcal{S}(A)$ and a positive semidefinite $\sigma\in \mathcal{L}(A)$ as
\begin{align}
D(\rho\| \sigma)\coloneqq
\tr[\rho (\log\rho -\log\sigma)]
\end{align}
if $\rho\ll \sigma$ and $D(\rho\| \sigma)\coloneqq \infty$ else.
The mutual information of $\rho_{AB}\in \mathcal{S}(AB)$ can be expressed in terms of the quantum relative entropy in the following ways~\cite{gupta2014multiplicativity,hayashi2016correlation,mckinlay2020decomposition}.
\begin{equation}\label{eq:i1}
I(A:B)_\rho
=D (\rho_{AB}\| \rho_A\otimes \rho_B)
=\inf_{\tau_B\in \mathcal{S}(B)}D (\rho_{AB}\| \rho_A\otimes \tau_B)
=\inf_{\substack{\sigma_A\in \mathcal{S}(A),\\ \tau_B\in \mathcal{S}(B)}}D (\rho_{AB}\| \sigma_A\otimes \tau_B)
\end{equation}
The latter two equalities can be derived based on the non-negativity of the quantum relative entropy~\cite{mckinlay2020decomposition}. 
This argument implies that the corresponding minimizers are uniquely given by $\rho_A$ and $\rho_B$, i.e.,
\begin{align}
\argmin_{\tau_B\in \mathcal{S}(B)}
D(\rho_{AB}\| \rho_A\otimes \tau_B)
&=\{\rho_B\},
\label{eq:i1-argmin-b}\\
\argmin_{(\sigma_A,\tau_B)\in \mathcal{S}(A)\times\mathcal{S}(B)}
D(\rho_{AB}\| \sigma_A\otimes \tau_B)
&=\{(\rho_A,\rho_B)\}.
\label{eq:i1-argmin}
\end{align}
More generally, the argument based on the non-negativity of the quantum relative entropy implies that for any $\sigma_A\in \mathcal{S}(A)$ such that $\rho_A\ll \sigma_A$ 
\begin{equation}\label{eq:min-re-t}
\argmin_{\tau_B\in \mathcal{S}(B)}
D(\rho_{AB}\| \sigma_A\otimes \tau_B)
=\{\rho_B\}.
\end{equation}

The \emph{(quantum) information variance} is defined for $\rho,\sigma\in \mathcal{S}(A)$ as~\cite{tomamichel2013hierarchy,li2014second}
\begin{equation}
V(\rho\| \sigma)\coloneqq \tr[\rho (\log \rho - \log \sigma - D(\rho\| \sigma))^2] 
=\tr[\rho(\log \rho-\log \sigma)^2]-(D(\rho\| \sigma))^2.
\end{equation}
The~\emph{mutual information variance} of $\rho_{AB}\in \mathcal{S}(AB)$ is defined as~\cite{hayashi2016correlation} 
\begin{align}\label{eq:def-variance}
V(A:B)_\rho
&\coloneqq V(\rho_{AB}\| \rho_A\otimes \rho_B)
=\tr[\rho_{AB}(\log \rho_{AB}-\log (\rho_A\otimes\rho_B) -I(A:B)_\rho )^2].
\end{align}

The \emph{Petz (quantum R\'enyi) divergence (of order $\alpha$)} is defined for $\alpha\in (0,1)\cup (1,\infty),\rho\in \mathcal{S}(A)$, and any positive semidefinite $\sigma\in \mathcal{L}(A)$ as~\cite{petz1986quasi}
\begin{equation}
D_\alpha (\rho\| \sigma)\coloneqq\frac{1}{\alpha -1} \log \tr [\rho^\alpha \sigma^{1-\alpha}]
\end{equation}
if $(\alpha <1\land \rho\not\perp\sigma)\lor \rho\ll \sigma$ and 
$D_\alpha (\rho\| \sigma)\coloneqq \infty$ else. 
Moreover, $D_0$ and $D_1$ are defined as the limits of $D_\alpha$ for $\alpha\rightarrow\{0,1\}$.
For $\alpha\in (-\infty,\infty)$, we define
$Q_\alpha (\rho\| \sigma)\coloneqq
\tr [\rho^\alpha \sigma^{1-\alpha}]$
for all positive semidefinite $\rho,\sigma\in \mathcal{L}(A)$. 
We will make use of the following properties of the Petz divergence.
\begin{rem}[Petz divergence]
\label{rem:petz-divergence}
Let $\rho\in \mathcal{S}(A)$ and let $\sigma\in \mathcal{L}(A)$ be positive semidefinite. 
Then all of the following hold.~\cite{petz1986quasi,ohya1993quantum,nussbaum2009chernoff,lin2015investigating,tomamichel2016quantum}
\begin{enumerate}[label=(\alph*)]
\item \emph{Data-processing inequality:}  
$D_\alpha (\rho\| \sigma)\geq D_\alpha (\mathcal{M}(\rho)\| \mathcal{M}(\sigma ))$ 
for any $\mathcal{M}\in \CPTP(A,A')$ 
and all $\alpha \in [0,2]$.
\item \emph{Invariance under isometries:}  
$D_\alpha (V\rho V^\dagger\|V \sigma V^\dagger)=D_\alpha(\rho\| \sigma)$ 
for any isometry $V\in \mathcal{L}(A,A')$ 
and all $\alpha \in [0,\infty)$.
\item \emph{Additivity:} 
Let $\rho'_B\in \mathcal{S}(B)$ and let $\sigma_B'\in \mathcal{L}(B)$ be positive semidefinite. 
Then  
$D_\alpha(\rho_A\otimes \rho_B'\| \sigma_A\otimes \sigma_B')=D_\alpha(\rho_A\| \sigma_A)+D_\alpha(\rho_B'\| \sigma_B')$ 
for all $\alpha\in [0,\infty)$.
\item \emph{Normalization:} $D_\alpha(\rho\| c\sigma)=D_\alpha(\rho\| \sigma)-\log c$
for all $\alpha\in [0,\infty),c\in (0,\infty)$.
\item \emph{Dominance:} 
If $\sigma'\in \mathcal{L}(A)$ is positive semidefinite and such that $\sigma\leq \sigma'$, then
$D_\alpha (\rho\| \sigma)\geq D_\alpha (\rho\| \sigma')$ 
for all $\alpha\in [0,2]$. 
\item \emph{Non-negativity:} 
If $\sigma\in \mathcal{S}(A)$, then
$D_\alpha(\rho\| \sigma)\in [0,\infty]$ for all 
$\alpha\in [0,\infty)$. 
\item \emph{Positive definiteness:} 
Let $\alpha\in (0,\infty)$.
If $\sigma\in \mathcal{S}(A)$, then $D_\alpha(\rho\| \sigma)=0$ iff $\rho=\sigma$.
Furthermore, $D_0(\rho\| \rho )=0$.
\item \emph{R\'enyi order $\alpha =1$:} 
$D_1 (\rho\| \sigma)=D(\rho\|\sigma )$.
\item \emph{Monotonicity in $\alpha$:} 
If $\alpha,\beta\in [0,\infty)$ are such that $\alpha\leq \beta$, then 
$D_\alpha(\rho\| \sigma)\leq D_\beta(\rho\| \sigma)$.
\item \emph{Continuity in $\alpha$:} 
If $\rho\not\perp\sigma$, then the function $[0,1)\rightarrow\mathbb{R},\alpha\mapsto D_\alpha(\rho\| \sigma)$ is continuous.
If $\rho\ll\sigma$, then the function $[0,\infty)\rightarrow\mathbb{R},\alpha\mapsto D_\alpha(\rho\| \sigma)$ is continuous.
\item \emph{Differentiability in $\alpha$:} 
If $\rho\not\perp\sigma$, then the function 
$(0,1)\rightarrow\mathbb{R},\alpha\mapsto D_\alpha(\rho\| \sigma)$ is continuously differentiable.
If $\rho\ll\sigma$, then the function 
$(0,\infty)\rightarrow\mathbb{R},\alpha\mapsto D_\alpha(\rho\| \sigma)$ is continuously differentiable.
Furthermore, if $\rho\ll \sigma$ and $\sigma\in \mathcal{S}(A)$, then 
$\frac{\mathrm{d}}{\mathrm{d} \alpha}D_\alpha (\rho\| \sigma)|_{\alpha=1}=\frac{1}{2} V(\rho\| \sigma)$.
\item \emph{Convexity in $\alpha$:} 
If $\rho\not\perp\sigma$, then the function 
$[0,1)\rightarrow\mathbb{R},\alpha\mapsto (\alpha -1)D_\alpha(\rho\| \sigma)$ is convex. 
If $\rho\ll\sigma$, then the function 
$[0,\infty)\rightarrow\mathbb{R},\alpha\mapsto (\alpha -1)D_\alpha(\rho\| \sigma)$ is convex. 
\end{enumerate}
\end{rem}

The \emph{(quantum) max-divergence} of $\rho\in \mathcal{S}(A)$ relative to a positive semidefinite $\sigma\in \mathcal{L}(A)$ is defined as~\cite{datta2009min}
\begin{align}
D_{\max}(\rho\| \sigma)\coloneqq \inf\{\lambda\in\mathbb{R} : \rho\leq \exp(\lambda)\sigma\}
\end{align}
with the convention that $\inf\emptyset=\infty$.

The \emph{sandwiched (quantum R\'enyi) divergence (of order $\alpha$)} is defined for $\alpha\in (0,1)\cup (1,\infty),\rho\in \mathcal{S}(A)$, and any positive semidefinite $\sigma\in \mathcal{L}(A)$ as~\cite{mueller2013quantum,wilde2014strong}
\begin{equation}
\widetilde{D}_\alpha (\rho\| \sigma)\coloneqq
\frac{1}{\alpha -1} \log \tr [(\sigma^{\frac{1-\alpha}{2\alpha}} \rho \sigma^{\frac{1-\alpha}{2\alpha}})^\alpha]
\end{equation}
if $(\alpha <1\land \rho\not\perp\sigma)\lor \rho\ll \sigma$ and 
$\widetilde{D}_\alpha (\rho\| \sigma)\coloneqq\infty$ else. 
Moreover, $\widetilde{D}_1$ and $\widetilde{D}_\infty$ are defined as the limits of $\widetilde{D}_\alpha$ for $\alpha\rightarrow\{1,\infty\}$. 
For $\alpha\in (0,\infty)$, we define
$\widetilde{Q}_\alpha (\rho\| \sigma)\coloneqq
\tr [(\sigma^{\frac{1-\alpha}{2\alpha}} \rho \sigma^{\frac{1-\alpha}{2\alpha}})^\alpha]$
for all positive semidefinite $\rho,\sigma\in \mathcal{L}(A)$.

\begin{rem}[Sandwiched divergence]
\label{rem:sandwiched-divergence}
~\cite{mueller2013quantum,frank2013monotonicity,beigi2013sandwiched,mosonyi2014quantum,lin2015investigating,hayashi2016correlation,mosonyi2017strong} 
Let $\rho\in \mathcal{S}(A)$ and let $\sigma\in \mathcal{L}(A)$ be positive semidefinite. 
Then all of the following hold.
\begin{enumerate}[label=(\alph*)]
\item \emph{Data-processing inequality:}
$\widetilde{D}_\alpha (\rho\| \sigma)\geq \widetilde{D}_\alpha (\mathcal{M}(\rho)\| \mathcal{M}(\sigma ))$ 
for any $\mathcal{M}\in \CPTP(A,A')$ 
and all $\alpha \in [\frac{1}{2},\infty]$.
\item \emph{Invariance under isometries:} 
$\widetilde{D}_\alpha (V\rho V^\dagger\|V \sigma V^\dagger)=\widetilde{D}_\alpha(\rho\| \sigma)$ 
for any isometry $V\in \mathcal{L}(A,A')$
and all $\alpha \in (0,\infty]$.
\item \emph{Additivity:} 
Let $\rho'_B\in \mathcal{S}(B)$ and let $\sigma_B'\in \mathcal{L}(B)$ be positive semidefinite. 
Then  
$\widetilde{D}_\alpha(\rho_A\otimes \rho_B'\| \sigma_A\otimes \sigma_B')=\widetilde{D}_\alpha(\rho_A\| \sigma_A)+\widetilde{D}_\alpha(\rho_B'\| \sigma_B')$ 
for all $\alpha\in (0,\infty]$.
\item \emph{Normalization:} 
$\widetilde{D}_\alpha(\rho\| c\sigma)=\widetilde{D}_\alpha(\rho\| \sigma)-\log c$
for all $\alpha\in (0,\infty], c\in (0,\infty)$.
\item \emph{Dominance:} 
If $\sigma'\in \mathcal{L}(A)$ is positive semidefinite and such that $\sigma\leq \sigma'$, then
$\widetilde{D}_\alpha (\rho\| \sigma)\geq \widetilde{D}_\alpha (\rho\| \sigma')$ 
for all $\alpha\in [\frac{1}{2},\infty]$. 
\item \emph{Non-negativity:} 
If $\sigma\in \mathcal{S}(A)$, then 
$\widetilde{D}_\alpha(\rho\| \sigma)\in [0,\infty]$ for all $\alpha\in (0,\infty]$.
\item \emph{Positive definiteness:} 
Let $\alpha\in (0,\infty]$. 
If $\sigma\in \mathcal{S}(A)$, then $\widetilde{D}_\alpha(\rho\| \sigma)=0$ iff $\rho=\sigma$.
\item \emph{R\'enyi order $\alpha \in \{1,\infty\}$:} 
$\widetilde{D}_1 (\rho\| \sigma)=D(\rho\|\sigma )$ and 
$\widetilde{D}_\infty (\rho\| \sigma)=D_{\max}(\rho\|\sigma )$.
\item \emph{Monotonicity in $\alpha$:} 
If $\alpha,\beta\in (0,\infty]$ are such that $\alpha\leq \beta$, then 
$\widetilde{D}_\alpha(\rho\| \sigma)\leq \widetilde{D}_\beta(\rho\| \sigma)$.
\item \emph{Continuity in $\alpha$:} 
If $\rho\not\perp\sigma$, then the function 
$(0,1)\rightarrow\mathbb{R},\alpha\mapsto \widetilde{D}_\alpha(\rho\| \sigma)$ is continuous.
If $\rho\ll\sigma$, then the function 
$(0,\infty)\rightarrow\mathbb{R},\alpha\mapsto \widetilde{D}_\alpha(\rho\| \sigma)$ is continuous and  
$\lim_{\alpha\rightarrow\infty}\widetilde{D}_\alpha(\rho\| \sigma)=\widetilde{D}_\infty (\rho\| \sigma)$.
\item \emph{Differentiability in $\alpha$:} 
If $\rho\ll\sigma$, then the function 
$(0,\infty)\rightarrow\mathbb{R},\alpha\mapsto \widetilde{D}_\alpha(\rho\| \sigma)$ is continuously differentiable.
Moreover, if $\rho\ll \sigma$ and $\sigma\in \mathcal{S}(A)$, then 
$\frac{\mathrm{d}}{\mathrm{d}\alpha}\widetilde{D}_\alpha (\rho\| \sigma)|_{\alpha=1}=\frac{1}{2} V(\rho\| \sigma)$.
\item \emph{Convexity in $\alpha$:} 
If $\rho\ll \sigma$, then the function 
$(0,\infty)\rightarrow \mathbb{R},\alpha\mapsto (\alpha-1)\widetilde{D}_\alpha(\rho\| \sigma)$ is convex.
\item \emph{Commuting case:} 
If $\rho\sigma=\sigma\rho$, then 
$\widetilde{D}_\alpha(\rho\| \sigma)=D_\alpha(\rho\| \sigma)$ for all $\alpha\in (0,\infty)$.
\end{enumerate}
\end{rem}

By the Araki-Lieb-Thirring inequality~\cite{lieb1991inequalities,araki1990inequality,bhatia1996matrix}, we have 
for any positive semidefinite $\rho,\sigma\in \mathcal{L}(A)$
\begin{equation}\label{eq:alt}
\widetilde{Q}_\alpha (\rho\| \sigma)\geq Q_\alpha (\rho\| \sigma)
\quad\text{if}\quad\alpha\in (0,1],
\qquad
\widetilde{Q}_\alpha (\rho\| \sigma)\leq Q_\alpha (\rho\| \sigma)
\quad\text{if}\quad\alpha\in [1,\infty).
\end{equation}

\subsection{Doubly minimized Petz and sandwiched R\'enyi mutual information}\label{ssec:prmi} 
Let $\rho_{AB}\in \mathcal{S}(AB)$. 
As mentioned in the introduction, we define the \emph{doubly minimized PRMI (of order $\alpha$)} for $\alpha\in [0,\infty)$ as
\begin{align}
I_\alpha^{\downarrow\downarrow} (A:B)_\rho 
&\coloneqq \inf_{\substack{\sigma_A\in \mathcal{S}(A),\\ \tau_B\in \mathcal{S}(B)}} D_\alpha (\rho_{AB}\| \sigma_A\otimes \tau_B ).
\label{eq:prmi2}
\end{align}
Similarly, we define the \emph{doubly minimized SRMI (of order $\alpha$)} 
for $\alpha\in (0,\infty]$ as
\begin{align}
\widetilde{I}^{\downarrow\downarrow}_\alpha (A:B)_\rho 
&\coloneqq \inf_{\substack{\sigma_A\in \mathcal{S}(A),\\ \tau_B\in \mathcal{S}(B)}}\widetilde{D}_\alpha (\rho_{AB}\| \sigma_A\otimes \tau_B ).
\label{eq:srmi2}
\end{align}

As the doubly minimized PRMI and SRMI are defined via a double minimization, it is useful to introduce notation for the corresponding quantities involving only a single minimization. 
Accordingly, we define for any positive semidefinite $\sigma_A\in \mathcal{L}(A)$
\begin{align}
I_{\alpha}^{\downarrow}(\rho_{AB}\| \sigma_A)
&\coloneqq\inf_{\tau_B\in \mathcal{S}(B)}D_{\alpha}(\rho_{AB}\| \sigma_A\otimes \tau_B) 
\qquad \forall\alpha\in [0,\infty),
\\
\widetilde{I}_{\alpha}^{\downarrow}(\rho_{AB}\| \sigma_A)
&\coloneqq\inf_{\tau_B\in \mathcal{S}(B)}\widetilde{D}_{\alpha}(\rho_{AB}\| \sigma_A\otimes \tau_B)
\qquad \forall\alpha\in (0,\infty].
\end{align}
We call them the \emph{minimized generalized PRMI (of order $\alpha$)} and 
the \emph{minimized generalized SRMI (of order $\alpha$)}, respectively. 
For completeness, we provide in Appendix~\ref{sec:gen-prmi} an overview of properties of the minimized generalized PRMI and SRMI. 
Some of these properties were established in previous work~\cite{hayashi2016correlation}, as explained in the appendix. 
These results are useful for the present paper (and will be invoked in several proofs), since the doubly minimized PRMI/SRMI result from an additional minimization of the minimized generalized PRMI/SRMI.

\section{Main results}\label{sec:main}

\subsection{Properties of the doubly minimized Petz R\'enyi mutual information}\label{ssec:main-1}

In this section, we present our findings on properties of the doubly minimized PRMI of order $\alpha$. We thereby focus on $\alpha\in [\frac{1}{2},1]$ because this will turn out to be the relevant range for the application of the doubly minimized PRMI in binary quantum state discrimination~\cite{burri2025prmisrmi2}. 
Several assertions are extended to larger ranges of $\alpha$ whenever the proof technique permits a straightforward extension.

In the following theorem (Theorem~\ref{thm:convexity}), we show that the minimization problem underlying the doubly minimized PRMI of order $\alpha$ is jointly convex in the quantum states $\sigma_A$ and $\tau_B$ for any $\alpha\in [\frac{1}{2},1)$, see~\eqref{eq:joint-convexity}. 
In~\cite{lapidoth2019two}, joint convexity has been proved for the classical case. 
More precisely, in~\cite[Lemma~15]{lapidoth2019two}, the joint convexity of 
$D_\alpha(P_{XY}\| Q_XR_Y)$ in the PMFs $Q_X$ and $R_Y$ has been proved for any $\alpha\in [\frac{1}{2},\infty)$.
Our proof for the quantum case follows the same scheme as the aforementioned proof for the classical case. 
The main difference lies in one particular proof step where the proof for the classical case employs the Cauchy-Schwarz inequality for $\mathbb{R}^2$ with the standard inner product~\cite[Eq.~(227)]{lapidoth2019two} in the form 
\begin{equation}\label{eq:cs}
\sqrt{x}\sqrt{y}+\sqrt{x'}\sqrt{y'}\leq \sqrt{x+x'}\sqrt{y+y'} \qquad \forall x,y,x',y'\in [0,\infty).
\end{equation}
In order to prove joint convexity for the quantum case, we employ an operator version of this inequality instead (Lemma~\ref{lem:cs}), which follows from the subadditivity of the geometric operator mean. 
Remarkably, Lemma~\ref{lem:cs} shows that the Cauchy-Schwarz inequality~\eqref{eq:cs} can be lifted from the positive real axis to positive semidefinite operators. 

\begin{lem}[Operator inequality from subadditivity of geometric operator mean]\label{lem:cs}
Let $X_A,X_A'\in \mathcal{L}(A), Y_B,Y_B'\in \mathcal{L}(B)$ be positive semidefinite. Then
\begin{align}
\sqrt{X_A}\otimes \sqrt{Y_B} + \sqrt{X_A'}\otimes \sqrt{Y_B'}
\leq \sqrt{X_A+X_A'}\otimes \sqrt{Y_B+Y_B'}. 
\label{eq:lem-cs}
\end{align}
\end{lem}
\begin{proof}
See Appendix~\ref{app:lem:cs}.
\end{proof}

\begin{thm}[Joint concavity/convexity]\label{thm:convexity}
Let $\alpha \in [\frac{1}{2},1)$.
Let $\rho_{AB}\in \mathcal{S}(AB),\sigma_A,\sigma_A'\in \mathcal{S}(A),\tau_B,\tau_B'\in \mathcal{S}(B)$, and let
$\lambda,\lambda' \in (0,1)$ be such that $\lambda+\lambda'=1$. Then
\begin{align}
Q_\alpha(\rho_{AB}\| (\lambda \sigma_A+\lambda'\sigma_A')\otimes (\lambda \tau_B+\lambda'\tau_B'))
&\geq \lambda Q_\alpha(\rho_{AB}\| \sigma_A\otimes \tau_B) +\lambda' Q_\alpha(\rho_{AB}\| \sigma_A'\otimes \tau_B'),
\label{eq:joint-concavity}
\\
D_\alpha(\rho_{AB}\| (\lambda \sigma_A+\lambda'\sigma_A')\otimes (\lambda \tau_B+\lambda'\tau_B'))
&\leq \lambda D_\alpha(\rho_{AB}\| \sigma_A\otimes \tau_B) +\lambda' D_\alpha(\rho_{AB}\| \sigma_A'\otimes \tau_B'),
\label{eq:joint-convexity}
\\
I_\alpha^{\downarrow}(\rho_{AB}\| \lambda \sigma_A+\lambda'\sigma_A')
&\leq \lambda I_\alpha^{\downarrow}(\rho_{AB}\| \sigma_A) +\lambda' I_\alpha^{\downarrow}(\rho_{AB}\| \sigma_A').
\label{eq:i-convexity}
\end{align}
Furthermore, if 
$\alpha\in (\frac{1}{2},1),\sigma_A,\sigma_A'\in \mathcal{S}_{\sim\rho_A}(A),\tau_B,\tau_B'\in \mathcal{S}_{\sim\rho_B}(B)$, 
and at least one of the inequalities in~\eqref{eq:joint-concavity} and~\eqref{eq:joint-convexity} holds with equality, then 
$\sigma_A=\sigma_A'$ and $\tau_B=\tau_B'$. 
If $\alpha\in (\frac{1}{2},1),\sigma_A,\sigma_A'\in \mathcal{S}_{\sim\rho_A}(A)$, 
and~\eqref{eq:i-convexity} holds with equality, then 
$\sigma_A=\sigma_A'$.
\end{thm}
\begin{proof}
See Appendix~\ref{app:convexity}.
\end{proof}

In the following theorem, we enumerate several properties of the doubly minimized PRMI.

\begin{thm}[Doubly minimized PRMI]\label{thm:prmi2}
Let $\rho_{AB}\in \mathcal{S}(AB)$.
Then all of the following hold.
\begin{enumerate}[label=(\alph*)]
\item \emph{Symmetry:} 
$I_{\alpha}^{\downarrow\downarrow}(A:B)_\rho = I_\alpha ^{\downarrow\downarrow}(B:A)_\rho$
for all $\alpha \in [0,\infty)$.
\item \emph{Non-increase under local operations:} 
$I_\alpha^{\downarrow\downarrow}(A:B)_{\rho} \geq I_\alpha^{\downarrow\downarrow}(A':B')_{\mathcal{M}\otimes\mathcal{N}(\rho)}$
for any $\mathcal{M}\in \CPTP(A,A'),\mathcal{N}\in \CPTP(B,B')$ and all $\alpha\in [0,2]$.
\item \emph{Invariance under local isometries:} 
$I_\alpha^{\downarrow\downarrow}(A':B')_{V\otimes W\rho V^\dagger\otimes W^\dagger}
=I_\alpha^{\downarrow\downarrow}(A:B)_{\rho}$
for any isometries $V\in \mathcal{L}(A,A'),W\in \mathcal{L}(B,B')$ and all $\alpha\in [0,\infty)$.
\item \emph{Additivity:} 
Let $\alpha\in [\frac{1}{2},2]$ and $\rho'_{DE}\in \mathcal{S}(DE)$. Then
\begin{equation}
I_\alpha^{\downarrow\downarrow}(AD:BE)_{\rho_{AB}\otimes \rho'_{DE}}
=I_\alpha^{\downarrow\downarrow}(A:B)_{\rho_{AB}}+I_\alpha^{\downarrow\downarrow}(D:E)_{\rho'_{DE}}.
\end{equation}
\item \emph{Duality:}
Let $\ket{\rho}_{ABC}\in ABC$ be such that $\tr_C[\proj{\rho}_{ABC}]=\rho_{AB}$.
Let $\alpha\in (0,\infty)$ and $\beta\coloneqq \frac{1}{\alpha}\in (0,\infty)$. 
Then
\begin{align}\label{eq:prmi2-duality}
I_\alpha^{\downarrow\downarrow}(A:B)_{\rho}=\begin{cases}
\inf\limits_{\sigma_A\in \mathcal{S}_{\not\perp\rho_A}(A)}
-\frac{1}{\beta-1}\log \widetilde{Q}_\beta(\rho_{AC}\| \sigma_A^{-1}\otimes \rho_C)
\hspace*{2em}\text{if }\alpha\in (0,1)
\\
\inf\limits_{\substack{\sigma_A\in \mathcal{S}(A):\\ \rho_A\ll \sigma_A}}
-\widetilde{D}_{\beta}(\rho_{AC}\| \sigma_A^{-1}\otimes \rho_C )
\hspace*{5em}\text{if }\alpha\in [1,\infty).
\end{cases}
\end{align}
\item \emph{Non-negativity:}
$I_\alpha^{\downarrow\downarrow}(A:B)_\rho\geq 0$ for all $\alpha\in [0,\infty)$.
\item \emph{Upper bound:}
Let $\alpha \in [0,\infty)$ and $r_A\coloneqq \rank(\rho_A)$.
Then $I_\alpha^{\downarrow\downarrow}(A:B)_\rho\leq 2\log r_A$.

Furthermore, if $\alpha\in [\frac{1}{2},2]$, then 
$I_\alpha^{\downarrow\downarrow}(A:B)_\rho= 2\log r_A$ iff 
$\spec(\rho_A)\subseteq\{0,1/r_A\}$ and $H(A|B)_\rho=-\log r_A$.

If $\alpha\in [0,\frac{1}{2})$ instead, then
$I_\alpha^{\downarrow\downarrow}(A:B)_\rho 
\leq \frac{1}{1-\alpha} H_\infty(A)_\rho 
\leq \frac{1}{1-\alpha} \log r_A < 2\log r_A$.
\item \emph{Existence of minimizers:}
Let $\alpha\in [0,\infty)$. Then
\begin{equation}
\emptyset \neq \argmin_{(\sigma_A,\tau_B)\in \mathcal{S}(A)\times \mathcal{S}(B)}
D_\alpha (\rho_{AB}\| \sigma_A\otimes \tau_B )
\subseteq \mathcal{S}_{\ll \rho_A}(A)\times\mathcal{S}_{\ll \rho_B}(B).
\end{equation}
\item \emph{Fixed-point property of partial minimizers:} 
Let $\alpha\in (0,\infty)$.
Let us define the following maps.
\begin{align}
\mathcal{N}_{A\rightarrow B}:
&\quad \mathcal{S}_{\not\perp\rho_A} (A)\rightarrow\mathcal{S}_{\ll\rho_B}(B),\quad
\sigma_A\mapsto \frac{(\tr_A[\rho_{AB}^\alpha \sigma_A^{1-\alpha}])^{\frac{1}{\alpha}}}{\tr[(\tr_A[\rho_{AB}^\alpha \sigma_A^{1-\alpha}])^{\frac{1}{\alpha}}]}\\
\mathcal{N}_{B\rightarrow A}:
&\quad \mathcal{S}_{\not\perp\rho_B}(B)\rightarrow\mathcal{S}_{\ll\rho_A}(A),\quad
\tau_B\mapsto \frac{(\tr_B[\rho_{AB}^\alpha \tau_B^{1-\alpha}])^{\frac{1}{\alpha}}}{\tr[(\tr_B[\rho_{AB}^\alpha \tau_B^{1-\alpha}])^{\frac{1}{\alpha}}]}\\
\mathcal{N}_{A\rightarrow A}:
&\quad \mathcal{S}_{\not\perp\rho_A}(A)\rightarrow\mathcal{S}_{\ll\rho_A}(A),\quad
\sigma_A\mapsto \mathcal{N}_{B\rightarrow A}\circ \mathcal{N}_{A\rightarrow B}(\sigma_A)
\end{align}
Furthermore, let
\begin{align}
\mathcal{M}_\alpha &\coloneqq
\argmin_{\sigma_A\in \mathcal{S}(A)} I_\alpha^{\downarrow} (\rho_{AB}\| \sigma_A),
\\
\mathcal{F}_\alpha &\coloneqq 
\{\sigma_A\in \mathcal{S}_{\ll\rho_A}(A): \mathcal{N}_{A\rightarrow A}(\sigma_A)=\sigma_A \}.
\end{align}
Then $\emptyset \neq \mathcal{M}_\alpha\subseteq \mathcal{F}_\alpha$.
\item \emph{Uniqueness of minimizer:} 
Let $\alpha\in (\frac{1}{2},1]$. 
Then there exists 
$(\hat{\sigma}_A,\hat{\tau}_B)\in \mathcal{S}_{\sim\rho_A}(A)\times \mathcal{S}_{\sim\rho_B}(B)$ such that 
\begin{align}
\argmin_{(\sigma_A,\tau_B)\in \mathcal{S}(A)\times\mathcal{S}(B)} 
D_\alpha (\rho_{AB}\| \sigma_A\otimes \tau_B)
=\{(\hat{\sigma}_A,\hat{\tau}_B)\}.
\end{align}
\item \emph{Fixed-point property of minimizers:} 
Let $\ket{\rho}_{ABC}\in ABC$ be such that 
$\tr_C[\proj{\rho}_{ABC}]=\rho_{AB}$.
Let $\alpha\in (\frac{1}{2},2]$, 
\begin{align}
\mathcal{M}_\alpha
&\coloneqq \argmin_{\sigma_A\in \mathcal{S}(A)}I_\alpha^{\downarrow} (\rho_{AB}\| \sigma_A)
,
\\
\mathcal{F}_\alpha
&\coloneqq
\Big\{\sigma_A\in \mathcal{S}_{\sim \rho_A }(A):
\sigma_A=\frac{\tr_C[(\sigma_A^{\frac{1-\alpha}{2}} \otimes \rho_C^{\frac{\alpha-1}{2}}\rho_{AC} \sigma_A^{\frac{1-\alpha}{2}} \otimes \rho_C^{\frac{\alpha-1}{2}})^{\frac{1}{\alpha}}]}{\tr[ (\sigma_A^{\frac{1-\alpha}{2}} \otimes \rho_C^{\frac{\alpha-1}{2}}\rho_{AC} \sigma_A^{\frac{1-\alpha}{2}} \otimes \rho_C^{\frac{\alpha-1}{2}})^{\frac{1}{\alpha}} ]}
\Big\}.
\label{eq:prmi2-fp}
\end{align}
Then $\mathcal{M}_\alpha=\mathcal{F}_\alpha$.
\item \emph{Asymptotic optimality of universal permutation invariant state:} 
Let $\alpha \in [0,2]$. Then 
\begin{equation}\label{eq:prmi2-omega1}
I_\alpha^{\downarrow\downarrow}(A:B)_\rho
= \lim\limits_{n\rightarrow\infty}\frac{1}{n}D_\alpha (\rho_{AB}^{\otimes n}\| \omega_{A^n}^n\otimes \omega_{B^n}^n)
= \lim\limits_{n\rightarrow\infty} \inf_{\substack{\sigma_{A^n}\in \mathcal{S}_{\sym}(A^{\otimes n}), \\ \tau_{B^n}\in \mathcal{S}_{\sym}(B^{\otimes n}) }}
\frac{1}{n}D_\alpha (\rho_{AB}^{\otimes n}\| \sigma_{A^n}\otimes \tau_{B^n}).
\end{equation}
Moreover, if $\alpha\in [\frac{1}{2},2]$, then for any $n\in \mathbb{N}_{>0}$
\begin{equation}\label{eq:prmi2-omega2}
I_\alpha^{\downarrow\downarrow}(A:B)_\rho 
= \inf\limits_{\substack{\sigma_{A^n}\in \mathcal{S}_{\sym}(A^{\otimes n}),\\ \tau_{B^n}\in \mathcal{S}_{\sym}(B^{\otimes n}) }}
\frac{1}{n} D_\alpha (\rho_{AB}^{\otimes n}\| \sigma_{A^n}\otimes \tau_{B^n})
= \inf\limits_{\substack{\sigma_{A^n}\in \mathcal{S}(A^n),\\ \tau_{B^n}\in \mathcal{S}(B^n) }} 
\frac{1}{n}D_\alpha (\rho_{AB}^{\otimes n}\| \sigma_{A^n}\otimes \tau_{B^n}).
\end{equation}
\item \emph{R\'enyi order $\alpha \in \{0,1\}$:} 
$I_1^{\downarrow\downarrow}(A:B)_\rho = I(A:B)_\rho$ and
\begin{align}
I_0^{\downarrow\downarrow}(A:B)_\rho 
&=\min_{\substack{\ket{\sigma}_A\in \supp(\rho_A),\ket{\tau}_B\in \supp(\rho_B): \\  \brak{\sigma}_A=1, \brak{\tau}_B=1 }}
D_0(\rho_{AB}\| \proj{\sigma}_A\otimes \proj{\tau}_B).
\end{align}
\item \emph{Monotonicity in $\alpha$:} 
If $\alpha,\beta\in [0,\infty)$ are such that $\alpha\leq \beta$, then 
$I_\alpha^{\downarrow\downarrow}(A:B)_\rho \leq I_\beta^{\downarrow\downarrow}(A:B)_\rho$.
\item \emph{Continuity in $\alpha$:} 
The function
$[0,\infty)\rightarrow [0,\infty),\alpha\mapsto I_\alpha^{\downarrow\downarrow}(A:B)_\rho$ is continuous.
\item \emph{Differentiability in $\alpha$:} 
The function $(\frac{1}{2},2)\rightarrow [0,\infty),\alpha\mapsto I_\alpha^{\downarrow\downarrow}(A:B)_\rho$ is continuously differentiable. 
For any $\alpha\in (\frac{1}{2},2)$ and any fixed 
$(\sigma_A,\tau_B)\in \argmin_{(\sigma_A',\tau_B')\in \mathcal{S}(A)\times \mathcal{S}(B)}D_\alpha (\rho_{AB}\| \sigma_A'\otimes \tau_B')$, the derivative at $\alpha$ is
\begin{align}\label{eq:prmi2-diff}
\frac{\mathrm{d}}{\mathrm{d} \alpha}I_\alpha^{\downarrow\downarrow}(A:B)_\rho 
=\frac{\partial}{\partial \alpha}D_\alpha (\rho_{AB}\| \sigma_A\otimes \tau_B).
\end{align}
In particular, $\frac{\mathrm{d}}{\mathrm{d} \alpha}I_\alpha^{\downarrow\downarrow}(A:B)_\rho|_{\alpha=1}=\frac{\mathrm{d}}{\mathrm{d} \alpha}D_\alpha (\rho_{AB}\| \rho_A\otimes \rho_B)|_{\alpha=1}=\frac{1}{2}V(A:B)_\rho$. 

Moreover, $\frac{\partial}{\partial \alpha^+}I_\alpha^{\downarrow\downarrow}(A:B)_\rho|_{\alpha=1/2}=\lim_{\beta\rightarrow 1/2^+}\frac{\mathrm{d}}{\mathrm{d} \alpha}I_\alpha^{\downarrow\downarrow}(A:B)_\rho|_{\alpha=\beta}\in [0,\infty)$.
\item \emph{Convexity in $\alpha$:}
The function
$[0,2]\rightarrow \mathbb{R},\alpha\mapsto (\alpha-1)I_\alpha^{\downarrow \downarrow}(A:B)_\rho$
is convex. 
\item \emph{Product states:} 
If $\rho_{AB}=\rho_A\otimes\rho_B$, then $I_\alpha^{\downarrow\downarrow}(A:B)_{\rho}=0$ for all $\alpha\in [0,\infty)$.
Conversely, for any $\alpha\in (0,\infty)$, if $I_\alpha^{\downarrow\downarrow}(A:B)_\rho=0$, then $\rho_{AB}=\rho_A\otimes \rho_B$.
\item \emph{$AC$-independent states:} 
Let $\ket{\rho}_{ABC}\in ABC$ be such that $\tr_C[\proj{\rho}_{ABC}]=\rho_{AB}$. 
If $\rho_{AC}=\rho_A\otimes \rho_C$, then for all $\alpha\in [0,\infty)$
\begin{equation}
I_\alpha^{\downarrow\downarrow}(A:B)_\rho
=\begin{cases}
\frac{1}{1-\alpha}H_\infty(A)_\rho
\hspace*{2em}\text{if }\alpha\in [0,\frac{1}{2}]
\\
2H_{\frac{1}{2\alpha-1}}(A)_\rho
\hspace*{2em}\text{if }\alpha\in (\frac{1}{2},\infty).
\end{cases}
\end{equation}
\item \emph{Pure states:} 
If there exists $\ket{\rho}_{AB}\in AB$ such that $\rho_{AB}=\proj{\rho}_{AB}$, then for all $\alpha \in [0,\infty)$
\begin{equation}\label{eq:prmi2-pure}
I_\alpha^{\downarrow\downarrow}(A:B)_{\proj{\rho}}
=D_\alpha(\proj{\rho}_{AB}\| \sigma_A\otimes\tau_B)
=\begin{cases}
\frac{1}{1-\alpha}H_\infty(A)_\rho
\hspace*{2em}\text{if }\alpha\in [0,\frac{1}{2}]
\\
2H_{\frac{1}{2\alpha-1}}(A)_\rho
\hspace*{2em}\text{if }\alpha\in (\frac{1}{2},\infty),
\end{cases}
\end{equation}
where $\sigma_A\coloneqq \rho_A^{\frac{1}{2\alpha -1}}/\tr[\rho_A^{\frac{1}{2\alpha -1}}], \tau_B\coloneqq \rho_B^{\frac{1}{2\alpha -1}}/\tr[\rho_B^{\frac{1}{2\alpha -1}}]$ 
if $\alpha\in (\frac{1}{2},\infty)$, and if $\alpha\in [0,\frac{1}{2}]$, then 
$\ket{\sigma}_A\in A$ is defined as a unit eigenvector of $\rho_A$ corresponding to the largest eigenvalue of $\rho_A$, $\sigma_A\coloneqq \proj{\sigma}_A$, 
$\ket{\tau}_B\coloneqq \bra{\sigma}_A\ket{\rho}_{AB}/\sqrt{\bra{\sigma}_A\rho_A\ket{\sigma}_A}$, and $\tau_B\coloneqq \proj{\tau}_B$.
\item \emph{CC states:} 
Let $P_{XY}$ be the joint PMF of two random variables $X,Y$ over $\mathcal{X}\coloneqq [d_A],\mathcal{Y}\coloneqq [d_B]$. 
If there exist orthonormal bases $\{\ket{a_x}_A\}_{x\in [d_A]},\{\ket{b_y}_B\}_{y\in [d_B]}$ for $A,B$ such that 
$\rho_{AB}=\sum_{x\in \mathcal{X}}\sum_{y\in \mathcal{Y}}P_{XY}(x,y)\proj{a_x,b_y}_{AB}$, then for all $\alpha\in [0,\infty)$
\begin{equation}
I_\alpha^{\downarrow\downarrow}(A: B)_\rho 
=I_\alpha^{\downarrow\downarrow}(X:Y)_P.
\end{equation}
\item \emph{Copy-CC states:} 
Let $P_X$ be the PMF of a random variable $X$ over $\mathcal{X}\coloneqq [\min(d_A,d_B)]$. 
Let $Y$ be a random variable over $\mathcal{Y}\coloneqq \mathcal{X}$ and let 
$P_{XY}(x,y)\coloneqq P_X(x)\delta_{x,y}$ for all $x\in \mathcal{X},y\in \mathcal{Y}$. 
If there exist orthonormal bases $\{\ket{a_x}_A\}_{x\in [d_A]},\{\ket{b_y}_B\}_{y\in [d_B]}$ for $A,B$ such that 
$\rho_{AB}=\sum_{x\in \mathcal{X}}\sum_{y\in \mathcal{Y}}P_{XY}(x,y)\proj{a_x,b_y}_{AB}$, then for all $\alpha\in [0,\infty)$
\begin{align}\label{eq:prmi2-copycc}
I_\alpha^{\downarrow\downarrow}(A:B)_\rho 
=I_\alpha^{\downarrow\downarrow}(X:Y)_P
=D_\alpha (\rho_{AB}\| \sigma_A\otimes\tau_B)
&=\begin{cases}
\frac{\alpha}{1-\alpha}H_{\infty}(A)_\rho
\qquad\text{if }\alpha\in [0,\frac{1}{2}]
\\
H_{\frac{\alpha}{2\alpha -1}}(A)_\rho 
\qquad\text{if }\alpha\in (\frac{1}{2},\infty),
\end{cases}
\end{align}
where 
$\sigma_A\coloneqq \rho_A^{\frac{\alpha}{2\alpha -1}}/\tr[\rho_A^{\frac{\alpha}{2\alpha -1}}] , 
\tau_B\coloneqq \rho_B^{\frac{\alpha}{2\alpha -1}}/\tr[\rho_B^{\frac{\alpha}{2\alpha -1}}]$ if $\alpha\in (\frac{1}{2},\infty)$, 
and if $\alpha\in [0,\frac{1}{2}]$, then we let $\hat{x}\in \argmax_{x\in \mathcal{X}}P_X(x)$ be arbitrary but fixed, and 
$\sigma_A\coloneqq \proj{a_{\hat{x}}}_A,\tau_B\coloneqq \proj{b_{\hat{x}}}_B$.
\end{enumerate}
\end{thm}
\begin{proof}
See Appendix~\ref{proof:prmi2}.
\end{proof}

\begin{rem}[Order in the list]
The properties in Theorem~\ref{thm:prmi2} are organized into three groups.
The first group (a)--(l) concerns general properties for a fixed R\'enyi order $\alpha$.
The second group (m)--(q) deals with special values of $\alpha$ and the behavior as $\alpha$ is varied.
The third group (r)--(v) addresses special states. 
This qualitative organization into three groups is also retained in subsequent lists, 
with varying numbers of items within the groups.
\end{rem}

\begin{rem}[Classical case]\label{rem:cc}
According to Theorem~\ref{thm:prmi2}~(u), the doubly minimized PRMI of a CC state reduces to the doubly minimized RMI. 
The doubly minimized RMI can therefore be regarded as a specific instance of the doubly minimized PRMI. 
For more properties of the doubly minimized RMI, we direct the reader to~\cite[Theorem~1]{lapidoth2019two}. 
\end{rem}

\begin{rem}[Inequivalence of PRMIs: Pure states]
The non-minimized, the singly minimized, and the doubly minimized PRMI of a fixed quantum state are not necessarily the same for $\alpha\neq 1$, as we now show. 
Consider a pure state $\rho_{AB}=\proj{\rho}_{AB}\in \mathcal{S}(AB)$. Then, for all $\alpha\in (0,\infty)$
\begin{align}
I_\alpha^{\uparrow\uparrow}(A:B)_\rho 
&=2H_{3-2\alpha}(A)_\rho,
\label{eq:prmi-pure1}\\
I_\alpha^{\uparrow\downarrow}(A:B)_\rho 
&=2H_{\frac{2-\alpha}{\alpha}}(A)_\rho,
\label{eq:prmi-pure2}\\
I_\alpha^{\downarrow\downarrow}(A:B)_\rho 
&=\begin{cases}
\frac{1}{1-\alpha}H_\infty(A)_\rho
\hspace*{2em}\text{if }\alpha\in (0,\frac{1}{2}]
\\
2H_{\frac{1}{2\alpha-1}}(A)_\rho
\hspace*{2em}\text{if }\alpha\in (\frac{1}{2},\infty).
\end{cases}
\label{eq:prmi-pure3}
\end{align}
\eqref{eq:prmi-pure1} follows by evaluating the definition of the non-minimized PRMI in~\eqref{eq:i-quantum-0}, 
\eqref{eq:prmi-pure2} follows from Proposition~\ref{prop:gen-prmi}~(k) and~(p), and 
\eqref{eq:prmi-pure3} follows from Theorem~\ref{thm:prmi2}~(t).
These three quantities are not necessarily the same. 
For illustrative purposes, an example is provided in Figure~\ref{fig:prmi}.
\end{rem}

\begin{rem}[Inequivalence of PRMIs: Copy-CC states]
The non-minimized, the singly minimized, and the doubly minimized PRMI of a fixed quantum state can deviate from each other for $\alpha\neq 1$ already in the classical setting, as we now show. 
Consider a copy-CC state 
$\rho_{AB}=\sum_{x\in \mathcal{X}}\sum_{y\in \mathcal{Y}}P_{XY}(x,y)\proj{a_x,b_y}_{AB}$ 
where $\{\ket{a_x}_A\}_{x\in [d_A]},\{\ket{b_y}_B\}_{y\in [d_B]}$ are orthonormal bases for $A,B$, 
$P_{XY}(x,y)\coloneqq P_X(x)\delta_{x,y}$ for all $x,y\in \mathcal{X}\equiv \mathcal{Y}\coloneqq [\min(d_A,d_B)]$, 
and $P_X$ is the PMF of a random variable $X$ over $\mathcal{X}$. 
Then, for all $\alpha\in (0,\infty)$
\begin{align}
I_\alpha^{\uparrow\uparrow}(A:B)_\rho 
&=I_\alpha^{\uparrow\uparrow}(X:Y)_P
=H_{2-\alpha}(A)_\rho,
\label{eq:prmi-copycc1}\\
I_\alpha^{\uparrow\downarrow}(A:B)_\rho 
&=I_\alpha^{\uparrow\downarrow}(X:Y)_P
=H_{\frac{1}{\alpha}}(A)_\rho,
\label{eq:prmi-copycc2}\\
I_\alpha^{\downarrow\downarrow}(A:B)_\rho 
&=I_\alpha^{\downarrow\downarrow}(X:Y)_P
=\begin{cases}
\frac{\alpha}{1-\alpha}H_{\infty}(A)_\rho
\qquad\text{if }\alpha\in (0,\frac{1}{2}]
\\
H_{\frac{\alpha}{2\alpha -1}}(A)_\rho 
\qquad\text{if }\alpha\in (\frac{1}{2},\infty).
\end{cases}
\label{eq:prmi-copycc3}
\end{align}
\eqref{eq:prmi-copycc1} follows by evaluating the definition of the non-minimized PRMI in~\eqref{eq:i-quantum-0}, 
\eqref{eq:prmi-copycc2} follows from Proposition~\ref{prop:gen-prmi}~(g),~(k), and~(q), and 
\eqref{eq:prmi-copycc3} follows from Theorem~\ref{thm:prmi2}~(v).
These three quantities are not necessarily the same. 
For illustrative purposes, an example is provided in Figure~\ref{fig:prmi-cc}.
\end{rem}

\begin{figure}
\input{fig_prmi.tex}
\caption{Comparison of PRMIs for a pure state. 
Suppose $d_A=2,d_B=2$,
and let $\{\ket{i}_A\}_{i=0}^1,\{\ket{i}_B\}_{i=0}^1$ be orthonormal vectors in $A,B$.
Let $\rho_{AB}\coloneqq\proj{\rho}_{AB}$, where 
$\ket{\rho}_{AB}\coloneqq \sqrt{p}\ket{0,0}_{AB}+\sqrt{1-p}\ket{1,1}_{AB}$ and $p\coloneqq 0.2$. 
Then, $\rho_A=p\proj{0}_A+(1-p)\proj{1}_A$. 
The solid lines depict the behavior of three PRMIs of $\rho_{AB}$, 
computed according to the expressions in~\eqref{eq:prmi-pure1}--\eqref{eq:prmi-pure3}. 
For comparison, the values of certain R\'enyi entropies of $\rho_A$ are indicated by dashed lines. 
The plot shows that the three PRMIs differ from each other for all $\alpha\in [0,1)\cup (1,\infty)$.}
\label{fig:prmi}

\bigskip

\input{fig_prmi_cc.tex}
\caption{Comparison of PRMIs for a copy-CC state. 
Suppose $d_A=2,d_B=2$,
let $\{\ket{i}_A\}_{i=0}^1,\{\ket{i}_B\}_{i=0}^1$ be orthonormal vectors in $A,B$. 
Let $P_{XY}(x,y)\coloneqq p\delta_{x,0}\delta_{y,0}+(1-p)\delta_{x,1}\delta_{y,1}$ for all $x,y\in \mathcal{X}\equiv \mathcal{Y}\coloneqq [d_A]$, where $p\coloneqq 0.2$. 
Let $\rho_{AB}\coloneqq \sum_{x\in \mathcal{X}}\sum_{y\in \mathcal{Y}} P_{XY}(x,y)\proj{x,y}_{AB}=p\proj{0,0}_{AB}+(1-p)\proj{1,1}_{AB}$. 
Then, $\rho_A=p\proj{0}_A+(1-p)\proj{1}_A$. 
The solid lines depict the behavior of three PRMIs for $\rho_{AB}$, 
computed according to the expressions in~\eqref{eq:prmi-copycc1}--\eqref{eq:prmi-copycc3}. 
For comparison, the values of certain R\'enyi entropies of $\rho_A$ are indicated by dashed lines.
The plot shows that the three PRMIs differ from each other for all $\alpha\in [0,1)\cup (1,\infty)$.}
\label{fig:prmi-cc}
\end{figure}

\subsection{Properties of the doubly minimized sandwiched R\'enyi mutual information}\label{ssec:srmi2}
In the following theorem, we present our results on properties of the doubly minimized SRMI of order $\alpha$. 
The focus of our results is on $\alpha\in [1,\infty]$, as this will prove to be the range of relevance for the later application of the doubly minimized SRMI in binary quantum state discrimination~\cite{burri2025prmisrmi2}.

\begin{thm}[Doubly minimized SRMI]\label{thm:srmi2}
Let $\rho_{AB}\in \mathcal{S}(AB)$.
Then all of the following hold.
\begin{enumerate}[label=(\alph*),noitemsep]
\item \emph{Symmetry:} 
$\widetilde{I}_{\alpha}^{\downarrow\downarrow}(A:B)_\rho = \widetilde{I}_\alpha ^{\downarrow\downarrow}(B:A)_\rho$
for all $\alpha \in (0,\infty]$.
\item \emph{Non-increase under local operations:} 
$\widetilde{I}_\alpha^{\downarrow\downarrow}(A:B)_\rho \geq \widetilde{I}_\alpha^{\downarrow\downarrow}(A':B')_{\mathcal{M}\otimes\mathcal{N}(\rho)}$ 
for any $\mathcal{M}\in\CPTP(A,A'),\mathcal{N}\in\CPTP(B,B')$ and all $\alpha\in [\frac{1}{2},\infty]$.
\item \emph{Invariance under local isometries:} 
$\widetilde{I}_\alpha^{\downarrow\downarrow}(A':B')_{V\otimes W\rho V^\dagger\otimes W^\dagger}
=\widetilde{I}_\alpha^{\downarrow\downarrow}(A:B)_{\rho}$ 
for any isometries $V\in \mathcal{L}(A,A'), W\in \mathcal{L}(B,B')$ and all $\alpha\in (0,\infty]$.
\item \emph{Additivity:} 
Let $\alpha\in [\frac{2}{3},\infty]$ and $\rho'_{DE}\in \mathcal{S}(DE)$. Then
\begin{equation}
\widetilde{I}_\alpha^{\downarrow\downarrow}(AD:BE)_{\rho_{AB}\otimes \rho'_{DE}}
=\widetilde{I}_\alpha^{\downarrow\downarrow}(A:B)_{\rho_{AB}}+\widetilde{I}_\alpha^{\downarrow\downarrow}(D:E)_{\rho'_{DE}}.
\end{equation}
\item \emph{Duality:}
Let $\ket{\rho}_{ABC}\in ABC$ be such that $\tr_C[\proj{\rho}_{ABC}]=\rho_{AB}$. 
Let $\alpha\in (\frac{1}{2},1)\cup (1,\infty]$ and $\beta\coloneqq \frac{\alpha}{2\alpha -1}\in[\frac{1}{2},1)\cup (1,\infty)$. Then 
$\frac{1}{\alpha}+\frac{1}{\beta}=2$ and all of the following hold. 

If $\alpha\in (\frac{1}{2},1)$, then 
\begin{align}\label{eq:i-duality-1}
\widetilde{I}_\alpha^{\downarrow\downarrow}(A:B)_{\rho}
&=-\frac{1}{\beta-1}\log 
\sup\limits_{\sigma_A\in \mathcal{S}(A)}
\inf\limits_{\mu_C\in \mathcal{S}_{>0}(C)}
\widetilde{Q}_\beta(\rho_{AC}\| \sigma_A^{-1}\otimes \mu_C).
\end{align}
If $\alpha\in [\frac{2}{3},1)$, then 
\begin{align}\label{eq:i-duality-2}
\widetilde{I}_\alpha^{\downarrow\downarrow}(A:B)_{\rho}
&=-\frac{1}{\beta-1}\log 
\inf\limits_{\mu_C\in \mathcal{S}_{>0}(C)} 
\sup\limits_{\sigma_A\in \mathcal{S}(A)}
\widetilde{Q}_\beta(\rho_{AC}\| \sigma_A^{-1}\otimes \mu_C).
\end{align}
If $\alpha\in (1,\infty]$, then 
\begin{align}
\widetilde{I}_\alpha^{\downarrow\downarrow}(A:B)_{\rho}
&=\frac{1}{1-\beta}\log 
\inf\limits_{\sigma_A\in \mathcal{S}_{\sim \rho_A}(A)}
\sup\limits_{\mu_C\in \mathcal{S}(C)} 
\widetilde{Q}_\beta(\rho_{AC}\| \sigma_A^{-1}\otimes \mu_C)
\label{eq:i-duality1}\\
&=\frac{1}{1-\beta}\log 
\sup\limits_{\mu_C\in \mathcal{S}(C)}
\inf\limits_{\sigma_A\in \mathcal{S}_{\sim \rho_A}(A)}
\widetilde{Q}_\beta(\rho_{AC}\| \sigma_A^{-1}\otimes \mu_C).
\label{eq:i-duality2}
\end{align}
\item \emph{Non-negativity:}
$\widetilde{I}_\alpha^{\downarrow\downarrow}(A:B)_\rho \geq 0$
for all $\alpha\in (0,\infty]$.
\item \emph{Upper bound:} 
Let $\alpha \in (0,\infty]$ and $r_A\coloneqq \rank(\rho_A)$.
Then $\widetilde{I}_\alpha^{\downarrow\downarrow}(A:B)_\rho\leq 2 H_{1/3}(A)_\rho \leq 2\log r_A$.

Furthermore, if $\alpha\in [\frac{2}{3},\infty]$, then 
$\widetilde{I}_\alpha^{\downarrow\downarrow}(A:B)_\rho=2\log r_A$ iff $\spec(\rho_A)\subseteq\{0,1/r_A\}$ and $H(A|B)_\rho=-\log r_A$.

If $\alpha\in [\frac{1}{2},\frac{2}{3})$ instead, then 
$\widetilde{I}_\alpha^{\downarrow\downarrow}(A:B)_\rho
\leq \frac{\alpha}{1-\alpha} H_\infty(A)_\rho 
\leq \frac{\alpha}{1-\alpha}\log r_A
< 2\log r_A$.
\item \emph{Deviation from doubly minimized PRMI:} 
$I_\alpha^{\downarrow\downarrow}(A:B)_\rho\geq \widetilde{I}_\alpha^{\downarrow\downarrow}(A:B)_\rho$ for all $\alpha\in (0,\infty)$ and 
$\widetilde{I}_\alpha^{\downarrow\downarrow}(A:B)_\rho\geq I_{2-\frac{1}{\alpha}}^{\downarrow\downarrow}(A:B)_\rho$ for all $\alpha\in [\frac{1}{2},\infty)$.
\item \emph{Existence of minimizers:}
Let $\alpha\in (0,\infty]$. Then 
\begin{equation}
\emptyset \neq 
\argmin_{(\sigma_A,\tau_B)\in \mathcal{S}(A)\times \mathcal{S}(B)}
\widetilde{D}_\alpha (\rho_{AB}\| \sigma_A\otimes \tau_B )
\subseteq \mathcal{S}_{\ll \rho_A}(A)\times \mathcal{S}_{\ll \rho_B}(B).
\end{equation}
\item \emph{Asymptotic optimality of universal permutation invariant state:} 
Let $\alpha \in [\frac{2}{3},\infty]$. Then 
\begin{equation}\label{eq:srmi2-omega1}
\widetilde{I}_\alpha^{\downarrow\downarrow}(A:B)_\rho
= \lim\limits_{n\rightarrow\infty}\frac{1}{n}\widetilde{D}_\alpha (\rho_{AB}^{\otimes n}\| \omega_{A^n}^n\otimes \omega_{B^n}^n)
= \lim\limits_{n\rightarrow\infty}\frac{1}{n}D_\alpha (\mathcal{P}_{\omega_{A^n}^n\otimes \omega_{B^n}^n}(\rho_{AB}^{\otimes n})\| \omega_{A^n}^n\otimes \omega_{B^n}^n)
\end{equation}
and for any $n\in \mathbb{N}_{>0}$
\begin{equation}\label{eq:srmi2-omega2}
\widetilde{I}_\alpha^{\downarrow\downarrow}(A:B)_\rho
=\inf_{\substack{\sigma_{A^n}\in \mathcal{S}_{\sym}(A^{\otimes n}), \\ \tau_{B^n}\in \mathcal{S}_{\sym}(B^{\otimes n}) }}
\frac{1}{n} \widetilde{D}_\alpha (\rho_{AB}^{\otimes n}\| \sigma_{A^n}\otimes \tau_{B^n})
=\inf_{\substack{\sigma_{A^n}\in \mathcal{S}(A^n), \\ \tau_{B^n}\in \mathcal{S}(B^n) }}
\frac{1}{n} \widetilde{D}_\alpha (\rho_{AB}^{\otimes n}\| \sigma_{A^n}\otimes \tau_{B^n}).
\end{equation}
Furthermore, for any $t\in [0,\infty)$
\begin{equation}\label{eq:srmi2-omega3}
\lim_{n\rightarrow\infty}t\sqrt{n} \left(\frac{1}{n} D_{1+\frac{t}{\sqrt{n}}}(\mathcal{P}_{\omega_{A^n}^n\otimes \omega_{B^n}^n}(\rho_{AB}^{\otimes n})\| \omega_{A^n}^n\otimes \omega_{B^n}^n) 
- I(A:B)_\rho \right)
=\frac{t^2}{2}V(A:B)_\rho.
\end{equation}
\item \emph{R\'enyi order $\alpha=1$:} 
$\widetilde{I}_1^{\downarrow\downarrow}(A:B)_\rho = I(A:B)_\rho$.
\item \emph{Monotonicity in $\alpha$:} 
If $\alpha,\beta\in (0,\infty]$ are such that $\alpha\leq \beta$, then 
$\widetilde{I}_\alpha^{\downarrow\downarrow}(A:B)_\rho \leq \widetilde{I}_\beta^{\downarrow\downarrow}(A:B)_\rho$.
\item \emph{Continuity in $\alpha$:} 
The function 
$(0,\infty)\rightarrow [0,\infty),\alpha \mapsto \widetilde{I}_\alpha^{\downarrow\downarrow} (A:B)_\rho$ is continuous and 
$\lim_{\alpha\rightarrow\infty}\widetilde{I}_\alpha^{\downarrow\downarrow} (A:B)_\rho=\widetilde{I}_{\infty}^{\downarrow\downarrow} (A:B)_\rho$.
\item \emph{Differentiability in $\alpha$:} 
The function 
$(1,\infty)\rightarrow [0,\infty),\alpha \mapsto \widetilde{I}_\alpha^{\downarrow\downarrow} (A:B)_\rho$ is continuously differentiable. 
For any $\alpha\in (1,\infty)$ and any fixed 
$(\sigma_A,\tau_B)\in \argmin_{(\sigma_A',\tau_B')\in \mathcal{S}(A)\times\mathcal{S}(B)} 
\widetilde{D}_\alpha(\rho_{AB}\| \sigma_A'\otimes \tau_B')$, the derivative at $\alpha$ is
\begin{equation}\label{eq:srmi2-diff}
\frac{\mathrm{d}}{\mathrm{d}\alpha} \widetilde{I}_\alpha^{\downarrow\downarrow}(A:B)_\rho 
=\frac{\partial}{\partial\alpha} \widetilde{D}_\alpha (\rho_{AB}\| \sigma_A\otimes \tau_B ).
\end{equation}
Moreover, 
$\frac{\partial}{\partial\alpha^+} \widetilde{I}_\alpha^{\downarrow\downarrow}(A:B)_\rho \big|_{\alpha=1}
=\frac{\mathrm{d}}{\mathrm{d}\alpha} \widetilde{D}_\alpha (\rho_{AB}\| \rho_A\otimes \rho_B )  \big|_{\alpha=1}
=\frac{1}{2}V(A:B)_\rho$.
\item \emph{Convexity in $\alpha$:} 
The function 
$[\frac{2}{3},\infty)\rightarrow \mathbb{R},\alpha\mapsto (\alpha -1) \widetilde{I}_\alpha^{\downarrow\downarrow}(A:B)_\rho$ is convex.
\item \emph{Product states:} 
Let $\alpha\in (0,\infty]$.
Then $\rho_{AB}=\rho_A\otimes \rho_B$ iff $\widetilde{I}^{\downarrow\downarrow}_\alpha(A:B)_\rho=0$.
\item \emph{$AC$-independent states:} 
Let $\ket{\rho}_{ABC}\in ABC$ be such that $\tr_C[\proj{\rho}_{ABC}]=\rho_{AB}$. 
If $\rho_{AC}=\rho_A\otimes\rho_C$, then for all $\alpha\in [\frac{1}{2},\infty)$
\begin{equation}\label{eq:srmi2-ac}
\widetilde{I}_\alpha^{\downarrow\downarrow}(A:B)_{\rho}
=\begin{cases}
\frac{\alpha}{1-\alpha}H_\infty(A)_\rho
\hspace*{2em}\text{if }\alpha\in [\frac{1}{2},\frac{2}{3}]
\\
2H_{\frac{\alpha}{3\alpha-2}}(A)_\rho
\hspace*{2em}\text{if }\alpha\in (\frac{2}{3},\infty).
\end{cases}
\end{equation}
\item \emph{Pure states:} 
If there exists $\ket{\rho}_{AB}\in AB$ such that $\rho_{AB}=\proj{\rho}_{AB}$, then for all $\alpha \in (0,\infty)$
\begin{equation}\label{eq:srmi2-pure}
\widetilde{I}_\alpha^{\downarrow\downarrow}(A:B)_{\proj{\rho}}
=\widetilde{D}_\alpha(\rho_{AB}\| \sigma_A\otimes\tau_B)
=\begin{cases}
\frac{\alpha}{1-\alpha}H_\infty(A)_\rho
\hspace*{2em}\text{if }\alpha\in (0,\frac{2}{3}]
\\
2H_{\frac{\alpha}{3\alpha-2}}(A)_\rho
\hspace*{2em}\text{if }\alpha\in (\frac{2}{3},\infty),
\end{cases}
\end{equation}
where 
$\sigma_A\coloneqq \rho_A^{\frac{\alpha}{3\alpha -2}}/\tr[\rho_A^{\frac{\alpha}{3\alpha -2}}], 
\tau_B\coloneqq \rho_B^{\frac{\alpha}{3\alpha -2}}/\tr[\rho_B^{\frac{\alpha}{3\alpha -2}}]$ 
if $\alpha\in (\frac{2}{3},\infty)$, 
and if $\alpha\in (0,\frac{2}{3}]$, 
then $\ket{\sigma}_A\in A$ is defined as a unit eigenvector of $\rho_A$ corresponding to the largest eigenvalue of $\rho_A$, 
$\sigma_A\coloneqq \proj{\sigma}_A$, $\ket{\tau}_B\coloneqq \bra{\sigma}_A\ket{\rho}_{AB}/\sqrt{\bra{\sigma}_A\rho_A\ket{\sigma}_A}$, and $\tau_B\coloneqq \proj{\tau}_B$.
\item \emph{CC states:} 
Let $P_{XY}$ be the joint PMF of two random variables $X,Y$ over $\mathcal{X}\coloneqq [d_A],\mathcal{Y}\coloneqq [d_B]$. 
If there exist orthonormal bases $\{\ket{a_x}_A\}_{x\in [d_A]},\{\ket{b_y}_B\}_{y\in [d_B]}$ for $A,B$ such that 
$\rho_{AB}=\sum_{x\in \mathcal{X}}\sum_{y\in \mathcal{Y}}P_{XY}(x,y)\proj{a_x,b_y}_{AB}$, then for all $\alpha\in [\frac{1}{2},\infty]$
\begin{equation}
\widetilde{I}_\alpha^{\downarrow\downarrow}(A: B)_\rho
=I_\alpha^{\downarrow\downarrow}(X:Y)_P.
\end{equation}
\end{enumerate}
\end{thm}
\begin{proof}
See Appendix~\ref{proof:srmi2}.
\end{proof}

\begin{rem}[Previous results on properties of the doubly minimized SRMI] \label{rem:previous}
In~\cite{cheng2023tight}, the following three properties of the doubly minimized SRMI of order $\alpha\in (1,\infty)$ have been established. 
Firstly, the optimization problem occurring in the definition of the doubly minimized SRMI~\eqref{eq:srmi2} is jointly convex in $\sigma_A$ and $\tau_B$. 
Secondly, minimizers $(\sigma_A,\tau_B)$ can be characterized in terms of a fixed-point property of these states on $AB$. 
Thirdly, the doubly minimized SRMI of order $\alpha$ is additive. 
Our work does not make use of the results or proof methods in~\cite{cheng2023tight}. 
In particular, our proof of additivity for $\alpha\in [\frac{2}{3},\infty]$ is independent of~\cite{cheng2023tight} because the proof methods used are different. 
In~\cite{cheng2023tight}, the proof of additivity is based on their fixed-point property of minimizers on $AB$. 
In contrast, our proof of additivity is based on the novel duality relation expressed in~\eqref{eq:i-duality-2} and~\eqref{eq:i-duality2} in Theorem~\ref{thm:srmi2}~(e). 
Our proof of duality also employs a fixed-point property, albeit a distinct type of fixed-point property that applies to the dual system $AC$ rather than the original system $AB$, see Lemma~\ref{lem:fixed_point} and Lemma~\ref{lem:fixedpoint}.
\end{rem}

\begin{rem}[Inequivalence of SRMIs: Pure states]
The non-minimized, the singly minimized, and the doubly minimized SRMI of a given quantum state $\rho_{AB}$ are not necessarily the same for $\alpha\neq 1$, as we now show. 
Consider a pure state $\rho_{AB}=\proj{\rho}_{AB}\in \mathcal{S}(AB)$. Then, for all $\alpha\in (0,\infty)$
\begin{align}
\widetilde{I}_\alpha^{\uparrow\uparrow}(A:B)_\rho 
&=2H_{\frac{2-\alpha}{\alpha}}(A)_\rho,
\label{eq:srmi-pure1}\\
\widetilde{I}_\alpha^{\uparrow\downarrow}(A:B)_\rho 
&=\begin{cases}
\frac{1}{1-\alpha} H_{\infty}(A)_\rho 
\quad\text{ if }\alpha\in (0,\frac{1}{2}]
\\
2H_{\frac{1}{2\alpha-1}}(A)_\rho 
\quad\text{ if }\alpha\in (\frac{1}{2},\infty),
\end{cases}
\label{eq:srmi-pure2}\\
\widetilde{I}_\alpha^{\downarrow\downarrow}(A:B)_\rho 
&=\begin{cases}
\frac{\alpha}{1-\alpha}H_\infty(A)_\rho
\hspace*{2em}\text{if }\alpha\in (0,\frac{2}{3}]
\\
2H_{\frac{\alpha}{3\alpha-2}}(A)_\rho
\hspace*{2em}\text{if }\alpha\in (\frac{2}{3},\infty).
\end{cases}
\label{eq:srmi-pure3}
\end{align}
\eqref{eq:srmi-pure1} follows by evaluating the definition of the non-minimized SRMI in~\eqref{eq:i-quantum-00}, 
\eqref{eq:srmi-pure2} follows from Proposition~\ref{prop:gen-srmi}~(l) and~(q), and 
\eqref{eq:srmi-pure3} follows from Theorem~\ref{thm:srmi2}~(r).
These three quantities are not necessarily the same; 
see Figure~\ref{fig:srmi} for an example. 
Furthermore, since the state $\rho_{AB}$ considered in Figure~\ref{fig:srmi} is the same as in Figure~\ref{fig:prmi}, one can, by comparing these two figures, infer that the sandwiched quantities can differ from the corresponding Petz quantities for $\alpha \neq 1$.
\end{rem}

\begin{figure}[t]\centering
\include{fig_srmi.tex}
\caption{Comparison of SRMIs for a pure state. 
Suppose $d_A=2,d_B=2$, 
and let $\{\ket{i}_A\}_{i=0}^1,\{\ket{i}_B\}_{i=0}^1$ be orthonormal vectors in $A,B$.
Let $\rho_{AB}\coloneqq\proj{\rho}_{AB}$, where $\ket{\rho}_{AB}\coloneqq \sqrt{p}\ket{0,0}_{AB}+\sqrt{1-p}\ket{1,1}_{AB}$ and $p\coloneqq 0.2$. 
Then, $\rho_A=p\proj{0}_A+(1-p)\proj{1}_A$. 
The solid lines depict the behavior of three SRMIs of $\rho_{AB}$, 
computed according to the expressions in~\eqref{eq:srmi-pure1}--\eqref{eq:srmi-pure3}. 
For comparison, the values of certain R\'enyi entropies of 
$\rho_A$ are indicated by dashed lines.
The plot shows that the three SRMIs differ from each other for all $\alpha\in (0,1)\cup (1,\infty]$.}
\label{fig:srmi}
\end{figure}

\begin{acknowledgments}
The author thanks Renato Renner for valuable discussions and comments. 
This work was supported by 
the Swiss National Science Foundation via grant No.\ 200021\_188541
and the National Centre of Competence in Research SwissMAP, 
and the Quantum Center at ETH Zurich.
\end{acknowledgments}

\appendix

\section{Properties of the minimized generalized Petz and sandwiched R\'enyi mutual information}\label{sec:gen-prmi}
\begin{prop}[Minimized generalized PRMI]\label{prop:gen-prmi} 
Let $\rho_{AB}\in \mathcal{S}(AB)$ and let $\sigma_A\in \mathcal{S}(A)$.
Then all of the following hold.
\begin{enumerate}[label=(\alph*)]
\item \emph{Non-increase under local operations:} 
$I_\alpha^{\downarrow}(\rho_{AB}\| \sigma_A) \geq I_\alpha^{\downarrow}(\mathcal{M}\otimes \mathcal{N}(\rho_{AB})\| \mathcal{M}(\sigma_A))$ 
for any $\mathcal{M}\in \CPTP(A,A'),\mathcal{N}\in \CPTP(B,B')$ and all $\alpha\in [0,2]$.
\item \emph{Invariance under local isometries:} 
$I_\alpha^{\downarrow}(V\otimes W\rho_{AB}V^\dagger \otimes W^\dagger\| V\sigma_{A}V^\dagger )
=I_\alpha^{\downarrow}(\rho_{AB}\| \sigma_A) $ 
for any isometries $V\in \mathcal{L}(A,A'),W\in \mathcal{L}(B,B')$ and all $\alpha\in [0,\infty)$.
\item \emph{Additivity:} \cite{hayashi2016correlation} 
Let $\alpha\in [0,\infty)$ and $\rho'_{DE}\in \mathcal{S}(DE),\sigma'_D\in \mathcal{S}(D)$. 
If $(\alpha\in [0,1)\land \rho_A\not\perp\sigma_A\land \rho_D'\not\perp\sigma_D')\lor (\rho_A\ll \sigma_A\land \rho'_D\ll \sigma'_D)$, then
\begin{equation}
I_\alpha^{\downarrow}(\rho_{AB}\otimes \rho'_{DE}\| \sigma_A\otimes \sigma'_D)
=I_\alpha^{\downarrow}(\rho_{AB}\| \sigma_A)+I_\alpha^{\downarrow}(\rho'_{DE}\| \sigma'_D).
\end{equation}
\item \emph{Duality:} \cite{hayashi2016correlation} 
Let $\ket{\rho}_{ABC}\in ABC$ be such that $\tr_C[\proj{\rho}_{ABC}]=\rho_{AB}$.
Let $\alpha\in (0,1)\cup (1,\infty)$ and $\beta\coloneqq \frac{1}{\alpha}$. 
If $(\alpha\in (0,1)\land \rho_A\not\perp \sigma_A)\lor  \rho_A\ll \sigma_A$, then
\begin{equation}\label{eq:gen-prmi-duality}
I_\alpha^{\downarrow}(\rho_{AB}\| \sigma_A)
=-\frac{1}{\beta-1}\log  \widetilde{Q}_\beta (\rho_{AC}\| \sigma_A^{-1}\otimes \rho_C).
\end{equation}
\item \emph{Non-negativity:}
Let $\alpha\in [0,\infty)$. Then $I_\alpha^{\downarrow}(\rho_{AB}\| \sigma_A)\in [0,\infty]$.
Furthermore, $I_\alpha^{\downarrow}(\rho_{AB}\| \sigma_A)$ is finite iff 
$(\alpha\in [0,1)\land \rho_A\not\perp\sigma_A)\lor \rho_A\ll \sigma_A$.
\item \emph{Existence and uniqueness of minimizer:} \cite{hayashi2016correlation} 
Let $\alpha\in [0,\infty)$. 
If $\alpha\neq 0\land ((\alpha\in (0,1)\land \rho_A\not\perp\sigma_A)\lor \rho_A\ll \sigma_A)$, then 
$\hat{\tau}_B\coloneqq (\tr_A[\rho_{AB}^\alpha \sigma_A^{1-\alpha}])^{\frac{1}{\alpha}}/\tr[(\tr_A[\rho_{AB}^\alpha \sigma_A^{1-\alpha}])^{\frac{1}{\alpha}}]\in \mathcal{S}_{\ll \rho_B}(B)$ and
\begin{equation}\label{eq:gen-prmi-tau}
\argmin_{\tau_B\in \mathcal{S}(B)}D_\alpha (\rho_{AB}\| \sigma_A\otimes \tau_B)=\{\hat{\tau}_B\}.
\end{equation}
If $\alpha=0\land \rho_A\not\perp\sigma_A$, then $\emptyset\neq \argmin_{\tau_B\in \mathcal{S}(B)}D_\alpha (\rho_{AB}\| \sigma_A\otimes \tau_B)\subseteq \mathcal{S}_{\ll \rho_B}(B)$.
\item \emph{Closed-form expression:} \cite{hayashi2016correlation} 
Let $\alpha\in (0,1)\cup (1,\infty)$.
If $(\alpha\in (0,1)\land \rho_A\not\perp\sigma_A)\lor  \rho_A\ll \sigma_A$, then
\begin{equation}\label{eq:i-gen-explicit}
I_{\alpha}^{\downarrow}(\rho_{AB}\| \sigma_A)
=\frac{1}{\alpha -1}\log\lVert \tr_A[\rho_{AB}^\alpha \sigma_A^{1-\alpha}] \rVert_{\frac{1}{\alpha}}.
\end{equation}
\item \emph{Asymptotic optimality of universal permutation invariant state:} 
Let $\alpha \in [0,2]$. 
If $(\alpha\in [0,1)\land \rho_A\not \perp\sigma_A)\lor \rho_A\ll \sigma_A$, then 
\begin{equation}\label{eq:prmi-gen-omega1}
I_\alpha^{\downarrow}(\rho_{AB}\| \sigma_A)
= \lim\limits_{n\rightarrow\infty}\frac{1}{n}D_\alpha (\rho_{AB}^{\otimes n}\| \sigma_{A}^{\otimes n}\otimes \omega_{B^n}^n)
\end{equation}
and for any $n\in \mathbb{N}_{>0}$
\begin{align}\label{eq:prmi-gen-omega2}
I_\alpha^{\downarrow}(\rho_{AB}\| \sigma_A)
= \inf_{ \tau_{B^n}\in \mathcal{S}_{\sym}(B^{\otimes n}) }
\frac{1}{n}D_\alpha (\rho_{AB}^{\otimes n}\| \sigma_{A}^{\otimes n}\otimes \tau_{B^n})
=\inf_{ \tau_{B^n}\in \mathcal{S}(B^n) }
\frac{1}{n}D_\alpha (\rho_{AB}^{\otimes n}\| \sigma_{A}^{\otimes n}\otimes \tau_{B^n}).
\end{align}
\item \emph{R\'enyi order $\alpha \in \{0,1\}$:} 
$I_1^{\downarrow}(\rho_{AB}\| \sigma_A)=D(\rho_{AB}\| \sigma_A\otimes \rho_B)$.
Furthermore, if $\rho_A\not\perp\sigma_A$, then
\begin{align}\label{eq:i0-gen-down}
I_{0}^{\downarrow}(\rho_{AB}\| \sigma_A)
=-\log \lVert \tr_A[\rho_{AB}^0 \sigma_A] \rVert_{\infty}
=\min_{\substack{\ket{\tau}_B\in \supp(\rho_B):\\ \brak{\tau}_B =1}}D_0(\rho_{AB}\| \sigma_A\otimes \proj{\tau}_B).
\end{align}
\item \emph{Monotonicity in $\alpha$:} 
If $\alpha,\beta\in [0,\infty)$ are such that $\alpha\leq \beta$, then 
$I_\alpha^{\downarrow}(\rho_{AB}\| \sigma_A)\leq I_\beta^{\downarrow}(\rho_{AB}\| \sigma_A)$.
\item \emph{Continuity in $\alpha$:} 
If $\rho_A\not\perp\sigma_A$, then the function 
$[0,1)\rightarrow [0,\infty),\alpha\mapsto I_\alpha^{\downarrow}(\rho_{AB}\| \sigma_A)$ is continuous.
If $\rho_A\ll\sigma_A$, then the function 
$[0,\infty)\rightarrow [0,\infty),\alpha\mapsto I_\alpha^{\downarrow}(\rho_{AB}\| \sigma_A)$ is continuous.
\item \emph{Differentiability in $\alpha$:} 
If $\rho_A\ll\sigma_A$, then all of the following hold.

The function 
$(0,2)\rightarrow [0,\infty),\alpha\mapsto I_\alpha^{\downarrow}(\rho_{AB}\| \sigma_A)$ is continuously differentiable. 
For any $\alpha\in (0,2)$ and any fixed 
$\tau_B\in \argmin_{\tau_B'\in \mathcal{S}(B)} D_\alpha(\rho_{AB}\| \sigma_A\otimes\tau_B')$, the derivative at $\alpha$ is 
\begin{equation}\label{eq:gen-prmi-diff}
\frac{\mathrm{d}}{\mathrm{d} \alpha}I_\alpha^\downarrow(\rho_{AB}\| \sigma_A)
=\frac{\partial}{\partial \alpha} D_\alpha(\rho_{AB}\| \sigma_A\otimes\tau_B).
\end{equation}
In particular,  
$\frac{\mathrm{d}}{\mathrm{d} \alpha}I_\alpha^\downarrow(\rho_{AB}\| \sigma_A)|_{\alpha=1}
=\frac{\mathrm{d}}{\mathrm{d} \alpha} D_\alpha(\rho_{AB}\| \sigma_A\otimes\rho_B)|_{\alpha=1}
=\frac{1}{2}V(\rho_{AB}\| \sigma_A\otimes\rho_B)$.
\item \emph{Convexity in $\alpha$:} 
If $\rho_A\not\perp\sigma_A$, then the function
$[0,1)\rightarrow \mathbb{R},\alpha\mapsto (\alpha-1)I_\alpha^\downarrow (\rho_{AB}\| \sigma_A)$
is convex. 
If $\rho_A\ll\sigma_A$, then the function
$[0,2]\rightarrow \mathbb{R},\alpha\mapsto (\alpha-1)I_\alpha^\downarrow (\rho_{AB}\| \sigma_A)$
is convex.
\item \emph{Product states:} 
If $\rho_{AB}=\rho_A\otimes\rho_B$, then $I_\alpha^{\downarrow}(\rho_{AB}\| \sigma_A)=D_\alpha(\rho_A\| \sigma_A)$ and $I_\alpha^{\downarrow}(\rho_{AB}\| \rho_A)=0$ for all $\alpha\in [0,\infty)$.
Conversely, for any $\alpha\in (0,\infty)$, if $I_\alpha^{\downarrow}(\rho_{AB}\| \sigma_A)=0$, then $\rho_{AB}=\rho_A\otimes\rho_B$ and $\sigma_A=\rho_A$.
\item \emph{$AC$-independent states:} 
Let $\ket{\rho}_{ABC}\in ABC$ be such that $\tr_C[\proj{\rho}_{ABC}]=\rho_{AB}$. 
Let $\alpha\in (0,1)\cup (1,\infty)$ and $\beta\coloneqq \frac{1}{\alpha}$.
If $\rho_{AC}=\rho_A\otimes\rho_C$ and
$(\alpha\in (0,1)\land\rho_A\not\perp\sigma_A)\lor \rho_A\ll \sigma_A$, then 
$I_\alpha^{\downarrow}(\rho_{AB}\|\sigma_A)=-\frac{1}{\beta -1}\log \widetilde{Q}_\beta(\rho_A\| \sigma_A^{-1})$.
\item \emph{Pure states:} 
Let $\alpha\in (0,1)\cup (1,\infty)$ and $\beta\coloneqq \frac{1}{\alpha}$.
If there exists $\ket{\rho}_{AB}\in AB$ such that $\rho_{AB}=\proj{\rho}_{AB}$ 
and $(\alpha\in (0,1)\land \rho_A\not\perp \sigma_A)\lor \rho_A\ll \sigma_A$, then
$I_\alpha^{\downarrow}(\proj{\rho}_{AB}\| \sigma_A)=-\frac{1}{\beta -1}\log \widetilde{Q}_\beta(\rho_A\| \sigma_A^{-1})$.
\item \emph{CC states:} 
Let $\alpha\in [0,\infty)$.
Let $P_{XY}$ be the joint PMF of two random variables $X,Y$ over $\mathcal{X}\coloneqq [d_A],\mathcal{Y}\coloneqq [d_B]$.
If there exist orthonormal bases $\{\ket{a_x}_A\}_{x\in [d_A]},\{\ket{b_y}_B\}_{y\in [d_B]}$ for $A,B$ such that 
$\rho_{AB}=\sum_{x\in \mathcal{X}}\sum_{y\in \mathcal{Y}}P_{XY}(x,y)\proj{a_x,b_y}_{AB}$
and $(\alpha\in [0,1)\land \rho_A\not\perp\sigma_A)\lor \rho_A\ll \sigma_A$, then 
\begin{equation}
I_\alpha^{\downarrow}(\rho_{AB}\| \sigma_A)
=\min_{\substack{\tau_B\in \mathcal{S}(B):\\ \exists (t_y)_{y\in \mathcal{Y}}\in [0,1]^{\times |\mathcal{Y}|}:\\ \tau_B=\sum\limits_{y\in \mathcal{Y}}t_y \proj{b_y}_B }} D_\alpha(\rho_{AB}\| \sigma_A\otimes \tau_B).
\end{equation}
\end{enumerate}
\end{prop}
\begin{proof}
See Appendix~\ref{proof:gen-prmi}.
\end{proof}

\begin{prop}[Minimized generalized SRMI]\label{prop:gen-srmi}
Let $\rho_{AB}\in \mathcal{S}(AB)$ and let $\sigma_A\in \mathcal{S}(A)$.
Then all of the following hold.
\begin{enumerate}[label=(\alph*)]
\item \emph{Non-increase under local operations:} 
$\widetilde{I}_{\alpha}^{\downarrow}(\rho_{AB}\| \sigma_A) \geq \widetilde{I}_\alpha^{\downarrow}(\mathcal{M}\otimes\mathcal{N}(\rho_{AB})\| \mathcal{M}(\sigma_A))$ 
for any $\mathcal{M}\in \CPTP(A,A'),\mathcal{N}\in \CPTP(B,B')$ and all $\alpha\in [\frac{1}{2},\infty]$.
\item \emph{Invariance under local isometries:} 
$\widetilde{I}_\alpha^{\downarrow}(V\otimes W\rho_{AB}V^\dagger\otimes W^\dagger\| V\sigma_A V^\dagger)
=\widetilde{I}_\alpha^{\downarrow}(\rho_{AB}\| \sigma_A)$ 
for any isometries $V\in \mathcal{L}(A,A'),W\in \mathcal{L}(B,B')$ and all $\alpha\in (0,\infty]$.
\item \emph{Additivity:}~\cite{hayashi2016correlation} 
Let $\alpha\in [\frac{1}{2},\infty]$ and $\rho'_{DE}\in \mathcal{S}(DE),\sigma'_D\in \mathcal{S}(D)$. 
If $(\alpha\in [\frac{1}{2},1)\land \rho_A\not\perp\sigma_A\land \rho'_D\not\perp \sigma'_D )\lor (\rho_A\ll \sigma_A\land \rho'_D\ll \sigma'_D)$, then
\begin{equation}
\widetilde{I}_\alpha^{\downarrow}(\rho_{AB}\otimes \rho'_{DE}\| \sigma_A\otimes \sigma'_D)
=\widetilde{I}_{\alpha}^{\downarrow}(\rho_{AB}\| \sigma_A)+\widetilde{I}_\alpha^{\downarrow}(\rho'_{DE}\| \sigma'_D).
\end{equation}
\item \emph{Duality:}~\cite{hayashi2016correlation} 
Let $\ket{\rho}_{ABC}\in ABC$ be such that $\tr_C[\proj{\rho}_{ABC}]=\rho_{AB}$. 
Let $\alpha,\beta\in [\frac{1}{2},\infty]$ be such that $\frac{1}{\alpha}+\frac{1}{\beta}=2$. 

If $\rho_A\ll \sigma_A$, then
$\widetilde{I}_{\alpha}^{\downarrow}(\rho_{AB}\| \sigma_A)=-\widetilde{I}_\beta^{\downarrow}(\rho_{AC}\| \sigma_A^{-1})$.

If $\alpha\in (\frac{1}{2},1)\land \rho_A\not\perp\sigma_A$, then 
\begin{equation}
\widetilde{I}_{\alpha}^{\downarrow}(\rho_{AB}\| \sigma_A)
=\sup_{\substack{\mu_C\in \mathcal{S}_{>0}(C)}}-\frac{1}{\beta-1}\log \widetilde{Q}_\beta (\rho_{AC}\| \sigma_A^{-1}\otimes \mu_C).
\end{equation}
\item \emph{Non-negativity:}
Let $\alpha\in (0,\infty]$. Then $\widetilde{I}_{\alpha}^{\downarrow}(\rho_{AB}\| \sigma_A)\in [0,\infty]$.
Furthermore, $\widetilde{I}_{\alpha}^{\downarrow}(\rho_{AB}\| \sigma_A)$ is finite iff 
$(\alpha\in (0,1)\land \rho_A\not\perp\sigma_A)\lor \rho_A\ll \sigma_A$.
\item \emph{Deviation from minimized generalized PRMI:} 
If $\rho_A\ll \sigma_A$, then 
$I_\alpha^{\downarrow}(\rho_{AB}\| \sigma_A)\geq \widetilde{I}_\alpha^{\downarrow}(\rho_{AB}\| \sigma_A)$ for all $\alpha\in (0,\infty)$ and 
$\widetilde{I}_\alpha^{\downarrow}(\rho_{AB}\| \sigma_A)\geq I_{2-\frac{1}{\alpha}}^{\downarrow}(\rho_{AB}\| \sigma_A)$ for all $\alpha\in [\frac{1}{2},\infty)$.
\item \emph{Existence of minimizers:} 
Let $\alpha\in (0,\infty]$. If $(\alpha\in (0,1)\land \rho_A\not\perp\sigma_A)\lor \rho_A\ll \sigma_A$, then
\begin{equation}\label{eq:gen-srmi-existence}
\emptyset \neq 
\argmin_{\tau_B\in \mathcal{S}(B)}\widetilde{D}_\alpha (\rho_{AB}\| \sigma_A\otimes \tau_B)
\subseteq \mathcal{S}_{\ll \rho_B}(B).
\end{equation}
\item \emph{Uniqueness and fixed-point property of minimizer:}~\cite{hayashi2016correlation}  
Let $\alpha\in (\frac{1}{2},\infty)$. Let 
\begin{align}
\mathcal{M}_\alpha
&\coloneqq\argmin_{\tau_B\in \mathcal{S}(B):\rho_B\ll \tau_B}\widetilde{D}_\alpha (\rho_{AB}\| \sigma_A\otimes\tau_B),
\\
\mathcal{F}_\alpha
&\coloneqq \Big\{
\tau_B\in \mathcal{S}(B):
\rho_B\ll \tau_B, 
\tau_B
=\frac{\tr_A[( (\sigma_A\otimes \tau_B)^{\frac{1-\alpha}{2\alpha}}\rho_{AB}(\sigma_A\otimes \tau_B)^{\frac{1-\alpha}{2\alpha}} )^\alpha ]}{\tr[( (\sigma_A\otimes \tau_B)^{\frac{1-\alpha}{2\alpha}}\rho_{AB}(\sigma_A\otimes \tau_B)^{\frac{1-\alpha}{2\alpha}} )^\alpha ]}
\Big\}.
\end{align}
If $\rho_A\ll \sigma_A$, then 
$\mathcal{M}_\alpha=\mathcal{F}_\alpha$ and this set contains exactly one element.

Moreover, if $\rho_A\ll \sigma_A$ and $\alpha\geq 1$, then $\mathcal{M}_\alpha=\argmin_{\tau_B\in \mathcal{S}(B)}\widetilde{D}_\alpha (\rho_{AB}\| \sigma_A\otimes\tau_B)$.
\item \emph{Asymptotic optimality of universal permutation invariant state:}~\cite{hayashi2016correlation}  
Let $\alpha \in [\frac{1}{2},\infty]$. If 
$(\alpha\in [\frac{1}{2},1)\land \rho_A\not\perp \sigma_A)\lor \rho_A\ll \sigma_A$, then 
\begin{equation}\label{eq:srmi-gen-omega1}
\widetilde{I}_\alpha^{\downarrow}(\rho_{AB}\| \sigma_A)
= \lim\limits_{n\rightarrow\infty}\frac{1}{n}\widetilde{D}_\alpha (\rho_{AB}^{\otimes n}\| \sigma_A^{\otimes n}\otimes \omega_{B^n}^n)
= \lim\limits_{n\rightarrow\infty}\frac{1}{n}D_\alpha (\mathcal{P}_{\sigma_{A}^{\otimes n}\otimes \omega_{B^n}^n}(\rho_{AB}^{\otimes n})\| \sigma_A^{\otimes n}\otimes \omega_{B^n}^n)
\end{equation}
and for any $n\in \mathbb{N}_{>0}$
\begin{align}\label{eq:srmi-gen-omega2}
\widetilde{I}_\alpha^{\downarrow}(\rho_{AB}\| \sigma_A)
=\inf_{ \tau_{B^n}\in \mathcal{S}_{\sym}(B^{\otimes n}) }
\frac{1}{n} \widetilde{D}_\alpha (\rho_{AB}^{\otimes n}\| \sigma_{A}^{\otimes n}\otimes \tau_{B^n})
=\inf_{ \tau_{B^n}\in \mathcal{S}(B^n) }
\frac{1}{n} \widetilde{D}_\alpha (\rho_{AB}^{\otimes n}\| \sigma_{A}^{\otimes n}\otimes \tau_{B^n}).
\end{align}
Furthermore, if $\rho_A\ll \sigma_A$, then 
for any $t\in [0,\infty)$
\begin{equation}\label{eq:srmi-gen-omega3}
\lim_{n\rightarrow\infty}t\sqrt{n} \left(\frac{1}{n} D_{1+\frac{t}{\sqrt{n}}}(\mathcal{P}_{\sigma_{A}^{\otimes n}\otimes \omega_{B^n}^n}(\rho_{AB}^{\otimes n})\| \sigma_{A}^{\otimes n}\otimes \omega_{B^n}^n) 
- \widetilde{I}_1^\downarrow (\rho_{AB}\| \sigma_A) \right)
=\frac{t^2}{2}V(\rho_{AB}\| \sigma_A\otimes \rho_B).
\end{equation}
\item \emph{R\'enyi order $\alpha = 1$:}~\cite{hayashi2016correlation}  
$\widetilde{I}_1^{\downarrow}(\rho_{AB}\| \sigma_A)=D(\rho_{AB}\| \sigma_A\otimes \rho_B)$.
\item \emph{Monotonicity in $\alpha$:} 
If $\alpha,\beta\in (0,\infty]$ are such that $\alpha\leq \beta$, then 
$\widetilde{I}_{\alpha}^{\downarrow}(\rho_{AB}\| \sigma_A)\leq \widetilde{I}_\beta^{\downarrow}(\rho_{AB}\| \sigma_A)$.
\item \emph{Continuity in $\alpha$:}~\cite{hayashi2016correlation}  
If $\rho_A\not\perp\sigma_A$, then the function 
$(0,1)\rightarrow [0,\infty),\alpha\mapsto \widetilde{I}_{\alpha}^{\downarrow}(\rho_{AB}\| \sigma_A)$ is continuous. 
If $\rho_A\ll\sigma_A$, then the function 
$(0,\infty)\rightarrow [0,\infty),\alpha\mapsto \widetilde{I}_{\alpha}^{\downarrow}(\rho_{AB}\| \sigma_A)$ is continuous and 
$\lim_{\alpha\rightarrow\infty}\widetilde{I}_{\alpha}^{\downarrow}(\rho_{AB}\| \sigma_A)=\widetilde{I}_\infty^{\downarrow}(\rho_{AB}\| \sigma_A)$.
\item \emph{Differentiability in $\alpha$:}~\cite{hayashi2016correlation} 
If $\rho_A\ll \sigma_A$, then all of the following hold.

The function 
$(1,\infty)\rightarrow [0,\infty),\alpha \mapsto \widetilde{I}_{\alpha}^{\downarrow}(\rho_{AB}\| \sigma_A)$ is continuously differentiable. 
For any $\alpha\in (1,\infty)$ and any fixed
$\tau_B\in \argmin_{\tau'_B\in \mathcal{S}(B)} \widetilde{D}_\alpha(\rho_{AB}\| \sigma_A\otimes \tau'_B)$, the derivative at $\alpha$ is
\begin{equation}
\frac{\mathrm{d}}{\mathrm{d}\alpha} \widetilde{I}_{\alpha}^{\downarrow}(\rho_{AB}\| \sigma_A)
=\frac{\partial}{\partial\alpha} \widetilde{D}_\alpha (\rho_{AB}\| \sigma_A\otimes \tau_B ).
\end{equation}
Moreover, 
$\frac{\partial}{\partial\alpha^+} \widetilde{I}_\alpha^{\downarrow}(\rho_{AB}\| \sigma_A) \big|_{\alpha=1}
=\frac{\mathrm{d}}{\mathrm{d}\alpha} \widetilde{D}_\alpha (\rho_{AB}\| \sigma_A\otimes \rho_B )  \big|_{\alpha=1}
=\frac{1}{2}V(\rho_{AB}\| \sigma_A\otimes \rho_B)$.
\item \emph{Convexity in $\alpha$:}~\cite{hayashi2016correlation}  
If $\rho_A\ll \sigma_A$, then the function 
$[\frac{1}{2},\infty)\rightarrow \mathbb{R},\alpha\mapsto (\alpha -1) \widetilde{I}_\alpha^\downarrow(\rho_{AB}\| \sigma_A)$ is convex.
\item \emph{Product states:} 
Let $\alpha\in (0,\infty]$.
If $\rho_{AB}=\rho_A\otimes\rho_B$, then $\widetilde{I}_{\alpha}^{\downarrow}(\rho_{AB}\| \sigma_A)=\widetilde{D}_\alpha(\rho_A\| \sigma_A)$ and $\widetilde{I}_\alpha^{\downarrow}(\rho_{AB}\| \rho_A)=0$.
Conversely, 
if $\widetilde{I}_{\alpha}^{\downarrow}(\rho_{AB}\| \sigma_A)=0$, then $\rho_{AB}=\rho_A\otimes\rho_B$ and $\sigma_A=\rho_A$.
\item \emph{$AC$-independent states:} 
Let $\ket{\rho}_{ABC}\in ABC$ be such that $\tr_C[\proj{\rho}_{ABC}]=\rho_{AB}$. 
Let $\alpha\in (\frac{1}{2},1)\cup (1,\infty]$ and $\beta\coloneqq \frac{\alpha}{2\alpha -1}\in [\frac{1}{2},1)\cup(1,\infty)$. 
If $\rho_{AC}=\rho_A\otimes\rho_C$ and 
$(\alpha\in (\frac{1}{2},1)\land \rho_A\not\perp\sigma_A)\lor \rho_A\ll \sigma_A$, then 
$\widetilde{I}_{\alpha}^{\downarrow}(\rho_{AB}\| \sigma_A)
=-\frac{1}{\beta-1}\log\widetilde{Q}_\beta(\rho_{A}\| \sigma_A^{-1})$.
\item \emph{Pure states:} 
Let $\alpha\in (0,1)\cup (1,\infty)$.

If there exists $\ket{\rho}_{AB}\in AB$ such that $\rho_{AB}=\proj{\rho}_{AB}$, then all of the following hold.

If $(\alpha\in (\frac{1}{2},1)\land\rho_A\not\perp\sigma_A)\lor (\alpha\in (1,\infty)\land \rho_A\ll \sigma_A)$, then 
\begin{equation}
\widetilde{I}_\alpha^{\downarrow}(\proj{\rho}_{AB}\| \sigma_A)
=-\frac{1}{\frac{\alpha}{2\alpha-1} -1}\log \widetilde{Q}_{\frac{\alpha}{2\alpha-1}}(\rho_A\| \sigma_A^{-1})
=-\frac{\alpha}{1-\alpha }\log \big\lVert \sigma_A^{\frac{1-\alpha}{2\alpha}}\rho_{A} \sigma_A^{\frac{1-\alpha}{2\alpha}}  \big\rVert_{\frac{\alpha}{2\alpha -1}}.
\end{equation}
If $\alpha \in (0,\frac{1}{2}]\land\rho_A\not\perp\sigma_A$, then 
$\widetilde{I}_\alpha^{\downarrow}(\proj{\rho}_{AB}\| \sigma_A)
=-\frac{\alpha}{1-\alpha}\log \big\lVert \sigma_A^{\frac{1-\alpha}{2\alpha}}\rho_{A} \sigma_A^{\frac{1-\alpha}{2\alpha}} \big\rVert_{\infty}$.
\item \emph{CC states:} 
Let $\alpha\in [\frac{1}{2},\infty]$. 
Let $P_{XY}$ be the joint PMF of two random variables $X,Y$ over $\mathcal{X}\coloneqq [d_A],\mathcal{Y}\coloneqq [d_B]$. 
If there exist orthonormal bases $\{\ket{a_x}_A\}_{x\in [d_A]},\{\ket{b_y}_B\}_{y\in [d_B]}$ for $A,B$ such that 
$\rho_{AB}=\sum_{x\in \mathcal{X}}\sum_{y\in \mathcal{Y}}P_{XY}(x,y)\proj{a_x,b_y}_{AB}$ 
and $(\alpha\in [\frac{1}{2},1)\land\rho_A\not\perp\sigma_A)\lor \rho_A\ll \sigma_A$, then 
\begin{equation}
\widetilde{I}_\alpha^{\downarrow}(\rho_{AB}\| \sigma_A)
=\min_{\substack{\tau_B\in \mathcal{S}(B):\\ \exists (t_y)_{y\in \mathcal{Y}}\in [0,1]^{\times |\mathcal{Y}|}:\\ \tau_B=\sum\limits_{y\in \mathcal{Y}}t_y \proj{b_y}_B }} \widetilde{D}_\alpha(\rho_{AB}\| \sigma_A\otimes \tau_B).
\end{equation}
\end{enumerate}
\end{prop}
\begin{proof}
See Appendix~\ref{proof:gen-srmi}.
\end{proof}

\subsection{Proof of Proposition~\ref{prop:gen-prmi}}\label{proof:gen-prmi}
\begin{proof}[Proof of (a), (e), (j), (n)]
These properties follow from the corresponding properties of the Petz divergence, see Remark~\ref{rem:petz-divergence}. 
In particular,~(n) follows from the additivity and positive definiteness of the Petz divergence.
\end{proof}
\begin{proof}[Proof of (i)]
$I_1^\downarrow(\rho_{AB}\| \sigma_A)=D(\rho_{AB}\| \sigma_A\otimes \rho_B)$ follows from~\eqref{eq:min-re-t}.

Now, suppose $\rho_A\not\perp\sigma_A$. Then
\begin{align}
\exp(-I_0^{\downarrow}(\rho_{AB}\| \sigma_A))
&=\max_{\tau_B\in \mathcal{S}(B)}\tr[\rho^0_{AB}\sigma_A\otimes\tau_B]
\\
&=\max_{\substack{\ket{\tau}_B\in \supp(\rho_B):\\ \brak{\tau}_B=1}}\tr[\rho^0_{AB}\sigma_A\otimes\proj{\tau}_B]
=\lVert \tr_A[\rho_{AB}^0\sigma_A] \rVert_\infty.
\label{eq:tau-ket-tau}
\end{align}
\end{proof}
\begin{proof}[Proof of (f)]
Let $\alpha\in [0,\infty)$.

Case 1: $\alpha\neq 0$. 
For this case, the assertion has been proved in~\cite[Eq.~(3.10)]{hayashi2016correlation}
under the assumption $\alpha\in (0,\infty)\land\rho_A\ll\sigma_A$ by means of a quantum Sibson identity~\cite[Eq.~(B10)]{hayashi2016correlation} (see also~\cite[Proposition~2]{cheng2022properties}).
However, as can be easily verified, their proof still works under the slightly weaker conditions specified in~(f) because the quantum Sibson identity still applies.

Case 2: $\alpha= 0$. For this case, the assertion follows from~(i).
\end{proof}
\begin{proof}[Proof of (g)] 
The closed-form expression follows from the explicit expression for the minimizer in~(f).
\end{proof}
\begin{proof}[Proof of (b)]
Let $\alpha\in [0,\infty)$. 
Let $\hat{\tau}_B\in \mathcal{S}(B)$ be such that 
$I_\alpha^{\downarrow} (\rho_{AB}\| \sigma_A)=D_\alpha (\rho_{AB}\| \sigma_A\otimes \hat{\tau}_B)$. Then 
\begin{align}
I_\alpha^{\downarrow} (\rho_{AB}\| \sigma_A)
&=D_\alpha (V\otimes W\rho_{AB}V^\dagger \otimes W^\dagger\| V\sigma_AV^\dagger \otimes W\hat{\tau}_B W^\dagger)
\label{eq:iso11}\\
&\geq I_\alpha^{\downarrow} (V\otimes W\rho_{AB}V^\dagger \otimes W^\dagger \| V\sigma_A V^\dagger)
\\
&=I_\alpha^{\downarrow}(W\rho_{AB} W^\dagger\| \sigma_A)
\label{eq:iso21}\\
&=\inf_{\substack{\tau_{B'}\in \mathcal{S}(B'):\\ \tau_{B'}\ll W\rho_B W^\dagger }}
D_\alpha (W\rho_{AB}W^\dagger \| \sigma_A\otimes \tau_{B'})
\label{eq:iso31}\\
&\geq \inf_{\tau_B\in \mathcal{S}(B)}
D_\alpha (W\rho_{AB}W^\dagger \| \sigma_A\otimes W \tau_BW^\dagger )
\\
&= I_\alpha^{\downarrow} (\rho_{AB}\| \sigma_A).
\label{eq:iso41}
\end{align}
\eqref{eq:iso11}, \eqref{eq:iso21} and \eqref{eq:iso41} follow from the isometric invariance of the Petz divergence. 
\eqref{eq:iso31} follows from~(f).
\end{proof}
\begin{proof}[Proof of (c)]
Let $\alpha\in [0,\infty)$.

Case 1: $\alpha\in (0,1)\cup (1,\infty)$. 
Then additivity follows from the closed-form expression~(g).

Case 2: $\alpha\in \{0,1\}$. 
Then additivity follows from~(i).
\end{proof}
\begin{proof}[Proof of (d)]
Duality has been proved in~\cite[Lemma~6]{hayashi2016correlation} under the assumption that $\rho_A\ll\sigma_A$. 
However, as can be easily verified, their proof still works under the slightly weaker conditions specified in~(d) due to~(g).
\end{proof}
\begin{proof}[Proof of (h)]
Let $\alpha\in [0,2]$. Then
\begin{subequations}\label{eq:i-down-reg}
\begin{align}
I_\alpha^{\downarrow}(\rho_{AB}\| \sigma_A)
&=\inf_{\tau_B\in \mathcal{S}(B)}\frac{1}{n}
D_\alpha (\rho_{AB}^{\otimes n}\| \sigma_{A}^{\otimes n}\otimes \tau_{B}^{\otimes n})
\label{eq:proof-reg11}\\
&\geq \inf_{\tau_{B^n}\in \mathcal{S}_{\sym}(B^{\otimes n})}
\frac{1}{n}D_\alpha (\rho_{AB}^{\otimes n}\| \sigma_{A}^{\otimes n}\otimes \tau_{B^n})
\label{eq:proof-reg21}\\
&\geq \frac{1}{n}D_\alpha (\rho_{AB}^{\otimes n}\| \sigma_A^{\otimes n}\otimes\omega_{B^n}^n)-\frac{\log g_{n,d_B}}{n}
\label{eq:proof-reg31}\\
&\geq \frac{1}{n}I_\alpha^{\downarrow} (\rho_{AB}^{\otimes n}\| \sigma_A^{\otimes n})-\frac{\log g_{n,d_B}}{n}
= I_\alpha^{\downarrow} (\rho_{AB}\| \sigma_A)-\frac{\log g_{n,d_B}}{n}.
\label{eq:proof-reg41}
\end{align}
\end{subequations}
\eqref{eq:proof-reg11} follows from the additivity of the Petz divergence.
\eqref{eq:proof-reg31} follows from Remark~\ref{rem:universal-state}~(b).
\eqref{eq:proof-reg41} follows from additivity~(c).
The assertion in~\eqref{eq:prmi-gen-omega1} follows from~\eqref{eq:i-down-reg} by taking the limit $n\rightarrow\infty$ and using Remark~\ref{rem:universal-state}~(b).

It remains to prove~\eqref{eq:prmi-gen-omega2}. 
For any $n\in \mathbb{N}_{>0}$
\begin{align}
I_\alpha^{\downarrow}(\rho_{AB}\| \sigma_A)
&\geq \inf_{\tau_{B^n}\in \mathcal{S}_{\sym}(B^{\otimes n})}
\frac{1}{n}D_\alpha (\rho_{AB}^{\otimes n}\| \sigma_{A}^{\otimes n}\otimes \tau_{B^n})
\label{eq:proof-gen-prmi-add-n0}\\
&\geq \inf_{\tau_{B^n}\in \mathcal{S}(B^n)}
\frac{1}{n}D_\alpha (\rho_{AB}^{\otimes n}\| \sigma_{A}^{\otimes n}\otimes \tau_{B^n})
=\frac{1}{n} I_\alpha^\downarrow (\rho_{AB}^{\otimes n}\| \sigma_{A}^{\otimes n})
=I_\alpha^{\downarrow}(\rho_{AB}\| \sigma_A).
\label{eq:proof-gen-prmi-add-n1}
\end{align}
\eqref{eq:proof-gen-prmi-add-n0} follows from~\eqref{eq:proof-reg21}. 
\eqref{eq:proof-gen-prmi-add-n1} follows from additivity~(c).
\end{proof}
\begin{proof}[Proof of (k)]
If $\rho_A\not\perp\sigma_A$, then continuity on $[0,1)$ follows from the continuity in $\alpha$ of the Petz divergence. 

Now, suppose $\rho_A\ll\sigma_A$. 
Then continuity on $[0,1)$ and on $[1,\infty)$ follows from the continuity in $\alpha$ of the Petz divergence. 
It remains to prove left-continuity at $\alpha=1$. 
For any $n\in \mathbb{N}_{>0}$
\begin{align}\label{eq:proof-gen-prmi-continuity}
\frac{1}{n}D_1 (\rho_{AB}^{\otimes n}\| \sigma_A^{\otimes n}\otimes\omega_{B^n}^n) -\frac{\log g_{n,d_B}}{n}
&\leq \lim_{\alpha\rightarrow 1^-}I_\alpha^\downarrow (\rho_{AB}\| \sigma_A )
\leq I_1^\downarrow (\rho_{AB}\| \sigma_A ).
\end{align}
The first inequality in~\eqref{eq:proof-gen-prmi-continuity} follows from~\eqref{eq:i-down-reg}. 
The second inequality in~\eqref{eq:proof-gen-prmi-continuity} follows from monotonicity~(j).
By Remark~\ref{rem:universal-state}~(b), the second term on the left-hand side of~\eqref{eq:proof-gen-prmi-continuity} vanishes in the limit $n\rightarrow\infty$. Thus,
\begin{align}\label{eq:poof-gen-prmi-cont}
I_1^\downarrow (\rho_{AB}\| \sigma_A )=\lim_{n\rightarrow\infty} \frac{1}{n}D_1 (\rho_{AB}^{\otimes n}\| \sigma_A^{\otimes n}\otimes\omega_{B^n}^n)
\leq \lim_{\alpha\rightarrow 1^-}I_\alpha^\downarrow (\rho_{AB}\| \sigma_A )
\leq I_1^\downarrow (\rho_{AB}\| \sigma_A ),
\end{align}
where the first equality in~\eqref{eq:poof-gen-prmi-cont} follows from~(h). 
Hence, $\lim_{\alpha\rightarrow 1^-}I_\alpha^\downarrow (\rho_{AB}\| \sigma_A )=I_1^\downarrow (\rho_{AB}\| \sigma_A )$.
\end{proof}
\begin{proof}[Proof of (m)] 
Convexity is inherited from the Petz divergence because, according to~\eqref{eq:prmi-gen-omega1} in~(h), 
$(\alpha-1)I_\alpha^\downarrow (\rho_{AB}\| \sigma_A)$ is the pointwise limit of a sequence of functions that are convex in $\alpha$.
\end{proof}
\begin{proof}[Proof of (l)]
We will first prove differentiability on $\alpha\in (0,1)\cup (1,2)$.
By the closed-form expression~(g), 
$I_\alpha^\downarrow (\rho_{AB}\| \sigma_A)$ is differentiable in $\alpha$ on $\alpha\in (0,1)\cup (1,2)$. 
In particular, the right and left derivatives of this function at $\alpha\in (0,1)\cup(1,2)$ coincide,
\begin{equation}\label{eq:diff+-}
\frac{\mathrm{d}}{\mathrm{d}\alpha}I_\alpha^\downarrow (\rho_{AB}\| \sigma_A)
=\frac{\partial}{\partial\alpha^+}I_\alpha^\downarrow (\rho_{AB}\| \sigma_A)
=\frac{\partial}{\partial\alpha^-}I_\alpha^\downarrow (\rho_{AB}\| \sigma_A).
\end{equation}
Let $\alpha\in (0,1)\cup (1,2)$ and 
let $\tau_B\in \argmin_{\tau_B'\in\mathcal{S}(B)} D_\alpha (\rho_{AB}\| \sigma_A\otimes \tau_B')$ be fixed. 
Then the right derivative of $I_\alpha^\downarrow (\rho_{AB}\| \sigma_A)$ at $\alpha$ is upper bounded by its left derivative at $\alpha$ because
\begin{align}
\frac{\partial}{\partial\alpha^+}I_\alpha^\downarrow (\rho_{AB}\| \sigma_A)
&=\lim_{\varepsilon\rightarrow 0^+}
\frac{1}{\varepsilon}(I_{\alpha+\varepsilon}^{\downarrow}(\rho_{AB}\|\sigma_A) - I_{\alpha}^{\downarrow}(\rho_{AB}\|\sigma_A) )
\\
&\leq \lim_{\varepsilon\rightarrow 0^+}\frac{1}{\varepsilon}(D_{\alpha+\varepsilon } (\rho_{AB}\| \sigma_A\otimes \tau_B) - D_\alpha (\rho_{AB}\| \sigma_A\otimes \tau_B) )
\\
&=\frac{\partial}{\partial\alpha}D_\alpha (\rho_{AB}\| \sigma_A\otimes \tau_B) 
\label{eq:proof-gen-prmi-diff1}\\
&=\lim_{\varepsilon\rightarrow 0^-}\frac{1}{\varepsilon}(D_{\alpha+\varepsilon } (\rho_{AB}\| \sigma_A\otimes \tau_B) - D_\alpha (\rho_{AB}\| \sigma_A\otimes \tau_B) )
\label{eq:proof-gen-prmi-diff2}\\
&\leq \lim_{\varepsilon\rightarrow 0^-}
\frac{1}{\varepsilon}(I_{\alpha+\varepsilon}^{\downarrow}(\rho_{AB}\|\sigma_A) - I_{\alpha}^{\downarrow}(\rho_{AB}\|\sigma_A) )
=\frac{\partial}{\partial\alpha^-}I_\alpha^\downarrow (\rho_{AB}\| \sigma_A).
\end{align}
\eqref{eq:proof-gen-prmi-diff1} and~\eqref{eq:proof-gen-prmi-diff2} follow from the differentiability in $\alpha$ of the Petz divergence, see Remark~\ref{rem:petz-divergence}. 
By~\eqref{eq:diff+-}, all inequalities must be saturated. 
This proves~\eqref{eq:gen-prmi-diff} for all $\alpha\in (0,1)\cup (1,2)$.

We will now prove that $I_\alpha^\downarrow (\rho_{AB}\| \sigma_A)$ is \emph{continuously} differentiable on $\alpha\in (0,1)\cup (1,2)$.
Let $g(\alpha)\coloneqq (\alpha-1)I_\alpha^\downarrow (\rho_{AB}\| \sigma_A)$ for all $\alpha\in (0,1)\cup (1,2)$. 
By the product rule, $g$ is differentiable on $\alpha\in (0,1)\cup (1,2)$. 
By~(m), $g$ is convex, so the differentiability of $g$ implies its continuous differentiability. 
By the product rule, this implies that also 
$I_\alpha^\downarrow (\rho_{AB}\| \sigma_A)$ is continuously differentiable on $\alpha\in (0,1)\cup (1,2)$. 

Next, we will prove differentiability of $I_\alpha^\downarrow (\rho_{AB}\| \sigma_A)$ at $\alpha=1$. 
The combination of~\eqref{eq:min-re-t} and a quantum Sibson identity~\cite[Eq.~(B10)]{hayashi2016correlation} 
implies that the limits
\begin{subequations}\label{eq:gen-mi-diff}
\begin{align}
\lim_{\beta\rightarrow 1^-}\frac{\mathrm{d}}{\mathrm{d}\alpha}I_\alpha^\downarrow (\rho_{AB}\| \sigma_A)\big|_{\alpha=\beta}
&=\frac{\mathrm{d}}{\mathrm{d} \alpha} D_\alpha (\rho_{AB}\| \sigma_A\otimes\rho_B)\big|_{\alpha=1},
\\
\lim_{\beta\rightarrow 1^+}\frac{\mathrm{d}}{\mathrm{d}\alpha}I_\alpha^\downarrow (\rho_{AB}\| \sigma_A)\big|_{\alpha=\beta}
&=\frac{\mathrm{d}}{\mathrm{d} \alpha} D_\alpha (\rho_{AB}\| \sigma_A\otimes\rho_B)\big|_{\alpha=1}
\end{align}
\end{subequations}
exist. 
Therefore, they are identical to the left and right derivative of $I_\alpha^{\downarrow}(\rho_{AB}\| \sigma_A)$ at $\alpha=1$, respectively. 
By \eqref{eq:gen-mi-diff}, the left and right derivative of $I_\alpha^{\downarrow}(\rho_{AB}\| \sigma_A)$ at $\alpha=1$ coincide. 
This proves differentiability at $\alpha=1$. 
The \emph{continuous} differentiability of $I_\alpha^{\downarrow}(\rho_{AB}\| \sigma_A)$ at $\alpha=1$ follows from the quantum Sibson identity~\cite[Eq.~(B10)]{hayashi2016correlation}.
\end{proof}
\begin{proof}[Proof of (o), (p)]
The assertion in (o) follows from duality~(d), and (p) follows from~(o).
\end{proof}
\begin{proof}[Proof of (q)]
Let $\alpha\in [0,\infty)$.

Case 1: $\alpha\in (0,\infty)$. By~(f), the unique minimizer is then
\begin{equation}
\hat{\tau}_B
=\text{const. }(\tr_A[\rho_{AB}^\alpha \sigma_A^{1-\alpha}])^{\frac{1}{\alpha}}
=\text{const. }\sum_{y\in \mathcal{Y}}\left(\sum_{x\in \mathcal{X}} P_{XY}(x,y)^\alpha \bra{a_x}_A \sigma_A^{1-\alpha}\ket{a_x}_A\right)^{\frac{1}{\alpha}}\proj{b_y}_B,
\end{equation}
so $\hat{\tau}_B$ has the desired form.

Case 2: $\alpha=0$. 
For all $y\in \mathcal{Y}$, let 
$c_y\coloneqq \sum_{x\in \mathcal{X}: P_{XY}(x,y)\neq 0}\bra{a_x}_A\sigma_A\ket{a_x}_A$. 
Let $\hat{y}\in \argmax_{y\in \mathcal{Y}}c_y$ and let $\ket{\hat{\tau}}_B\coloneqq \ket{b_{\hat{y}}}_B$.
By (i),
\begin{align}
\exp(-I_0^{\downarrow}(\rho_{AB}\| \sigma_A))
=\Big\lVert \sum_{y\in \mathcal{Y}}c_y\proj{b_y}_B \Big\rVert_\infty 
=\max_{y\in \mathcal{Y}}c_y 
=c_{\hat{y}}
=\exp(-D_0(\rho_{AB}\| \sigma_A\otimes \proj{\hat{\tau}}_B)).
\end{align}
Therefore, $\hat{\tau}_B\coloneqq\proj{\hat{\tau}}_B=\proj{b_{\hat{y}}}_B$ is a minimizer that has the desired form.
\end{proof}

\subsection{Proof of Proposition~\ref{prop:gen-srmi}}\label{proof:gen-srmi}

\begin{proof}[Proof of (a), (e), (k), (o)]
These properties follow from the corresponding properties of the sandwiched divergence, see Remark~\ref{rem:sandwiched-divergence}. 
In particular,~(o) follows from the additivity and positive definiteness of the sandwiched divergence. 
\end{proof}
\begin{proof}[Proof of (h), (j), (l), (m)]
These properties have been proved in previous work. 
For (h) and (j), see~\cite[Lemma~5]{hayashi2016correlation}. 
For (l), see~\cite[Corollary~10]{hayashi2016correlation}.
For (m), see \cite[Proposition~11]{hayashi2016correlation}.
\end{proof}
\begin{proof}[Proof of (g)]
Let $\alpha\in (0,\infty]$. Suppose $(\alpha\in (0,1)\land \rho_A\not\perp\sigma_A)\lor \rho_A\ll \sigma_A$. 

Let $\tau_B\in \argmin_{\tau_B'\in \mathcal{S}(B)}\widetilde{D}_\alpha (\rho_{AB}\|\sigma_A\otimes \tau_B')$. We will now show that $\tau_B\ll \rho_B$.

Case 1: $\alpha\in (0,1)\cup (1,\infty)$. 
Let $\hat{\tau}_B\coloneqq (\rho_B^0\tau_B^{\frac{1-\alpha}{\alpha}}\rho_B^0)^{\frac{\alpha}{1-\alpha}}/c$ where $c\coloneqq \tr[(\rho_B^0\tau_B^{\frac{1-\alpha}{\alpha}}\rho_B^0)^{\frac{\alpha}{1-\alpha}}]$.
Then, because $c\leq 1$,
\begin{align}
\widetilde{I}_\alpha^\downarrow (\rho_{AB}\| \sigma_A)
=\widetilde{D}_\alpha (\rho_{AB}\| \sigma_A\otimes \tau_B)
&= \widetilde{D}_\alpha (\rho_{AB}\| \sigma_A\otimes \hat{\tau}_B)- \log c
\\
&\geq \widetilde{D}_\alpha (\rho_{AB}\| \sigma_A\otimes \hat{\tau}_B)
\geq \widetilde{I}_\alpha^\downarrow (\rho_{AB}\| \sigma_A).
\end{align}
It follows that both inequalities are saturated.
Hence, $c=1$.
Therefore, $\tau_B\ll \rho_B$.

Case 2: $\alpha=1$. Then $\tau_B=\rho_B$~\cite{hayashi2016correlation}. 
Therefore, $\tau_B\ll \rho_B$.

Case 3: $\alpha=\infty$. 
Let $\hat{\tau}_B\coloneqq \rho_B^0\tau_B\rho_B^0/c$ where $c\coloneqq \tr[\rho_B^0\tau_B]$. 
Then, because $c\leq 1$, 
\begin{align}
\widetilde{I}_\infty^{\downarrow}(\rho_{AB}\| \sigma_A)
=D_{\max}(\rho_{AB}\| \sigma_A \otimes \tau_B)
&=D_{\max}(\rho_{AB}\| \sigma_A \otimes \hat{\tau}_B)-\log c
\\
&\geq D_{\max}(\rho_{AB}\| \sigma_A \otimes \hat{\tau}_B)
\geq \widetilde{I}_\infty^{\downarrow}(\rho_{AB}\| \sigma_A).
\label{eq:i-gen-ex}
\end{align}
It follows that both inequalities are saturated. 
Hence, $c=1$.
Therefore, $\tau_B\ll \rho_B$.

\end{proof}
\begin{proof}[Proof of (b)] 
Let $\alpha\in (0,\infty]$. 
Let $\hat{\tau}_B\in \mathcal{S}(B)$ be such that 
$\widetilde{I}_\alpha^{\downarrow} (\rho_{AB}\| \sigma_A)=\widetilde{D}_\alpha (\rho_{AB}\| \sigma_A\otimes \hat{\tau}_B)$. Then 
\begin{align}
\widetilde{I}_\alpha^{\downarrow} (\rho_{AB}\| \sigma_A)
&=\widetilde{D}_\alpha (V\otimes W\rho_{AB}V^\dagger \otimes W^\dagger\| V\sigma_AV^\dagger \otimes W\hat{\tau}_B W^\dagger)
\label{eq:iso12}\\
&\geq \widetilde{I}_\alpha^{\downarrow} (V\otimes W\rho_{AB}V^\dagger \otimes W^\dagger \| V\sigma_A V^\dagger)
\\
&=\widetilde{I}_\alpha^{\downarrow}(W\rho_{AB} W^\dagger\| \sigma_A)
\label{eq:iso22}\\
&=\inf_{\substack{\tau_{B'}\in \mathcal{S}(B'):\\ \tau_{B'}\ll W\rho_B W^\dagger }}
\widetilde{D}_\alpha (W\rho_{AB}W^\dagger \| \sigma_A\otimes \tau_{B'})
\label{eq:iso32}\\
&\geq\inf_{\tau_B\in \mathcal{S}(B)}
\widetilde{D}_\alpha (W\rho_{AB}W^\dagger \| \sigma_A\otimes W \tau_BW^\dagger )
\\
&= \widetilde{I}_\alpha^{\downarrow} (\rho_{AB}\| \sigma_A).
\label{eq:iso42}
\end{align}
\eqref{eq:iso12}, \eqref{eq:iso22}, and~\eqref{eq:iso42} follow from the isometric invariance of the sandwiched divergence. 
\eqref{eq:iso32} follows from~(g).
\end{proof}
\begin{proof}[Proof of (d)]
In~\cite[Lemma~6]{hayashi2016correlation}, the assertion has been proved for the case $\rho_A\ll\sigma_A$. 
However, for the case 
$\alpha\in (\frac{1}{2},1)\land \rho_A\not\perp\sigma_A$, an analogous proof works.
\end{proof}
\begin{proof}[Proof of (c)]
For $\alpha\in (\frac{1}{2},\infty]$, additivity follows from duality~(d), see~\cite[Lemma~7]{hayashi2016correlation}. 
From this, additivity for $\alpha=\frac{1}{2}$ follows due to the continuity in $\alpha$ of the sandwiched divergence.
\end{proof}
\begin{proof}[Proof of (f)]
$I_\alpha^{\downarrow}(\rho_{AB}\| \sigma_A)\geq \widetilde{I}_\alpha^{\downarrow}(\rho_{AB}\| \sigma_A)$ for all $\alpha\in (0,\infty)$ 
follows from~\eqref{eq:alt}.

Let $\ket{\rho}_{ABC}$ be such that $\tr_C[\proj{\rho}_{ABC}]=\rho_{AB}$. Then, for any $\alpha\in (\frac{1}{2},\infty)$
\begin{align}\label{eq:gen-srmi-dev}
\widetilde{I}_\alpha^\downarrow(\rho_{AB}\| \sigma_A)
\geq -\widetilde{D}_{\frac{\alpha}{2\alpha -1}}(\rho_{AC}\| \sigma_A^{-1}\otimes \rho_C)
=I_{2-\frac{1}{\alpha}}^{\downarrow}(\rho_{AB}\| \sigma_A).
\end{align}
The inequality in~\eqref{eq:gen-srmi-dev} follows from duality~(d).
The equality in~\eqref{eq:gen-srmi-dev} follows from the duality of the minimized generalized PRMI, see Proposition~\ref{prop:gen-prmi}. 
From this, the assertion for $\alpha=\frac{1}{2}$ follows by continuity in $\alpha$~(l).
\end{proof}
\begin{proof}[Proof of (i)]
The assertion in~\eqref{eq:srmi-gen-omega3} has been proved in~\cite[Corollary~9]{hayashi2016correlation}. 

The assertion in~\eqref{eq:srmi-gen-omega1} has been proved in~\cite[Proposition~8]{hayashi2016correlation} under the assumption that $\rho_A\ll \sigma_A$. 
However, their proof remains valid under the slightly less restrictive conditions specified in~(i).

It remains to prove~\eqref{eq:srmi-gen-omega2}. 
For any $n\in \mathbb{N}_{>0}$
\begin{align}
\widetilde{I}_\alpha^{\downarrow}(\rho_{AB}\| \sigma_A)
&=  \inf_{\tau_{B}\in \mathcal{S}(B)}
\frac{1}{n}\widetilde{D}_\alpha (\rho_{AB}^{\otimes n}\| \sigma_{A}^{\otimes n}\otimes \tau_{B}^{\otimes n})
\label{eq:proof-rreg0}\\
&\geq \inf_{\tau_{B^n}\in \mathcal{S}_{\sym}(B^{\otimes n})}
\frac{1}{n} \widetilde{D}_\alpha (\rho_{AB}^{\otimes n}\| \sigma_{A}^{\otimes n}\otimes \tau_{B^n})
\label{eq:proof-rreg1}\\
&\geq \inf_{\tau_{B^n}\in \mathcal{S}(B^n)}
\frac{1}{n}\widetilde{D}_\alpha (\rho_{AB}^{\otimes n}\| \sigma_{A}^{\otimes n}\otimes \tau_{B^n})
=\frac{1}{n} \widetilde{I}_\alpha^\downarrow (\rho_{AB}^{\otimes n}\| \sigma_{A}^{\otimes n})
=\widetilde{I}_\alpha^{\downarrow}(\rho_{AB}\| \sigma_A).
\label{eq:proof-rreg2}
\end{align}
\eqref{eq:proof-rreg0} follows from the additivity of the sandwiched divergence. 
\eqref{eq:proof-rreg2} follows from additivity~(c).
\end{proof}
\begin{proof}[Proof of (n)]
As noted in~\cite[Corollary~10]{hayashi2016correlation}, convexity on $[\frac{1}{2},\infty)$ is inherited from the sandwiched divergence because, according to the first equality in~\eqref{eq:srmi-gen-omega1} in~(i), $(\alpha-1)\widetilde{I}_\alpha^\downarrow(\rho_{AB}\| \sigma_A)$ is the pointwise limit of a sequence of functions that are convex in $\alpha$. 
\end{proof}
\begin{proof}[Proof of (p)]
This assertion follows from duality (d).
\end{proof}
\begin{proof}[Proof of (q)]
Let $\alpha\in (0,\infty]$.

Case 1: $\alpha\in (\frac{1}{2},\infty)$. For this case, the assertion follows from~(p).

Case 2: $\alpha\in (0,\frac{1}{2}]$. Then $\frac{1-\alpha}{\alpha}\in [1,\infty)$. Hence, 
\begin{align}
\exp\left(-\frac{1-\alpha}{\alpha}\widetilde{I}_{\alpha}^{\downarrow}(\proj{\rho}_{AB}\| \sigma_A)\right)
&=\sup_{\tau_B\in \mathcal{S}(B)} \bra{\rho}_{AB}(\sigma_A\otimes\tau_B)^{\frac{1-\alpha}{\alpha}}\ket{\rho}_{AB}
\\
&= \sup_{\substack{\ket{\tau}_B\in B:\\ \brak{\tau}_B=1}} \bra{\rho}_{AB}\sigma_A^{\frac{1-\alpha}{\alpha}}\otimes \proj{\tau}_B\ket{\rho}_{AB}
\\
&= \sup_{\substack{\ket{\tau}_B\in B:\\ \brak{\tau}_B=1}}\tr[\proj{\tau}_B \tr_A[\sigma_A^{\frac{1-\alpha}{2\alpha}}\proj{\rho}_{AB}\sigma_A^{\frac{1-\alpha}{2\alpha}}]]
\\
&=\big\lVert \tr_A[\sigma_A^{\frac{1-\alpha}{2\alpha}}\proj{\rho}_{AB}\sigma_A^{\frac{1-\alpha}{2\alpha}}] \big\rVert_\infty
\\
&=\big\lVert \tr_B[\sigma_A^{\frac{1-\alpha}{2\alpha}}\proj{\rho}_{AB}\sigma_A^{\frac{1-\alpha}{2\alpha}}] \big\rVert_\infty
=\big\lVert \sigma_A^{\frac{1-\alpha}{2\alpha}}\rho_A\sigma_A^{\frac{1-\alpha}{2\alpha}} \big\rVert_\infty.
\end{align}
\end{proof}
\begin{proof}[Proof of (r)] 
Let $\alpha\in [\frac{1}{2},\infty]$.
Let $\tau_B\in \mathcal{S}(B)$ be such that $\widetilde{I}_\alpha^{\downarrow}(\rho_{AB}\| \sigma_A)=\widetilde{D}_\alpha(\rho_{AB}\| \sigma_A\otimes \tau_B)$.
Let 
$\tau_B'\coloneqq \sum_{y\in \mathcal{Y}}\proj{b_y}_{B}\tau_B\proj{b_y}_{B}\in \mathcal{S}(B)$. 
Then, 
\begin{align}
\widetilde{D}_\alpha(\rho_{AB}\| \sigma_A\otimes \tau_B)
&\geq 
\widetilde{D}_\alpha (\sum_{y\in \mathcal{Y}}\proj{b_y}_{B}\rho_{AB}\proj{b_y}_{B}\| \sigma_A\otimes \sum_{y\in \mathcal{Y}}\proj{b_y}_{B}\tau_B\proj{b_y}_{B})
\label{eq:proof-migen-cc}\\
&=\widetilde{D}_\alpha(\rho_{AB}\| \sigma_A\otimes \tau_B')
\geq \widetilde{I}_\alpha^{\downarrow}(\rho_{AB}\| \sigma_A).
\end{align}
\eqref{eq:proof-migen-cc} follows from the data-processing inequality for the sandwiched divergence, see Remark~\ref{rem:sandwiched-divergence}. 
Therefore, $\widetilde{I}_\alpha^{\downarrow}(\rho_{AB}\| \sigma_A)=\widetilde{D}_\alpha(\rho_{AB}\| \sigma_A\otimes \tau_B')$.
Since $\tau_B'$ has the desired form, the assertion is implied.
\end{proof}

\section{Proof of Lemma~\ref{lem:cs}}\label{app:lem:cs}

\begin{proof}
\begin{align}
\sqrt{X_A}\otimes \sqrt{Y_B} + \sqrt{X_A'}\otimes \sqrt{Y_B'}
&= (X_A\otimes 1_B)\# (1_A\otimes Y_B) + (X_A'\otimes 1_B)\#  (1_A\otimes Y_B')
\\
&\leq ((X_A\otimes 1_B)+(X_A'\otimes 1_B))\# ((1_A\otimes Y_B)+(1_A\otimes Y_B'))
\\
&= \sqrt{X_A+X_A'}\otimes \sqrt{Y_B+Y_B'}
\end{align}
The inequality follows from the subadditivity~\cite{kubo1980means} of the geometric operator mean.
\end{proof}

\section{Proof of Theorem~\ref{thm:convexity}}\label{app:convexity}

\subsection{Lemmas for Theorem~\ref{thm:convexity}}\label{ssec:lemmas}
\begin{lem}[Saturation of operator inequality from saturation of trace inequality]\label{lem:xyz}
Let $X,Y,Z\in \mathcal{L}(A)$ be positive semidefinite and such that $X \leq Y$. 
Then all of the following hold.
\begin{enumerate}[label=(\alph*)]
\item $\tr[XZ]\leq \tr[YZ]$.
\item If $\tr[XZ]=\tr[YZ]$ and $Y\ll Z$, then $X=Y$.
\end{enumerate}
\end{lem}
\begin{proof}
By spectral decomposition, $Z=\sum_{\lambda\in \spec(Z)} \lambda P_\lambda$, where 
$P_\lambda\in \mathcal{L}(A)$ denotes the orthogonal projection onto the eigenspace corresponding to $\lambda$. 

We will now prove~(a).
\begin{align}
\tr[YZ]-\tr[XZ]
&=\tr[(Y-X)Z] 
=\sum_{\lambda\in \spec(Z)} \lambda  \tr[(Y-X)P_\lambda ] 
=\sum_{\substack{\lambda\in \spec(Z):\\ \lambda\neq 0}} \lambda  \tr[(Y-X)P_\lambda ] 
\geq 0
\label{eq:lem-xyz-a}
\end{align}

We will now prove~(b). Suppose $\tr[XZ]=\tr[YZ]$. Then, \eqref{eq:lem-xyz-a} implies that
$\tr[(Y-X)P_\lambda] = 0$ for all $\lambda\in \spec(Z)$ such that $\lambda\neq 0$. 
Since the pinching map is trace-preserving, 
\begin{align}
\tr[Y-X]
=\tr[\mathcal{P}_Z(Y-X)]
=\sum_{\lambda\in \spec(Z)}\tr[(Y-X)P_\lambda]
=\sum_{\substack{\lambda\in \spec(Z):\\ \lambda\neq 0}}\tr[(Y-X)P_\lambda]
=0,
\label{eq:lem-xyz-tr}
\end{align}
where we used that $X\leq Y$ and $Y\ll Z$ implies $(Y-X)\ll Z$. 
By~\eqref{eq:lem-xyz-tr}, $\tr[Y-X]=0$. 
Because $(Y-X)\geq 0$, the trace of $(Y-X)$ vanishes iff the operator itself vanishes. 
Therefore, $X=Y$.
\end{proof}

\begin{lem}[Strict concavity for states]\label{lem:strict-concavity}
Let $\sigma,\sigma'\in \mathcal{S}(A),p\in (1,\infty)$, 
and let $\lambda,\lambda'\in (0,1)$ be such that $\lambda+\lambda'=1$. Then 
\begin{align}\label{eq:strict-concavity}
\lambda \sigma^{\frac{1}{p}}+\lambda' {\sigma'}^{\frac{1}{p}} \leq (\lambda\sigma+\lambda'\sigma')^{\frac{1}{p}}.
\end{align}
Moreover, if~\eqref{eq:strict-concavity} holds with equality, then $\sigma=\sigma'$. 
\end{lem}
\begin{proof}
The inequality in~\eqref{eq:strict-concavity} follows from the operator concavity of $X\mapsto X^{\frac{1}{p}}$ since $\frac{1}{p}\in (0,1)$. 

Now, suppose~\eqref{eq:strict-concavity} holds with equality. 
Let $X\coloneqq \sigma^{\frac{1}{p}}$ and $X'\coloneqq {\sigma'}^{\frac{1}{p}}$. 
By the saturation of the inequality~\eqref{eq:strict-concavity}, 
$\lambda X+\lambda'X'=(\lambda X^p+\lambda' {X'}^p)^{\frac{1}{p}}$. 
Hence, 
\begin{align}
\lVert \lambda X+\lambda' X'\lVert_p 
=\lVert (\lambda X^p+\lambda' {X'}^p)^{\frac{1}{p}} \lVert_p 
=(\tr[\lambda X^p+\lambda' {X'}^p])^{\frac{1}{p}}
=(\lambda+\lambda')^{\frac{1}{p}}=1.
\end{align}
By the subadditivity of norms, 
\begin{align}\label{eq:norm-subadditivity}
1=\lVert \lambda X+\lambda'X'\rVert_p
\leq \lVert \lambda X\rVert_p+\lVert \lambda'X'\rVert_p
=\lambda +\lambda'=1.
\end{align}
Hence, the inequality in~\eqref{eq:norm-subadditivity} must be saturated.  
The saturation of the inequality in~\eqref{eq:norm-subadditivity} implies that $(\lambda X)^p$ is proportional to $(\lambda' X')^p$ for some strictly positive proportionality constant (due to the variational characterization of the  Schatten norms~\cite[Lemma~3.2]{tomamichel2016quantum}). 
Since $(\lambda X)^p=\lambda^p\sigma$ and $(\lambda' X')^p=(\lambda')^p\sigma'$, we conclude that $\sigma=\sigma'$.
\end{proof}

\subsection{Proof of Theorem~\ref{thm:convexity}}
\begin{proof}
Let $\alpha\in [\frac{1}{2},1)$. 
We will now prove~\eqref{eq:joint-concavity}.
\begin{align}
&\lambda Q_\alpha(\rho_{AB}\| \sigma_A \otimes \tau_B ) + \lambda' Q_\alpha(\rho_{AB}\| \sigma_A' \otimes \tau_B' )\\
&= \lambda \tr[\rho_{AB}^\alpha (\sigma_A\otimes \tau_B)^{1-\alpha}]+\lambda'\tr[\rho_{AB}^\alpha (\sigma_A'\otimes \tau_B')^{1-\alpha}]\\
&=\tr[\rho_{AB}^\alpha (\lambda (\sigma_A\otimes \tau_B)^{1-\alpha} +\lambda' (\sigma_A'\otimes \tau_B')^{1-\alpha} )]
\label{eq:c3}\\
&=\tr[\rho_{AB}^\alpha (\sqrt{\lambda \sigma_A^{2(1-\alpha)}}\otimes \sqrt{\lambda\tau_B^{2(1-\alpha)}} 
+\sqrt{\lambda' {\sigma_A'}^{2(1-\alpha)}}\otimes \sqrt{\lambda'{\tau_B'}^{2(1-\alpha)}} )]
\label{eq:c4}\\
&\leq\tr[\rho_{AB}^\alpha \sqrt{\lambda\sigma_A^{2(1-\alpha)} +\lambda'{\sigma_A'}^{2(1-\alpha)} } \otimes \sqrt{\lambda \tau_B^{2(1-\alpha)}+\lambda' {\tau_B'}^{2(1-\alpha)} }]
\label{eq:c5}\\
&\leq\tr[\rho_{AB}^\alpha \sqrt{(\lambda\sigma_A +\lambda'\sigma_A')^{2(1-\alpha)}} 
\otimes \sqrt{\lambda \tau_B^{2(1-\alpha)}+\lambda' {\tau_B'}^{2(1-\alpha)} }]
\label{eq:c51}\\
&\leq\tr[\rho_{AB}^\alpha \sqrt{(\lambda\sigma_A +\lambda'\sigma_A')^{2(1-\alpha)}} 
\otimes \sqrt{(\lambda \tau_B+\lambda' \tau_B')^{2(1-\alpha)} }]
\label{eq:c6}\\
&=\tr[\rho_{AB}^\alpha (\lambda\sigma_A +\lambda'\sigma_A')^{1-\alpha} 
\otimes (\lambda \tau_B+\lambda' \tau_B')^{1-\alpha}]
\label{eq:c7}\\
&=Q_\alpha(\rho_{AB}\| (\lambda\sigma_A +\lambda'\sigma_A') \otimes (\lambda \tau_B+\lambda' \tau_B') )
\end{align}
\eqref{eq:c5} follows from Lemma~\ref{lem:cs}. 
\eqref{eq:c51} and~\eqref{eq:c6} hold because $X\mapsto X^{2(1-\alpha)}$ is operator concave since $2(1-\alpha)\in (0,1]$, and $X\mapsto \sqrt{X}$ is operator monotone. 
This completes the proof of~\eqref{eq:joint-concavity}.

We will now prove~\eqref{eq:joint-convexity}.
If $\rho_{AB}\perp \sigma_A\otimes\tau_B$, then $D_\alpha(\rho_{AB}\| \sigma_A\otimes\tau_B)=\infty$, so~\eqref{eq:joint-convexity} is trivially true.
If $\rho_{AB}\perp \sigma_A'\otimes\tau_B'$, then~\eqref{eq:joint-convexity} is trivially true for the same reason.
It remains to prove~\eqref{eq:joint-convexity} for the case where both
$\rho_{AB}\not\perp \sigma_A\otimes\tau_B$ and $\rho_{AB}\not\perp\sigma_A'\otimes \tau_B'$ hold. 
Then
\begin{align}
&\exp\left((\alpha - 1) ( \lambda D_\alpha(\rho_{AB}\| \sigma_A\otimes\tau_B) + \lambda' D_\alpha(\rho_{AB}\| \sigma_A'\otimes\tau_B' ) )\right)
\\
&=\tr[\rho_{AB}^\alpha (\sigma_A\otimes \tau_B)^{1-\alpha}]^\lambda
\tr[\rho_{AB}^\alpha (\sigma_A'\otimes \tau_B')^{1-\alpha}]^{\lambda'}
\label{eq:c1}\\
&\leq \lambda \tr[\rho_{AB}^\alpha (\sigma_A\otimes \tau_B)^{1-\alpha}]+\lambda'\tr[\rho_{AB}^\alpha (\sigma_A'\otimes \tau_B')^{1-\alpha}]
\label{eq:c2}\\
&=\lambda Q_\alpha(\rho_{AB}\| \sigma_A \otimes \tau_B ) + \lambda' Q_\alpha(\rho_{AB}\| \sigma_A' \otimes \tau_B' )
\label{eq:c21}\\
&\leq Q_\alpha(\rho_{AB}\| (\lambda\sigma_A +\lambda'\sigma_A') \otimes (\lambda \tau_B+\lambda' \tau_B') )
\label{eq:c22}\\
&= \exp\left((\alpha - 1)D_\alpha(\rho_{AB}\| (\lambda\sigma_A +\lambda'\sigma_A') \otimes (\lambda \tau_B+\lambda' \tau_B') ) \right).
\end{align}
\eqref{eq:c2} follows from the weighted arithmetic-geometric mean inequality. 
\eqref{eq:c22} follows from~\eqref{eq:joint-concavity}. 
This completes the proof of~\eqref{eq:joint-convexity}.

We will now prove~\eqref{eq:i-convexity}.
If $\rho_A\perp\sigma_A$, then $I_\alpha^{\downarrow}(\rho_{AB}\| \sigma_A)=\infty$, so~\eqref{eq:i-convexity} is trivially true. 
For the same reason,~\eqref{eq:i-convexity} is trivially true if $\rho_A\perp \sigma_A'$.
It remains to prove~\eqref{eq:i-convexity} for the case where both $\rho_A\not\perp\sigma_A$ and $\rho_A\not\perp\sigma_A'$ hold.
Let 
$\hat{\tau}_B\in \argmin_{\tilde{\tau}_B\in \mathcal{S}(B)} D_\alpha (\rho_{AB}\| \sigma_A\otimes \tilde{\tau}_B)$ and let 
$\hat{\tau}'_B\in \argmin_{\tilde{\tau}_B\in \mathcal{S}(B)} D_\alpha (\rho_{AB}\| \sigma_A'\otimes \tilde{\tau}_B)$. 
Then
\begin{align}
\lambda I_\alpha^{\downarrow}(\rho_{AB}\| \sigma_A) +\lambda' I_\alpha^{\downarrow}(\rho_{AB}\| \sigma_A')
&=\lambda D_\alpha (\rho_{AB}\| \sigma_A\otimes \hat{\tau}_B) + \lambda' D_\alpha (\rho_{AB}\| \sigma_A'\otimes \hat{\tau}_B')
\label{eq:i-convex1}\\
&\geq D_\alpha (\rho_{AB}\| (\lambda\sigma_A +\lambda'\sigma_A') \otimes (\lambda\hat{\tau}_B +\lambda'\hat{\tau}_B'))
\label{eq:i-convex2}\\
&\geq I_\alpha^{\downarrow} (\rho_{AB}\| \lambda\sigma_A +\lambda'\sigma_A').
\end{align}
\eqref{eq:i-convex2} follows from~\eqref{eq:joint-convexity}.
This completes the proof of~\eqref{eq:i-convexity}.

We will now prove the assertions below~\eqref{eq:i-convexity}. 
Suppose $\alpha\in (\frac{1}{2},1)$. 

First, suppose in addition that 
$\sigma_A,\sigma_A'\in \mathcal{S}_{\sim\rho_A}(A),
\tau_B,\tau_B'\in \mathcal{S}_{\sim\rho_B}(B)$, 
and that~\eqref{eq:joint-concavity} or~\eqref{eq:joint-convexity} holds with equality. 
The proof above then implies that~\eqref{eq:joint-concavity} holds with equality.
In particular, the inequalities in~\eqref{eq:c5}, \eqref{eq:c51}, and~\eqref{eq:c6} must then hold with equality.
Let us define the following positive semidefinite operators.
\begin{align}
X_B&\coloneqq \sqrt{\lambda \tau_B^{2(1-\alpha)}+\lambda' {\tau_B'}^{2(1-\alpha)}}\\
Y_B&\coloneqq \sqrt{(\lambda\tau_B+\lambda'\tau_B')^{2(1-\alpha)}}\\
Z_B&\coloneqq \tr_A[\rho_{AB}^\alpha \sqrt{(\lambda\sigma_A+\lambda'\sigma'_A)^{2(1-\alpha)}}]
\end{align}
The saturation of the inequality in~\eqref{eq:c6} can then be expressed as $\tr[X_BZ_B]= \tr[Y_BZ_B]$.
Since $\tilde{X}\mapsto \tilde{X}^{2(1-\alpha)}$ is operator concave, $X_B^2\leq Y_B^2$. 
Since $\tilde{X}\mapsto \tilde{X}^{1/2}$ is operator monotone, $X_B\leq Y_B$.
Furthermore, $Y_B\ll \rho_B$ and $\rho_B\ll Z_B$, so $Y_B\ll Z_B$. 
By applying Lemma~\ref{lem:xyz}~(b), it follows that $X_B=Y_B$, i.e.,
\begin{equation}\label{eq:lam-tau}
\lambda\tau_B^{2(1-\alpha)} +\lambda'{\tau'_B}^{2(1-\alpha)}
=\left(\lambda\tau_B + \lambda'\tau_B'\right)^{2(1-\alpha)}.
\end{equation}
An analogous argument can be made to conclude that
\begin{equation}\label{eq:lam-sigma}
\lambda\sigma_A^{2(1-\alpha)} +\lambda'{\sigma'_A}^{2(1-\alpha)}
=\left(\lambda\sigma_A + \lambda'\sigma_A'\right)^{2(1-\alpha)}.
\end{equation}
By applying Lemma~\ref{lem:strict-concavity}, we deduce from~\eqref{eq:lam-tau} and~\eqref{eq:lam-sigma} that 
$\sigma_A=\sigma_A'$ and $\tau_B=\tau_B'$.

Now, suppose $\sigma_A,\sigma_A'\in \mathcal{S}_{\sim\rho_A}(A)$, 
and that~\eqref{eq:i-convexity} holds with equality instead. 
Let $\hat{\tau}_B\in \argmin_{\tilde{\tau}_B\in \mathcal{S}(B)} D_\alpha (\rho_{AB}\| \sigma_A\otimes \tilde{\tau}_B)$ and
let $\hat{\tau}_B'\in \argmin_{\tilde{\tau}_B\in \mathcal{S}(B)} D_\alpha (\rho_{AB}\| \sigma_A'\otimes \tilde{\tau}_B)$. 
By Proposition~\ref{prop:gen-prmi}~(f), $\hat{\tau}_B,\hat{\tau}_B'\in \mathcal{S}_{\sim\rho_B}(B)$. 
Since the inequality in~\eqref{eq:i-convex2} must be saturated, it follows from the saturation of~\eqref{eq:joint-convexity} that $\sigma_A=\sigma_A'$.
\end{proof}

\section{Proof of Theorem~\ref{thm:prmi2}}\label{proof:prmi2}
For the proof of the fixed-point property, Theorem~\ref{thm:prmi2}~(k), we make use of a lemma that asserts a general equivalence of optimizers and fixed-points. 
We begin by stating and proving this lemma, and then proceed to the proof of Theorem~\ref{thm:prmi2}.

\subsection{Lemma for Theorem~\ref{thm:prmi2}~(k)}\label{app:fixed}

In order to prove the following lemma, we will use Fr\'echet derivatives. 
We will now elucidate the notation for these derivatives. 
Consider $\mathcal{B}_A\coloneqq \{X_A\in \mathcal{L}(A):X_A\text{ is self-adjoint} \}$ with $\lVert\cdot \rVert_\infty$ as a Banach space over $\mathbb{R}$. 
Similarly, consider $\mathcal{B}_B\coloneqq \{Y_B\in \mathcal{L}(B):Y_B\text{ is self-adjoint} \}$ with $\lVert\cdot \rVert_\infty$ as a Banach space over $\mathbb{R}$. 
Let $U\subseteq \mathcal{B}_A$ be an open set, and let 
$f:U\rightarrow\mathcal{B}_B$ be a Fr\'echet differentiable function. 
Then we denote the Fr\'echet derivative of $f$ at $X\in U$ by $Df(X)\in \mathcal{L}(\mathcal{B}_A,\mathcal{B}_B)$. 
(Since we are working exclusively with finite-dimensional Hilbert spaces $A$ and $B$, the Fr\'echet derivative of $f$ remains unchanged if the norms of $\mathcal{B}_A$ and $\mathcal{B}_B$ are given by some other Schatten $p$-norm for $p\in [1,\infty)$ due to norm equivalence.) 
For a Fr\'echet differentiable map $f:\mathcal{S}_{>0}(A)\rightarrow\mathcal{B}_B$, we define the directional derivative of $f$ at $\sigma\in \mathcal{S}_{>0}(A)$ in the direction of $\omega\in \mathcal{S}(A)$ as 
$\partial_\omega f(\sigma)\coloneqq Df(\sigma)(\omega-\sigma)$. 
For some basic properties of the directional derivative $\partial_\omega$, we refer the reader to~\cite[Appendix~C1]{hayashi2016correlation}.

The following lemma is an extension of~\cite[Lemma~22]{hayashi2016correlation}. 

\begin{lem}[Equivalence of optimizers and fixed-points]\label{lem:fixed_point}
Let $\alpha\in (0,\infty),\beta\in \mathbb{R},\gamma\in [-1,0)\cup (0,1)$. 
Let $\rho_{AB}\in \mathcal{S}(AB)$ and $\tau_B\in \mathcal{S}_{\not\perp\rho_B}(B)$. 
Let $X_{AB}\coloneqq \tau_B^\beta \rho_{AB}^{\frac{1}{2}}$ and $X_A\coloneqq \tr_B[X_{AB}]$.
Let us define the following functions and sets.
\begin{align}
\chi_{\alpha,\gamma} :&\quad
\mathcal{S}(A) \rightarrow [0,\infty), \, 
\sigma_A\mapsto 
\chi_{\alpha,\gamma}(\sigma_A)
\coloneqq \tr[(X_{AB}^\dagger \sigma_A^\gamma X_{AB})^\alpha ]
\\
\mathcal{X}_{\alpha,\gamma} :&\quad 
\mathcal{S}_{\sim X_A}(A) \rightarrow\mathcal{S}_{\sim X_A}(A), \, 
\sigma_A\mapsto 
\mathcal{X}_{\alpha,\gamma}(\sigma_A)\coloneqq \frac{\tr_B[(\sigma_A^{\frac{\gamma}{2}}X_{AB}X_{AB}^\dagger \sigma_A^{\frac{\gamma}{2}})^\alpha ]}{\tr[(\sigma_A^{\frac{\gamma}{2}}X_{AB}X_{AB}^\dagger \sigma_A^{\frac{\gamma}{2}})^\alpha ]}
\\
\mathcal{F}_{\alpha,\gamma}
&\coloneqq \{\sigma_A\in \mathcal{S}_{\sim X_A}(A): \mathcal{X}_{\alpha,\gamma}(\sigma_A)=\sigma_A\}
\label{eq:def-f-ag}\\
\mathcal{M}_{\alpha,\gamma}
&\coloneqq \begin{cases}
\argmax\limits_{\sigma_A\in \mathcal{S}_{\sim X_A}(A)}\chi_{\alpha,\gamma}(\sigma_A)
\quad\text{ if }\gamma \in (0,1)\\
\argmin\limits_{\sigma_A\in \mathcal{S}_{\sim X_A}(A)}\chi_{\alpha,\gamma}(\sigma_A)
\quad\text{ if }\gamma \in [-1,0)
\end{cases}
\label{eq:def-m-ag}
\end{align}
Then all of the following hold. 
\begin{enumerate}[label=(\alph*)]
\item If $\gamma\in (0,1)$ and $\alpha\in (0,\frac{1}{\gamma}]$, then 
$\mathcal{F}_{\alpha,\gamma}=\mathcal{M}_{\alpha,\gamma}\neq \emptyset$.
\item If $\gamma\in [-1,0)$, then 
$\mathcal{F}_{\alpha,\gamma}=\mathcal{M}_{\alpha,\gamma}\neq \emptyset$.
\end{enumerate}
\end{lem}
\begin{proof}[Proof of (a)] 
$X_A$ is positive semidefinite. 
Without loss of generality, suppose $X_{A}$ is positive definite.
(Otherwise, the same proof works if $A$ is restricted to $\supp(X_A)$.)
Then the sets defined in~\eqref{eq:def-f-ag} and~\eqref{eq:def-m-ag} can be expressed as follows.
\begin{align}
\mathcal{F}_{\alpha,\gamma}
&=\{\sigma_A\in \mathcal{S}_{>0}(A):\mathcal{X}_{\alpha,\gamma}(\sigma_A)=\sigma_A\}
\\
\mathcal{M}_{\alpha,\gamma}
&=\argmax_{\sigma_A\in \mathcal{S}_{>0}(A)}\chi_{\alpha,\gamma}(\sigma_A)
\end{align}
We will now prove that there exists a maximizer for 
\begin{equation}\label{eq:fix-max-sigma}
\max_{\sigma_A\in \mathcal{S}_{>0}(A)}\chi_{\alpha,\gamma}(\sigma_A).
\end{equation}
To this end, we will show that the directional derivative of $\chi_{\alpha,\gamma}(\sigma_A)$ at $\sigma_A\in \mathcal{S}_{>0}(A)$ in the direction of the maximally mixed state $\omega_A$ is strictly positive 
if (at least) one of the eigenvalues of $\sigma_A$ becomes sufficiently small. 
Let $\sigma_A\in \mathcal{S}_{>0}(A)$, let $\omega_A\coloneqq 1_A/d_A\in \mathcal{S}(A)$ and let us define the map 
$f:\mathcal{S}_{>0}(A)\rightarrow \mathcal{L}(A),\sigma_A\mapsto \sigma_A^\gamma$. 
Its Fr\'echet derivative at $\sigma_A$ is 
$Df(\sigma_A)(X_A)=f^{[1]}(\sigma_A)\odot X_A$ for all self-adjoint $X_A\in \mathcal{L}(A)$, where 
$\odot$ denotes the Hadamard product taken in an eigenbasis of $\sigma_A$ and $f^{[1]}$ is the first-order divided difference of $f$~\cite{bhatia2007positive}. 
Then,
\begin{align}
\partial_{\omega_A} \chi_{\alpha,\gamma}(\sigma_A)
&=\alpha  \tr[X_{AB}(X_{AB}^\dagger \sigma_A^\gamma X_{AB})^{\alpha -1}X_{AB}^\dagger
\partial_{\omega_A} f(\sigma_A) ]
\label{eq:d-omega-chi0}\\
&=\alpha \tr[\tr_B[X_{AB}(X_{AB}^\dagger \sigma_A^\gamma X_{AB})^{\alpha -1}X_{AB}^\dagger] 
\partial_{\omega_A} f(\sigma_A) ]
\\
&=\alpha \gamma \tr[\tr_B[X_{AB}(X_{AB}^\dagger \sigma_A^\gamma X_{AB})^{\alpha -1}X_{AB}^\dagger] 
\left(\frac{1}{d_A}\sigma_A^{\gamma-1}-\sigma_A^\gamma \right) ].
\label{eq:d-omega-chi}
\end{align}
\eqref{eq:d-omega-chi} holds because 
\begin{align}
\partial_{\omega_A} f(\sigma_A)
&=Df(\sigma_A)(\omega_A-\sigma_A)
=f^{[1]}(\sigma_A)\odot (\omega_A-\sigma_A)
\\
&=\gamma \sigma_A^{\gamma -1}\odot (\omega_A-\sigma_A)
=\gamma \sigma_A^{\gamma -1} (\omega_A-\sigma_A)
=\gamma \left(\frac{1}{d_A}\sigma_A^{\gamma -1}-\sigma_A^\gamma\right),
\label{eq:proof-fixed0}
\end{align}
where the first two equalities in~\eqref{eq:proof-fixed0} hold because 
$\omega_A-\sigma_A$ is diagonal in any eigenbasis of $\sigma_A$. 
Since $\alpha>0$ and $\gamma>0$, the prefactor in~\eqref{eq:d-omega-chi} is strictly positive. 
The trace in~\eqref{eq:d-omega-chi} is the trace of the product of two self-adjoint operators on $A$, and the first operator is positive definite. 
Now, consider the case where (at least) one of the eigenvalues of $\sigma_A$ is arbitrarily small but non-zero. 
Since $0<\gamma<1$, the corresponding eigenvalue of $\sigma_A^{\gamma -1}$ becomes arbitrarily large, whereas none of the eigenvalues of $\sigma_A^\gamma$ diverges in this limit. 
Consequently, if one of the eigenvalues of $\sigma_A$ is sufficiently small, then 
$\partial_{\omega_A} \chi_{\alpha,\gamma}(\sigma_A)$ is strictly positive. 
Hence, there exists a maximizer for the optimization problem in~\eqref{eq:fix-max-sigma}, so 
$\mathcal{M}_{\alpha,\gamma}\neq \emptyset$.

It remains to prove that $\mathcal{M}_{\alpha,\gamma}=\mathcal{F}_{\alpha,\gamma}$. 
The function $\chi_{\alpha,\gamma}$ is concave \cite[Theorem 2.1(a)]{evert2022convexity} (see also~\cite{epstein1973remarks,carlen2008minkowski,hiai2013concavity}).
Hence, for any $\sigma_A\in \mathcal{S}_{>0}(A)$: $\sigma_A\in \mathcal{M}_{\alpha,\gamma}$ iff $\partial_{\omega_A}\chi_{\alpha,\gamma}(\sigma_A)=0 \, \forall \omega_A\in \mathcal{S}(A)$.
The argument in~\cite[Proof of Lemma~22]{hayashi2016correlation} shows that
the second condition is equivalent to the condition that 
\begin{equation}\label{eq:fix-prop-1}
\sigma_A^{-\frac{1}{2}}
\tr_B[(\sigma_A^{\frac{\gamma}{2}} X_{AB}X_{AB}^\dagger \sigma_A^{\frac{\gamma}{2}})^\alpha ]
\sigma_A^{-\frac{1}{2}}
\end{equation}
is proportional to the identity. 
This is equivalent to $\sigma_A\in \mathcal{F}_{\alpha,\gamma}$.
Therefore, $\mathcal{M}_{\alpha,\gamma}=\mathcal{F}_{\alpha,\gamma}$.
\end{proof}
\begin{proof}[Proof of (b)]
$X_A$ is positive semidefinite. 
Without loss of generality, suppose $X_{A}$ is positive definite.
Then the sets defined in~\eqref{eq:def-f-ag} and~\eqref{eq:def-m-ag} can be expressed as follows.
\begin{align}
\mathcal{F}_{\alpha,\gamma}
&=\{\sigma_A\in \mathcal{S}_{>0}(A):\mathcal{X}_{\alpha,\gamma}(\sigma_A)=\sigma_A\}
\\
\mathcal{M}_{\alpha,\gamma}
&=\argmin_{\sigma_A\in \mathcal{S}_{>0}(A)}\chi_{\alpha,\gamma}(\sigma_A)
\end{align}
Since $\alpha>0,\gamma<0$, and 
$\chi_{\alpha,\gamma}(\sigma_A)=\tr[(X_{AB}^\dagger \sigma_A^\gamma X_{AB})^\alpha ]$,
it is clear that $\chi_{\alpha,\gamma}(\sigma_A)$ diverges to $+\infty$ 
if (at least) one of the eigenvalues of $\sigma_A$ becomes arbitrarily small. 
Hence, there exists a minimizer for 
$\min_{\sigma_A\in \mathcal{S}_{>0}(A)}\chi_{\alpha,\gamma}(\sigma_A)$, so 
$\mathcal{M}_{\alpha,\gamma}\neq \emptyset$.

It remains to prove that $\mathcal{M}_{\alpha,\gamma}=\mathcal{F}_{\alpha,\gamma}$.
The function $\chi_{\alpha,\gamma}$ is convex~\cite[Theorem 2.1(b)]{evert2022convexity} (see also~\cite{hiai2013concavity}).
Hence, for any $\sigma_A\in \mathcal{S}_{>0}(A)$: $\sigma_A\in \mathcal{M}_{\alpha,\gamma}$ iff 
$\partial_{\omega_A}\chi_{\alpha,\gamma}(\sigma_A)=0 \, \forall \omega_A\in \mathcal{S}(A)$.
The argument in~\cite[Proof of Lemma~22]{hayashi2016correlation} shows that
the second condition is equivalent to the condition that the operator in~\eqref{eq:fix-prop-1}
is proportional to the identity. 
Therefore, $\mathcal{M}_{\alpha,\gamma}=\mathcal{F}_{\alpha,\gamma}$.
\end{proof}

\subsection{Proof of Theorem~\ref{thm:prmi2}}\label{app:proof-prmi2}

\begin{proof}[Proof of (a)]
This assertion follows from the symmetry of the definition of the doubly minimized PRMI in~\eqref{eq:prmi2} with respect to $A$ and $B$.
\end{proof}
\begin{proof}[Proof of (b), (e), (f), (h), (i), (n), (r)]
Since $I_\alpha^{\downarrow\downarrow}(A:B)_\rho=\inf_{\sigma_A\in \mathcal{S}(A)}I_\alpha^{\downarrow}(\rho_{AB}\| \sigma_A)$, 
these properties follow from corresponding properties of the minimized generalized PRMI, see Proposition~\ref{prop:gen-prmi}. 
In particular, (h) and~(i) follow from Proposition~\ref{prop:gen-prmi}~(f). 
\end{proof}
\begin{proof}[Proof of (c)]
\begin{align}
I_\alpha^{\downarrow\downarrow}(A':B')_{V\otimes W\rho_{AB}V^\dagger\otimes W^\dagger}
&=\inf_{\sigma_{A'}\in \mathcal{S}(A')}I_\alpha^\downarrow (V\otimes W\rho_{AB}V^\dagger\otimes W^\dagger\| \sigma_{A'})
\\
&=\inf_{\sigma_{A'}\in \mathcal{S}(A')}I_\alpha^\downarrow (V\rho_{AB}V^\dagger \| \sigma_{A'})
=I_\alpha^{\downarrow\downarrow}(A':B)_{V\rho_{AB}V^\dagger}
\label{eq:iso-dd1}\\
&=\inf_{\tau_{B}\in \mathcal{S}(B)}I_\alpha^\downarrow (V\rho_{AB}V^\dagger \| \tau_{B})
\\
&=\inf_{\tau_{B}\in \mathcal{S}(B)}I_\alpha^\downarrow (\rho_{AB}\| \tau_{B})
=I_\alpha^{\downarrow\downarrow}(A:B)_{\rho_{AB}}
\label{eq:iso-dd2}
\end{align}
Above, we have used the invariance of the minimized generalized PRMI under local isometries, see Proposition~\ref{prop:gen-prmi}~(b), twice:
for the first equality in~\eqref{eq:iso-dd1},
and for the first equality in~\eqref{eq:iso-dd2}.
\end{proof}
\begin{proof}[Proof of (j)] 
If $\alpha=1$, then the assertion is true due to~\eqref{eq:i1-argmin}. 
It remains to prove~(j) for $\alpha\in (\frac{1}{2},1)$.

Let $\rho_{AB}\in \mathcal{S}(AB),\alpha\in (\frac{1}{2},1)$ be fixed. 
Let $\mathcal{M}_\alpha\coloneqq \argmin_{\sigma_A\in \mathcal{S}(A)}I_\alpha^\downarrow (\rho_{AB}\| \sigma_A)$. 
We will now prove that $\mathcal{M}_\alpha\subseteq \mathcal{S}_{\sim\rho_A}(A)$ by cases.

Case 1: $\rho_{AB}=\rho_A\otimes \rho_B$. 
By the positive definiteness of the Petz divergence, 
$\mathcal{M}_\alpha=\{\rho_A\}$. 
Therefore, $\mathcal{M}_\alpha\subseteq \mathcal{S}_{\sim \rho_A}(A)$. 

Case 2: $\rho_{AB}\neq\rho_A\otimes \rho_B$. 
For this case, we will prove that
\begin{align}\label{eq:proof-case3}
\argmin_{({\sigma}_A,{\tau}_B)\in \mathcal{S}(A)\times \mathcal{S}(B) }D_\alpha (\rho_{AB}\| {\sigma}_A\otimes {\tau}_B)
=\argmin_{({\sigma}_A,{\tau}_B)\in \mathcal{S}_{\sim \rho_A}(A)\times \mathcal{S}_{\sim \rho_B}(B) }D_\alpha (\rho_{AB}\| {\sigma}_A\otimes {\tau}_B).
\end{align}
Clearly, the assertion that $\mathcal{M}_\alpha\subseteq \mathcal{S}_{\sim \rho_A}(A)$ then follows from~\eqref{eq:proof-case3}. 
In order to prove~\eqref{eq:proof-case3}, let  
\begin{equation}\label{eq:proof-support}
(\hat{\sigma}_A,\hat{\tau}_B)\in 
\argmin_{(\sigma_A,\tau_B)\in \mathcal{S}(A)\times \mathcal{S}(B)} D_\alpha (\rho_{AB}\| \sigma_A\otimes \tau_B).
\end{equation}
By Proposition~\ref{prop:gen-prmi}~(f), $\hat{\sigma}_A\ll \rho_A$ and $\hat{\tau}_B\ll \rho_B$.
It remains to show that $\rho_A\ll\hat{\sigma}_A$ and $\rho_B\ll\hat{\tau}_B$.
Note that if $\rho_A\ll\hat{\sigma}_A$, then $\rho_B\ll\hat{\tau}_B$ by Proposition~\ref{prop:gen-prmi}~(f). 
So it suffices to show that $\rho_A\ll\hat{\sigma}_A$. 
We will prove this by contradiction.

Suppose it is false that $\rho_A\ll\hat{\sigma}_A$. 
Let $\sigma_A\in \mathcal{S}(A)$ be the quantum state that is proportional to the orthogonal projection onto $\ker(\hat{\sigma}_A)$. 
Let $\tau_B\in \mathcal{S}(B)$, to be specified later.
Let 
\begin{align}
\hat{f}\coloneqq Q_\alpha(\rho_{AB}\| \hat{\sigma}_A\otimes \hat{\tau}_B )=\max_{(\sigma_A',\tau_B')\in \mathcal{S}(A)\times\mathcal{S}(B)}Q_\alpha(\rho_{AB}\| {\sigma}_A'\otimes {\tau}_B' )>0.
\end{align}
Let us define the following two functions of $\lambda\in [0,1]$.
\begin{align}
f_1(\lambda)&\coloneqq  Q_\alpha(\rho_{AB}\| ((1-\lambda)\hat{\sigma}_A+\lambda\sigma_A)\otimes ((1-\lambda)\hat{\tau}_B+\lambda\tau_B) )
\\
f_2(\lambda)
&\coloneqq \hat{f}-(2-\alpha)\lambda \hat{f}
+ \lambda^{1-\alpha}Q_\alpha (\rho_{AB}\| {\sigma}_A\otimes ((1-\lambda)\hat{\tau}_B+\lambda\tau_B) )
\label{eq:proof-def-f3}
\end{align}
Then, for all $\lambda\in [0,1]$
\begin{subequations}\label{eq:proof-fmax>f3}
\begin{align}
\hat{f}
&\geq f_1(\lambda)
\label{eq:proof-supp0}\\
&=(1-\lambda)^{1-\alpha} Q_\alpha(\rho_{AB}\| \hat{\sigma}_A\otimes ((1-\lambda)\hat{\tau}_B+\lambda\tau_B) )
+\lambda^{1-\alpha} Q_\alpha(\rho_{AB}\| \sigma_A\otimes ((1-\lambda)\hat{\tau}_B+\lambda\tau_B) )
\label{eq:proof-supp00}\\
&\geq (1-\lambda)^{1-\alpha}( (1-\lambda) \hat{f}
+  \lambda Q_\alpha(\rho_{AB}\| \hat{\sigma}_A\otimes \tau_B ) )
+\lambda^{1-\alpha} Q_\alpha(\rho_{AB}\| \sigma_A\otimes ((1-\lambda)\hat{\tau}_B+\lambda\tau_B) )
\label{eq:proof-supp1}\\
&\geq 
(1-\lambda)^{2-\alpha}\hat{f}
+\lambda^{1-\alpha} Q_\alpha (\rho_{AB}\| {\sigma}_A\otimes ((1-\lambda)\hat{\tau}_B+\lambda\tau_B) )
\label{eq:proof-supp2}\\
&\geq 
(1-(2-\alpha)\lambda)\hat{f}
+\lambda^{1-\alpha} Q_\alpha (\rho_{AB}\| {\sigma}_A\otimes ((1-\lambda)\hat{\tau}_B+\lambda\tau_B) )
=f_2(\lambda).
\label{eq:proof-supp3}
\end{align}
\end{subequations}
\eqref{eq:proof-supp0} follows from~\eqref{eq:proof-support}. 
\eqref{eq:proof-supp00} holds because $\hat{\sigma}_A\perp\sigma_A$ implies that 
$((1-\lambda) \hat{\sigma}_A+\lambda \sigma_A)^{1-\alpha}=((1-\lambda) \hat{\sigma}_A)^{1-\alpha}+(\lambda \sigma_A)^{1-\alpha}$.
\eqref{eq:proof-supp1} follows from the concavity of $Q_\alpha(\rho_{AB}\| \cdot )$~\cite{tomamichel2016quantum}. 
\eqref{eq:proof-supp3} holds because $(1-\lambda)^{2-\alpha}\geq 1-(2-\alpha)\lambda$ for all $\lambda\in [0,1]$.

Case 2.1: $Q_\alpha (\rho_{AB}\| \sigma_A\otimes \hat{\tau}_B)>0$. 
Then we define $\tau_B\coloneqq \hat{\tau}_B$. 
By~\eqref{eq:proof-def-f3}, 
\begin{align}\label{eq:proof-supp31}
f_2(\lambda)
&=\hat{f}+\lambda (-(2-\alpha)\hat{f}+\lambda^{-\alpha}Q_\alpha(\rho_{AB}\| \sigma_A\otimes \hat{\tau}_B) )
\end{align}
for all $\lambda\in [0,1]$. 
Let $\lambda_0\coloneqq
( \frac{Q_\alpha(\rho_{AB}\| \sigma_A\otimes \hat{\tau}_B) }{(2-\alpha)\hat{f}} )^{\frac{1}{\alpha}}>0$. 
Then $f_2(\lambda_0)=\hat{f}$. 
Since $-\alpha<0$, it follows from~\eqref{eq:proof-supp31} that 
$f_2(\lambda)>\hat{f}$ for all $\lambda\in (0,\lambda_0)$.
This contradicts~\eqref{eq:proof-fmax>f3}.

Case 2.2: $Q_\alpha (\rho_{AB}\| \sigma_A\otimes \hat{\tau}_B)=0$. 
Then we define $\tau_B$ as the quantum state which is proportional to the orthogonal projection onto $\ker(\hat{\tau}_B)$. 
Then, for all $\lambda\in [0,1]$
\begin{subequations}\label{eq:proof-supp<}
\begin{align}
Q_\alpha(\rho_{AB}\| \sigma_A\otimes ((1-\lambda)\hat{\tau}_B+\lambda\tau_B))
&=\lambda^{1-\alpha}Q_\alpha(\rho_{AB}\| \sigma_A\otimes \tau_B)+(1-\lambda)^{1-\alpha} Q_\alpha (\rho_{AB}\| \sigma_A\otimes \hat{\tau}_B)
\label{eq:proof-supp<1}\\
&=\lambda^{1-\alpha}Q_\alpha(\rho_{AB}\| \sigma_A\otimes \tau_B).
\label{eq:proof-supp<2}
\end{align}
\end{subequations}
The left-hand side of~\eqref{eq:proof-supp<1} is non-zero for any $\lambda\in (0,1)$ because $\sigma_A\ll \rho_A$ and $(1-\lambda)\hat{\tau}_B+\lambda \tau_B$ has full rank. 
Hence, the expression in~\eqref{eq:proof-supp<2} must be non-zero for any $\lambda\in (0,1)$, which implies that $Q_\alpha(\rho_{AB}\| \sigma_A\otimes \tau_B)>0$.
The combination of~\eqref{eq:proof-supp<} and~\eqref{eq:proof-def-f3} implies that 
\begin{align}\label{eq:proof-supp32}
f_2(\lambda)
= \hat{f}+\lambda (
-(2-\alpha)\hat{f} +\lambda^{1-2\alpha}Q_\alpha (\rho_{AB}\| \sigma_A\otimes \tau_B)
)
\qquad \forall \lambda\in [0,1].
\end{align}
Let $\lambda_0\coloneqq (\frac{Q_\alpha(\rho_{AB}\| \sigma_A\otimes \tau_B)}{(2-\alpha)\hat{f}} )^{\frac{1}{2\alpha-1}}>0$. 
Then $f_2(\lambda_0)=\hat{f}$. 
Since $1-2\alpha <0$, it follows from~\eqref{eq:proof-supp32} that 
$f_2(\lambda)>\hat{f}$ for all $\lambda\in (0,\lambda_0)$.
This contradicts~\eqref{eq:proof-fmax>f3}.
This completes the proof of the inclusion~$\mathcal{M}_\alpha\subseteq \mathcal{S}_{\sim \rho_A}(A)$.

We will now prove that $\mathcal{M}_\alpha$ cannot contain more than one element. 
Let $\sigma_A,\sigma_A'\in \mathcal{M}_\alpha$. 
Let $\lambda\coloneqq \frac{1}{2},\lambda'\coloneqq \frac{1}{2}$. 
Then, \eqref{eq:i-convexity} in Theorem~\ref{thm:convexity} implies that also 
$\lambda \sigma_A+\lambda'\sigma_A'\in \mathcal{M}_\alpha$. 
We have 
\begin{align}
I_\alpha^{\downarrow\downarrow}(A:B)_\rho 
&=I_\alpha^\downarrow(\rho_{AB}\| \lambda \sigma_A+\lambda'\sigma_A') 
=\lambda I_\alpha^\downarrow(\rho_{AB}\| \sigma_A)+\lambda' I_\alpha^\downarrow(\rho_{AB}\| \sigma_A').
\end{align}
Hence, the inequality in~\eqref{eq:i-convexity} is saturated. 
By Theorem~\ref{thm:convexity}, $\sigma_A=\sigma_A'$. 
This proves that $\mathcal{M}_\alpha$ contains at most one element. 
We conclude that there exists $\hat{\sigma}_A\in \mathcal{S}_{\sim\rho_A}(A)$ such that $\mathcal{M}_\alpha=\{\hat{\sigma}_A\}$.
The assertion now follows from Proposition~\ref{prop:gen-prmi}~(f).
\end{proof}
\begin{proof}[Proof of (k)]
Let $\alpha\in (\frac{1}{2},2]$.

Case 1: $\alpha=1$. 
Then $\mathcal{M}_\alpha=\{\rho_A\}$ due to~\eqref{eq:min-re-t}.
For any $\sigma_A\in \mathcal{S}_{\sim\rho_A}(A)$: $\sigma_A\in \mathcal{F}_\alpha$ iff 
$\sigma_A=\tr_C[\rho_{AC}]$.
Hence, $\mathcal{F}_\alpha=\{\rho_A\}$. 
Therefore, $\mathcal{M}_\alpha=\mathcal{F}_\alpha$.

Case 2: $\alpha\in (\frac{1}{2},1)$.
Let $\beta\coloneqq \frac{1}{\alpha}\in (1,2)$ and $\gamma\coloneqq 1-\alpha =\frac{\beta -1}{\beta}\in (0,\frac{1}{2})$. 
Let $X_{AC}\coloneqq \rho_C^{\frac{1-\beta}{2\beta}}\rho_{AC}^{\frac{1}{2}}$ and 
$X_A\coloneqq\tr_C[X_{AC}]$. 
Then $\supp(X_A)=\supp (\rho_A)$.
By~(j) and the duality of the minimized generalized PRMI, see Proposition~\ref{prop:gen-prmi}~(d), 
\begin{align}
\mathcal{M}_\alpha
=\argmax_{\sigma_A\in \mathcal{S}_{\sim\rho_A}(A)} 
\widetilde{Q}_\beta (\rho_{AC}\| \sigma_A^{-1}\otimes \rho_C)
=\argmax_{\sigma_A\in \mathcal{S}_{\sim X_A}(A)} 
\tr[(X_{AC}^\dagger \sigma_A^\gamma X_{AC} )^\beta ].
\end{align}
Lemma~\ref{lem:fixed_point}~(a) implies that $\mathcal{M}_\alpha=\mathcal{F}_\alpha$ since $\gamma\in (0,1)$ and $0<\beta\leq 2\leq \frac{1}{\gamma}$.

Case 3: $\alpha\in (1,2]$.
Let $\beta\coloneqq \frac{1}{\alpha}\in [\frac{1}{2},1)$ and $\gamma\coloneqq 1-\alpha =\frac{\beta -1}{\beta}\in [-1,0)$. 
Let $X_{AC}\coloneqq \rho_C^{\frac{1-\beta}{2\beta}}\rho_{AC}^{\frac{1}{2}}$ and 
$X_A\coloneqq\tr_C[X_{AC}]$. 
Then $\supp(X_A)=\supp (\rho_A)$.
By~(h) and the duality of the minimized generalized PRMI, see Proposition~\ref{prop:gen-prmi}~(d), 
\begin{align}
\mathcal{M}_\alpha
=\argmin_{\sigma_A\in \mathcal{S}_{\sim\rho_A}(A)} 
\widetilde{Q}_\beta (\rho_{AC}\| \sigma_A^{-1}\otimes \rho_C)
=\argmin_{\sigma_A\in \mathcal{S}_{\sim X_A}(A)} 
\tr[(X_{AC}^\dagger \sigma_A^\gamma X_{AC} )^\beta ].
\end{align}
Lemma~\ref{lem:fixed_point}~(b) implies that $\mathcal{M}_\alpha=\mathcal{F}_\alpha$ since $\gamma\in [-1,0)$ and $0<\beta$.
\end{proof}

\begin{proof}[Proof of (d)]
Let $\alpha\in [\frac{1}{2},2]$.

Case 1: $\alpha\in (\frac{1}{2},2]$. 
Let $\ket{\rho}_{ABC}\in ABC$ be such that $\tr_C[\proj{\rho}_{ABC}]=\rho_{AB}$, and let 
$\ket{\rho'}_{DEF}\in DEF$ be such that $\tr_F[\proj{\rho'}_{DEF}]=\rho'_{DE}$.

Let $\sigma_A\in \argmin_{\tilde{\sigma}_A\in \mathcal{S}(A)}I_\alpha^\downarrow (\rho_{AB}\| \tilde{\sigma}_A)$.
According to~(k), $\sigma_A\in \mathcal{S}_{\sim\rho_A}(A)$ and 
\begin{align}\label{eq:fix1}
\sigma_A=\frac{\tr_C[(\sigma_A^{\frac{1-\alpha}{2}} \otimes \rho_C^{\frac{\alpha-1}{2}}\rho_{AC} \sigma_A^{\frac{1-\alpha}{2}} \otimes \rho_C^{\frac{\alpha-1}{2}})^{\frac{1}{\alpha}}]}{\tr[ (\sigma_A^{\frac{1-\alpha}{2}} \otimes \rho_C^{\frac{\alpha-1}{2}}\rho_{AC} \sigma_A^{\frac{1-\alpha}{2}} \otimes \rho_C^{\frac{\alpha-1}{2}})^{\frac{1}{\alpha}} ]}.
\end{align}
Let $\sigma'_D\in \argmin_{\tilde{\sigma}'_D\in \mathcal{S}(D)}I_\alpha^\downarrow (\rho'_{DE}\| \tilde{\sigma}'_D)$. 
According to~(k), $\sigma'_{D}\in \mathcal{S}_{\sim\rho'_D}(D)$ and 
\begin{align}\label{eq:fix2}
{\sigma'_D}=\frac{\tr_F[({\sigma'_D}^{\frac{1-\alpha}{2}} \otimes {\rho'_F}^{\frac{\alpha-1}{2}}\rho'_{DF} {\sigma'_D}^{\frac{1-\alpha}{2}} \otimes {\rho'_F}^{\frac{\alpha-1}{2}})^{\frac{1}{\alpha}}]}{\tr[ ({\sigma'_D}^{\frac{1-\alpha}{2}} \otimes {\rho'_F}^{\frac{\alpha-1}{2}}\rho'_{DF} {\sigma'_D}^{\frac{1-\alpha}{2}} \otimes {\rho'_F}^{\frac{\alpha-1}{2}})^{\frac{1}{\alpha}} ]}.
\end{align}
Clearly, we have $\sigma_A\otimes\sigma'_{D}\in \mathcal{S}_{\sim\rho_A\otimes\rho'_D}(AD)$, 
and the combination of~\eqref{eq:fix1} and~\eqref{eq:fix2} implies that 
\begin{equation}\label{eq:fix3}
\sigma_A\otimes {\sigma'_D}
=\frac{\tr_{CF}[((\sigma_A\otimes {\sigma'_D})^{\frac{1-\alpha}{2}} \otimes (\rho_C\otimes {\rho'_F})^{\frac{\alpha-1}{2}} \rho_{AC}\otimes \rho'_{DF} 
(\sigma_A\otimes {\sigma'_D})^{\frac{1-\alpha}{2}} \otimes (\rho_C\otimes {\rho'_F})^{\frac{\alpha-1}{2}})^{\frac{1}{\alpha}}]}{\tr[ ((\sigma_A\otimes {\sigma'_D})^{\frac{1-\alpha}{2}} \otimes (\rho_C\otimes {\rho'_F})^{\frac{\alpha-1}{2}} \rho_{AC}\otimes \rho'_{DF} 
(\sigma_A\otimes {\sigma'_D})^{\frac{1-\alpha}{2}} \otimes (\rho_C\otimes {\rho'_F})^{\frac{\alpha-1}{2}})^{\frac{1}{\alpha}} ]}.
\end{equation}
We can deduce from~\eqref{eq:fix3} that
$\sigma_A\otimes \sigma'_D \in \argmin_{\tilde{\sigma}_{AD}\in \mathcal{S}(AD)} I_\alpha^{\downarrow}(\rho_{AB}\otimes \rho'_{DE}\| \tilde{\sigma}_{AD} )$ due to~(k). Hence,
\begin{align}
I_\alpha^{\downarrow\downarrow}(AD:BE)_{\rho_{AB}\otimes \rho'_{DE}}
&=I_\alpha^{\downarrow}(\rho_{AB}\otimes \rho'_{DE}\| \sigma_{A}\otimes \sigma'_D )
\\
&=I_\alpha^{\downarrow}(\rho_{AB}\| \sigma_{A})
+I_\alpha^{\downarrow}( \rho'_{DE}\| \sigma'_D )
\label{eq:proof-prmi2-add1}\\
&= I_\alpha^{\downarrow\downarrow}(A:B)_{\rho_{AB}}+I_\alpha^{\downarrow\downarrow}(D:E)_{\rho'_{DE}}.
\label{eq:proof-add}
\end{align}
\eqref{eq:proof-prmi2-add1} follows from the additivity of the minimized generalized PRMI, see Proposition~\ref{prop:gen-prmi}~(c).

Case 2: $\alpha=\frac{1}{2}$. 
Then the assertion follows from case 1 by taking the limit $\alpha\rightarrow 1/2^+$ due to the continuity in $\alpha$ of the Petz divergence.
\end{proof}

\begin{proof}[Proof of (l)]
Let $\alpha \in [0,2]$. Then, for any $n\in \mathbb{N}_{>0}$
\begin{subequations}\label{eq:reg-geq}
\begin{align}
I_\alpha^{\downarrow\downarrow}(A:B)_\rho
&=\inf_{\substack{\sigma_A\in \mathcal{S}(A),\\ \tau_B\in \mathcal{S}(B)}} 
\frac{1}{n}D_\alpha (\rho_{AB}^{\otimes n}\| \sigma_A^{\otimes n}\otimes \tau_B^{\otimes n})
\label{eq:r0}\\
&\geq \inf_{\substack{\sigma_{A^n}\in \mathcal{S}_{\sym}(A^{\otimes n}),\\ \tau_{B^n}\in \mathcal{S}_{\sym}(B^{\otimes n}) }} 
\frac{1}{n}D_\alpha (\rho_{AB}^{\otimes n}\| \sigma_{A^n}\otimes \tau_{B^n})
\label{eq:r1}\\
&\geq \frac{1}{n}D_\alpha (\rho_{AB}^{\otimes n}\| \omega_{A^n}^n \otimes \omega_{B^n}^n )
- \frac{\log g_{n,d_A}}{n}- \frac{\log g_{n,d_B}}{n}
\label{eq:r2}\\
&\geq \inf_{\substack{\sigma_{A^n}\in \mathcal{S}_{\sym}(A^{\otimes n}),\\ \tau_{B^n}\in \mathcal{S}_{\sym}(B^{\otimes n}) }} 
\frac{1}{n}D_\alpha (\rho_{AB}^{\otimes n}\| \sigma_{A^n}\otimes \tau_{B^n})
- \frac{\log g_{n,d_A}}{n}- \frac{\log g_{n,d_B}}{n}.
\label{eq:r21}
\end{align}
\end{subequations}
\eqref{eq:r0} follows from the additivity of the Petz divergence.
\eqref{eq:r2} follows from Remark~\ref{rem:universal-state}~(b). 
\eqref{eq:r21} follows from Remark~\ref{rem:universal-state}~(a).
In the limit $n\rightarrow\infty$, the two terms on the right-hand side of~\eqref{eq:r2} vanish due to Remark~\ref{rem:universal-state}~(b). 
Hence, $I_\alpha^{\downarrow\downarrow}(A:B)_\rho
\geq \limsup_{n\rightarrow\infty}\frac{1}{n}D_\alpha (\rho_{AB}^{\otimes n}\| \omega_{A^n}^n \otimes \omega_{B^n}^n )$.

Case 1: $\alpha\in [\frac{1}{2},2]$. Then, for any $n\in \mathbb{N}_{>0}$
\begin{align}
I_\alpha^{\downarrow\downarrow}(A:B)_\rho
&\geq \inf_{\substack{\sigma_{A^n}\in \mathcal{S}_{\sym}(A^{\otimes n}),\\ \tau_{B^n}\in \mathcal{S}_{\sym}(B^{\otimes n}) }} 
\frac{1}{n}D_\alpha (\rho_{AB}^{\otimes n}\| \sigma_{A^n}\otimes \tau_{B^n})
\label{eq:r00}\\
&\geq \inf_{\substack{\sigma_{A^n}\in \mathcal{S}(A^n),\\ \tau_{B^n}\in \mathcal{S}(B^n) }} 
\frac{1}{n}D_\alpha (\rho_{AB}^{\otimes n}\| \sigma_{A^n}\otimes \tau_{B^n})
=\frac{1}{n} I_\alpha^{\downarrow\downarrow}(A^n:B^n)_{\rho^{\otimes n}}
=I_\alpha^{\downarrow\downarrow}(A:B)_\rho.
\label{eq:r01}
\end{align}
\eqref{eq:r00} follows from~\eqref{eq:r1}. 
\eqref{eq:r01} follows from additivity~(d). 
This proves the assertion in~\eqref{eq:prmi2-omega2}. 
Combining this result with~\eqref{eq:reg-geq} and taking the limit $n\rightarrow\infty$ proves the assertion in~\eqref{eq:prmi2-omega1} for $\alpha\in [\frac{1}{2},2]$. 

Case 2: $\alpha \in (0,\frac{1}{2})$. 
Let $A'$ and $B'$ be Hilbert spaces isomorphic to $A$ and $B$, respectively. 
Then, for any $n\in \mathbb{N}_{>0}$
\begin{align}
&\exp((\alpha-1)D_\alpha(\rho_{AB}^{\otimes n}\| \omega_{A^n}^n\otimes \omega_{B^n}^n))
\label{eq:rc0}\\
&=\tr[(\rho_{AB}^{\otimes n})^\alpha (\omega^n_{A^n}\otimes \omega^n_{B^n})^{1-\alpha}]
\label{eq:rc1}\\
&=\tr[(\rho_{AB}^{\otimes n})^\alpha \omega^n_{A^n}\otimes \omega^n_{B^n} (\omega^n_{A^n}\otimes \omega^n_{B^n})^{-\alpha}]
\label{eq:rc2}\\
&=\smashoperator[r]{\int\limits_{\mathcal{U}(AA')}} \mathrm{d}\mu_H(U)
\smashoperator[r]{\int\limits_{\mathcal{U}(BB')}} \mathrm{d}\mu_H(V)
\tr[(\rho_{AB}^{\otimes n})^\alpha \sigma(U)_A^{\otimes n}\otimes \tau(V)_B^{\otimes n} (\omega^n_{A^n}\otimes \omega^n_{B^n})^{-\alpha}]
\label{eq:rc3}\\
&=\smashoperator[r]{\int\limits_{\mathcal{U}(AA')}} \mathrm{d}\mu_H(U)
\smashoperator[r]{\int\limits_{\mathcal{U}(BB')}} \mathrm{d}\mu_H(V)
\tr[(\sigma(U)_A^{\otimes n}\otimes \tau(V)_B^{\otimes n})^{\frac{1}{2}} (\rho_{AB}^{\otimes n})^\alpha (\sigma(U)_A^{\otimes n}\otimes \tau(V)_B^{\otimes n})^{\frac{1}{2}} (\omega^n_{A^n}\otimes \omega^n_{B^n})^{-\alpha}]
\label{eq:rc4}\\
&\leq g_{n,d_A}^\alpha g_{n,d_B}^\alpha
\smashoperator[r]{\int\limits_{\mathcal{U}(AA')}} \mathrm{d}\mu_H(U)
\smashoperator[r]{\int\limits_{\mathcal{U}(BB')}} \mathrm{d}\mu_H(V)
\tr\left[(\sigma(U)_A^{\otimes n}\otimes \tau(V)_B^{\otimes n})^{\frac{1}{2}} (\rho_{AB}^{\otimes n})^\alpha (\sigma(U)_A^{\otimes n}\otimes \tau(V)_B^{\otimes n})^{\frac{1}{2}-\alpha}\right]
\label{eq:rc5}\\
&=g_{n,d_A}^\alpha g_{n,d_B}^\alpha 
\smashoperator[r]{\int\limits_{\mathcal{U}(AA')}} \mathrm{d}\mu_H(U)
\smashoperator[r]{\int\limits_{\mathcal{U}(BB')}} \mathrm{d}\mu_H(V)
(\tr[ (\rho_{AB})^\alpha (\sigma(U)_A\otimes \tau(V)_B)^{1-\alpha} ])^n 
\label{eq:rc6}\\
&\leq g_{n,d_A}^\alpha g_{n,d_B}^\alpha 
\sup_{\sigma_{A}\in \mathcal{S}(A)}\sup_{\tau_{B}\in \mathcal{S}(B)} (\tr[ (\rho_{AB})^\alpha (\sigma_A\otimes \tau_B)^{1-\alpha} ])^n
\label{eq:rc7}\\
&=g_{n,d_A}^\alpha g_{n,d_B}^\alpha \exp((\alpha-1)nI_\alpha^{\downarrow\downarrow}(A:B)_\rho).
\end{align}
\eqref{eq:rc3} follows from Remark~\ref{rem:universal-state}~(d). 
\eqref{eq:rc4} follows from Remark~\ref{rem:universal-state}~(c). 
\eqref{eq:rc5} follows from Remark~\ref{rem:universal-state}~(b) and 
the operator anti-monotonicity of $X\mapsto X^{-\alpha}$ for $\alpha\in (0,1)$. 
Hence, for any $n\in \mathbb{N}_{>0}$ 
\begin{align}
\frac{1}{n}D_\alpha(\rho_{AB}^{\otimes n}\| \omega_{A^n}^n\otimes \omega_{B^n}^n)
&\geq I_\alpha^{\downarrow\downarrow}(A:B)_\rho - \frac{\alpha}{1-\alpha} \left(\frac{\log g_{n,d_A}}{n}+ \frac{\log g_{n,d_B}}{n}\right).
\label{eq:rc10}
\end{align}
By Remark~\ref{rem:universal-state}~(b), the second term on the right-hand side vanishes in the limit $n\rightarrow\infty$. Therefore,
\begin{equation}\label{eq:reg-liminf}
\liminf_{n\rightarrow\infty}\frac{1}{n}D_\alpha (\rho_{AB}^{\otimes n}\| \omega_{A^n}^n\otimes \omega_{B^n}^n)
\geq I_\alpha^{\downarrow\downarrow}(A:B)_\rho.
\end{equation}

Case 3: $\alpha=0$. 
By monotonicity in $\alpha$~(n), we have for any $n\in \mathbb{N}_{>0}$
\begin{equation}
I_0^{\downarrow\downarrow}(A:B)_\rho\leq \liminf_{\alpha\rightarrow 0^+}I_\alpha^{\downarrow\downarrow}(A:B)_\rho
\leq \liminf_{\alpha\rightarrow 0^+} \frac{1}{n}D_\alpha (\rho_{AB}^{\otimes n}\| \omega_{A^n}^n\otimes \omega_{B^n}^n)
= \frac{1}{n}D_0(\rho_{AB}^{\otimes n}\| \omega_{A^n}^n\otimes \omega_{B^n}^n),
\end{equation}
where we have used~\eqref{eq:rc10} and the continuity in $\alpha$ of the Petz divergence, see Remark~\ref{rem:petz-divergence}. 
Therefore,~\eqref{eq:reg-liminf} also holds for $\alpha=0$.
\end{proof}
\begin{proof}[Proof of (m)]
$I_1^{\downarrow\downarrow}(A:B)_\rho=I(A:B)_\rho$ follows from~\eqref{eq:i1}. 
Furthermore, we have 
\begin{align}
I_0^{\downarrow\downarrow}(A:B)_\rho
&=\inf_{\sigma_A\in \mathcal{S}(A)} I_0^{\downarrow}(\rho_{AB}\| \sigma_A)
=\inf_{\sigma_A\in \mathcal{S}(A)}\inf_{\substack{\ket{\tau}_B\in \supp(\rho_B):\\ \brak{\tau}_B=1 }} D_0(\rho_{AB}\| \sigma_A\otimes\proj{\tau}_B)
\label{eq:proof-prmi2-0}\\
&=\inf_{\substack{\ket{\tau}_B\in \supp(\rho_B):\\ \brak{\tau}_B=1 }}I_0^\downarrow (\rho_{AB}\| \proj{\tau}_B)
\\
&=\inf_{\substack{\ket{\tau}_B\in \supp(\rho_B):\\ \brak{\tau}_B=1 }}\inf_{\substack{\ket{\sigma}_A\in \supp(\rho_A):\\ \brak{\sigma}_A=1 }} D_0(\rho_{AB}\| \proj{\sigma}_A\otimes\proj{\tau}_B).
\label{eq:proof-prmi2-00}
\end{align}
Above, we have used Proposition~\ref{prop:gen-prmi}~(i) twice: for the second equality in~\eqref{eq:proof-prmi2-0}, and for~\eqref{eq:proof-prmi2-00}.
\end{proof}
\begin{proof}[Proof of (o)]
The continuity of $I_\alpha^{\downarrow\downarrow}(A:B)_\rho$ on $\alpha\in [0,1)$ and on $\alpha\in [1,\infty)$ follows from the continuity in $\alpha$ of the Petz divergence. 
It remains to prove left-continuity at $\alpha=1$. 
By \eqref{eq:reg-geq}, we have for any $n\in \mathbb{N}_{>0}$
\begin{align}\label{eq:proof-prmi2-i1}
\frac{1}{n}D_1(\rho_{AB}^{\otimes n}\| \omega_{A^n}^n\otimes \omega_{B^n}^n)
-\frac{\log g_{n,d_A}}{n}-\frac{\log g_{n,d_B}}{n}
\leq \lim_{\alpha\rightarrow 1^-}I_\alpha^{\downarrow\downarrow}(A:B)_\rho
\leq I_1^{\downarrow\downarrow}(A:B)_\rho,
\end{align}
where the last inequality follows from the monotonicity in $\alpha$~(n).
By Remark~\ref{rem:universal-state}~(b), the second term and the third term on the left-hand side of~\eqref{eq:proof-prmi2-i1} vanish in the limit  $n\rightarrow\infty$. Therefore,
\begin{align}\label{eq:prmi2-proof-left-c}
I_1^{\downarrow\downarrow}(A:B)_\rho
=\lim_{n\rightarrow\infty}\frac{1}{n}D_1(\rho_{AB}^{\otimes n}\| \omega_{A^n}^n\otimes \omega_{B^n}^n)
\leq \lim_{\alpha\rightarrow 1^-}I_\alpha^{\downarrow\downarrow}(A:B)_\rho
\leq I_1^{\downarrow\downarrow}(A:B)_\rho,
\end{align}
where the first equality in~\eqref{eq:prmi2-proof-left-c} follows from~(l). 
Hence, $\lim_{\alpha\rightarrow 1^-}I_\alpha^{\downarrow\downarrow}(A:B)_\rho=I_1^{\downarrow\downarrow}(A:B)_\rho$.
\end{proof}
\begin{proof}[Proof of (q)]
Convexity is inherited from the Petz divergence because, according to the first equality in~\eqref{eq:prmi2-omega1} in~(l), 
$(\alpha-1)I_\alpha^{\downarrow\downarrow}(A:B)_\rho$ is the pointwise limit of a sequence of functions 
that are convex in $\alpha$.
\end{proof}
\begin{proof}[Proof of (p)] 
Let us define the following two functions.
\begin{align}
f:\quad (1/2,2)\rightarrow \mathbb{R},\quad \alpha
&\mapsto I_\alpha^{\downarrow\downarrow}(A:B)_\rho
\\
g:\quad (1/2,2)\rightarrow \mathbb{R},\quad \alpha
&\mapsto (\alpha -1)I_\alpha^{\downarrow\downarrow}(A:B)_\rho
\end{align}
By~(o), $f$ is continuous. 
By~(o) and~(q), $g$ is convex and continuous. 
Due to the convexity of $g$, the left and right derivative of $g$ exist at all points within its domain. 
Since $f(\alpha)=\frac{1}{\alpha-1}g(\alpha)$, we have for any $\alpha\in (\frac{1}{2},1)\cup (1,2)$
\begin{subequations}\label{eq:f-left-right}
\begin{align}
\frac{\partial}{\partial \alpha^-}f(\alpha)
&=-\frac{1}{(\alpha-1)^2}g(\alpha)+\frac{1}{\alpha -1}\frac{\partial}{\partial \alpha^-}g(\alpha)
=-\frac{1}{\alpha-1}f(\alpha)+\frac{1}{\alpha -1}\frac{\partial}{\partial \alpha^-}g(\alpha),
\\
\frac{\partial}{\partial \alpha^+}f(\alpha)
&= -\frac{1}{(\alpha-1)^2}g(\alpha)+\frac{1}{\alpha -1}\frac{\partial}{\partial \alpha^+}g(\alpha)
=-\frac{1}{\alpha-1}f(\alpha)+\frac{1}{\alpha -1}\frac{\partial}{\partial \alpha^+}g(\alpha).
\end{align}
\end{subequations}

For any $\alpha\in (\frac{1}{2},1)\cup (1,2)$ and 
any fixed $(\sigma_A,\tau_B)\in \argmin_{(\sigma_A',\tau_B')\in \mathcal{S}(A)\times \mathcal{S}(B)}D_\alpha (\rho_{AB}\| \sigma_A'\otimes \tau_B')$
\begin{subequations}\label{eq:prmi2-diff0}
\begin{align}
\frac{\partial}{\partial \alpha^+}f(\alpha)
&=\lim_{\varepsilon\rightarrow 0^+} \frac{1}{\varepsilon} (I_{\alpha+\varepsilon}^{\downarrow\downarrow}(A:B)_\rho -I_{\alpha}^{\downarrow\downarrow}(A:B)_\rho )
\\
&\leq \lim_{\varepsilon\rightarrow 0^+} \frac{1}{\varepsilon} (D_{\alpha+\varepsilon}(\rho_{AB}\| \sigma_A\otimes \tau_B) -D_{\alpha}(\rho_{AB}\| \sigma_A\otimes \tau_B) )
\\
&=\frac{\partial}{\partial \alpha}D_{\alpha}(\rho_{AB}\| \sigma_A\otimes \tau_B)
\label{eq:prmi2-diff1}\\
&=\lim_{\varepsilon\rightarrow 0^-} \frac{1}{\varepsilon} (D_{\alpha+\varepsilon}(\rho_{AB}\| \sigma_A\otimes \tau_B) -D_{\alpha}(\rho_{AB}\| \sigma_A\otimes \tau_B) )
\label{eq:prmi2-diff2}\\
&\leq \lim_{\varepsilon\rightarrow 0^-} \frac{1}{\varepsilon} (I_{\alpha+\varepsilon}^{\downarrow\downarrow}(A:B)_\rho -I_{\alpha}^{\downarrow\downarrow}(A:B)_\rho )
=\frac{\partial}{\partial \alpha^-}f(\alpha).
\end{align}
\end{subequations}
\eqref{eq:prmi2-diff1} and~\eqref{eq:prmi2-diff2} follow from the differentiability in $\alpha$ of the Petz divergence, see Remark~\ref{rem:petz-divergence}. 

We will now prove the continuous differentiability of $f$ on $\alpha\in (1,2)$. 
Since $g$ is convex, 
$\frac{\partial}{\partial \alpha^-}g(\alpha)\leq \frac{\partial}{\partial \alpha^+}g(\alpha)$ for all $\alpha\in (1,2)$. 
Hence, $\frac{\partial}{\partial \alpha^-}f(\alpha)\leq \frac{\partial}{\partial \alpha^+}f(\alpha)$ for all $\alpha\in (1,2)$ due to~\eqref{eq:f-left-right}. 
By~\eqref{eq:prmi2-diff0}, it follows that the left and right derivative of $f$ coincide, so $f$ is differentiable on $\alpha\in (1,2)$ and~\eqref{eq:prmi2-diff} holds. 
Since $f$ is differentiable and $g(\alpha)=(\alpha-1)f(\alpha)$ for all $\alpha\in (1,2)$, also $g$ is differentiable on $\alpha\in (1,2)$. 
Since $g$ is convex, its differentiability on $\alpha\in (1,2)$ implies its \emph{continuous} differentiability on $\alpha\in (1,2)$. 
By the product rule, this implies that also $f$ is \emph{continuously} differentiable on $\alpha\in (1,2)$. 

We will now prove the continuous differentiability of $f$ on $\alpha\in (\frac{1}{2},1)$. 
For any $\alpha\in (\frac{1}{2},1)$, let 
\begin{equation}
(\sigma_A^{(\alpha)},\tau_B^{(\alpha)})\in \argmin_{(\sigma_A,\tau_B)\in \mathcal{S}(A)\times \mathcal{S}(B)}D_\alpha (\rho_{AB}\|\sigma_A\otimes \tau_B)
\end{equation}
denote the unique minimizer, see~(j). 
Let us define the function 
\begin{align}
(1/2,1)\rightarrow \mathbb{R}, \quad 
\alpha\mapsto 
h(\alpha)\coloneqq 
f(\alpha)+(\alpha -1) 
\frac{\partial}{\partial \alpha}D_\alpha (\rho_{AB}\| \sigma_A^{(\alpha)}\otimes \tau_B^{(\alpha)}),
\end{align}
where $\sigma_A^{(\alpha)}$ and $\tau_B^{(\alpha)}$ are held fixed.  
The map $\alpha\mapsto (\sigma_A^{(\alpha)},\tau_B^{(\alpha)})$ is continuous on $\alpha\in (\frac{1}{2},1)$
due to the uniqueness of $(\sigma_A^{(\alpha)},\tau_B^{(\alpha)})$. 
By the continuous differentiability of the Petz divergence, see Remark~\ref{rem:petz-divergence}, it follows that $h$ is continuous. 
For any $\alpha\in (\frac{1}{2},1)$,
\begin{align}
\frac{\partial}{\partial \alpha^+}f(\alpha)
&\leq \frac{\partial}{\partial \alpha}D_\alpha (\rho_{AB}\| \sigma_A^{(\alpha)}\otimes \tau_B^{(\alpha)})
\leq \frac{\partial}{\partial \alpha^-}f(\alpha),
\label{eq:diff-h0}\\ 
\frac{\partial}{\partial \alpha^-}g(\alpha)
&\leq h(\alpha)
\leq \frac{\partial}{\partial \alpha^+}g(\alpha).
\label{eq:diff-g0}
\end{align}
\eqref{eq:diff-h0} follows from~\eqref{eq:prmi2-diff0}, and it is understood that $\sigma_A^{(\alpha)}$ and $ \tau_B^{(\alpha)}$ are held fixed in~\eqref{eq:diff-h0}. 
\eqref{eq:diff-g0} follows from~\eqref{eq:f-left-right} and~\eqref{eq:diff-h0}. 
Therefore, for any $\alpha\in (\frac{1}{2},1)$, 
\begin{subequations}\label{eq:prmi2-g0}
\begin{align}
h(\alpha)
=\lim_{\varepsilon\rightarrow 0^+}h(\alpha-\varepsilon)
&\leq \lim_{\varepsilon\rightarrow 0^+}\frac{\partial}{\partial \alpha^+}g(\alpha-\varepsilon)
\label{eq:prmi2-g1}\\
&\leq \frac{\partial}{\partial \alpha^+}g(\alpha)
\leq \lim_{\varepsilon\rightarrow 0^+}\frac{\partial}{\partial \alpha^-}g(\alpha+\varepsilon)
\leq\lim_{\varepsilon\rightarrow 0^+}h(\alpha+\varepsilon)
=h(\alpha).
\label{eq:prmi2-g2}
\end{align}
\end{subequations}
The first two inequalities in~\eqref{eq:prmi2-g2} follow from the convexity of $g$. 
It follows that all inequalities in~\eqref{eq:prmi2-g0} must be saturated, so 
$\frac{\partial}{\partial \alpha^+}g(\alpha)=h(\alpha)$ for all $\alpha\in (\frac{1}{2},1)$. 
Since $h$ is continuous, also $\frac{\partial}{\partial \alpha^+}g(\alpha)$ is continuous on $\alpha\in (\frac{1}{2},1)$. 
Since $g$ is convex, the continuity of the right derivative of $g$ implies that $g$ is differentiable and $g'(\alpha)=\frac{\partial}{\partial \alpha^+}g(\alpha)=h(\alpha)$ for all $\alpha\in (\frac{1}{2},1)$. 
Since $h$ is continuous, this proves that 
$g$ is \emph{continuously} differentiable on $\alpha\in (\frac{1}{2},1)$. 
By the product rule, this implies that also 
$f$ is continuously differentiable on $\alpha\in (\frac{1}{2},1)$. 

Next, we will prove the continuous differentiability of $f$ at $\alpha=1$. 
The combination of~\eqref{eq:i1-argmin} and a quantum Sibson identity~\cite[Eq.~(B10)]{hayashi2016correlation} 
implies that 
$\lim_{\alpha\rightarrow 1^-}(\sigma_A^{(\alpha)},\tau_B^{(\alpha)})=(\rho_A,\rho_B)$
and that the limits
\begin{subequations}\label{eq:prmi2-diff-a1}
\begin{align}
\lim_{\beta\rightarrow 1^-}\frac{\mathrm{d}}{\mathrm{d} \alpha}I_\alpha^{\downarrow\downarrow}(A:B)_\rho\big|_{\alpha=\beta}
&=\frac{\mathrm{d}}{\mathrm{d} \alpha}D_\alpha (\rho_{AB}\| \rho_A\otimes \rho_B)\big|_{\alpha=1}, 
\\
\lim_{\beta\rightarrow 1^+}\frac{\mathrm{d}}{\mathrm{d} \alpha}I_\alpha^{\downarrow\downarrow}(A:B)_\rho\big|_{\alpha=\beta}
&=\frac{\mathrm{d}}{\mathrm{d} \alpha}D_\alpha (\rho_{AB}\| \rho_A\otimes \rho_B)\big|_{\alpha=1}
\end{align}
\end{subequations}
exist. 
Therefore, they are identical to the left and right derivative of $I_\alpha^{\downarrow\downarrow}(A:B)_\rho$ at $\alpha=1$, respectively. 
By~\eqref{eq:prmi2-diff-a1}, the left and right derivative of $I_\alpha^{\downarrow\downarrow}(A:B)_\rho$ at $\alpha=1$ coincide. 
This proves differentiability at $\alpha=1$. 
The \emph{continuous} differentiability of $I_\alpha^{\downarrow\downarrow}(A:B)_\rho$ at $\alpha=1$ follows from the quantum Sibson identity~\cite[Eq.~(B10)]{hayashi2016correlation}. 

It remains to prove that $I_\alpha^{\downarrow\downarrow}(A:B)_\rho$ is right-differentiable at $\alpha=1/2$.  
Let $(\sigma_A^{(1/2)},\tau_B^{(1/2)})\coloneqq \lim_{\alpha\rightarrow 1/2^+}(\sigma_A^{(\alpha)},\tau_B^{(\alpha)})$. 
Then, $\rho_{AB}\not\perp\sigma_A^{(1/2)}\otimes\tau_B^{(1/2)}$, and the limit
\begin{align}\label{eq:diff-at-12}
\lim_{\beta\rightarrow 1/2^+}\frac{\mathrm{d}}{\mathrm{d} \alpha}I_\alpha^{\downarrow\downarrow}(A:B)_\rho\big|_{\alpha=\beta}
&=\frac{\mathrm{d}}{\mathrm{d} \alpha}D_\alpha (\rho_{AB}\| \sigma_A^{(1/2)}\otimes \tau_B^{(1/2)})\big|_{\alpha=1/2}
\end{align}
exists. 
Therefore, this limit is identical to the right derivative of $I_\alpha^{\downarrow\downarrow}(A:B)_\rho$ at $\alpha=1/2$. 
Note that~\eqref{eq:diff-at-12} lies in $[0,\infty)$ due to the differentiability and monotonicity in $\alpha$ of the Petz divergence, see Remark~\ref{rem:petz-divergence}.
\end{proof}
\begin{proof}[Proof of (s)]
Let $\alpha\in [0,\infty)$.

Case 1: $\alpha\in (0,\frac{1}{2}]$. Let $\beta\coloneqq \frac{1}{\alpha}\in [2,\infty)$. 
Let $\ket{\hat{\sigma}}_A\in A$ be a unit eigenvector of $\rho_A$ corresponding to its largest eigenvalue.
By duality~(e),
\begin{align}
\exp((1-\beta )I_\alpha^{\downarrow\downarrow}(A:B)_\rho )
&= \sup_{\sigma_A\in \mathcal{S}(A)} \widetilde{Q}_\beta(\rho_A\| \sigma_A^{-1})
\leq \sup_{\sigma_A\in \mathcal{S}(A)} Q_\beta(\rho_A\| \sigma_A^{-1})
\\
&=\sup_{\sigma_A\in \mathcal{S}(A)} \tr[\rho_A^\beta \sigma_A^{\beta -1}]
=\lVert \rho_A\rVert_\infty^\beta
\\
&= \widetilde{Q}_\beta(\rho_A\| \proj{\hat{\sigma}}_A^{-1})
\leq \sup_{\sigma_A\in \mathcal{S}(A)} \widetilde{Q}_\beta(\rho_A\| \sigma_A^{-1}),
\end{align}
so all inequalities must be saturated. 
Therefore, 
$I_\alpha^{\downarrow\downarrow}(A:B)_\rho=\frac{\beta}{\beta -1} H_{\infty}(A)_\rho=\frac{1}{1-\alpha}H_\infty(A)_\rho$.

Case 2: $\alpha\in (\frac{1}{2},1)$. Let $\beta\coloneqq \frac{1}{\alpha}\in (1,2)$. 
Let $\hat{\sigma}_A\coloneqq \rho_A^{\frac{1}{2\alpha -1}} / \tr[ \rho_A^{\frac{1}{2\alpha -1}}]
=\rho_A^{\frac{\beta}{2-\beta}}/\tr[\rho_A^{\frac{\beta}{2-\beta}}]$.
By duality~(e),
\begin{align}
\exp((1-\beta )I_\alpha^{\downarrow\downarrow}(A:B)_\rho )
&= \sup_{\sigma_A\in \mathcal{S}(A)} \widetilde{Q}_\beta(\rho_A\| \sigma_A^{-1})
\leq \sup_{\sigma_A\in \mathcal{S}(A)} Q_\beta(\rho_A\| \sigma_A^{-1})
\\
&= \sup_{\sigma_A\in \mathcal{S}(A)} \tr[\rho_A^\beta\sigma_A^{\beta-1}]
=\lVert \rho_A^\beta \rVert_{\frac{1}{2-\beta}}
=\tr[\rho_A^{\frac{\beta}{2-\beta}}]^{2-\beta}
\label{eq:prmi2-ac-2}\\
&=\widetilde{Q}_\beta(\rho_A\| \hat{\sigma}_A^{-1}) 
\leq \sup_{\sigma_A\in \mathcal{S}(A)} \widetilde{Q}_\beta(\rho_A\| \sigma_A^{-1}),
\end{align}
so all inequalities must be saturated. 
In~\eqref{eq:prmi2-ac-2}, we have used the variational characterization of the Schatten norms~\cite[Lemma~3.2]{tomamichel2016quantum}.
Therefore, 
$I_\alpha^{\downarrow\downarrow}(A:B)_\rho=2H_{\frac{\beta}{2-\beta}}(A)_\rho=2H_{\frac{1}{2\alpha -1}}(A)_\rho$.

Case 3: $\alpha\in (1,\infty)$. Let $\beta\coloneqq \frac{1}{\alpha}\in (0,1)$. 
Let $\hat{\sigma}_A\coloneqq \rho_A^{\frac{1}{2\alpha -1}} / \tr[ \rho_A^{\frac{1}{2\alpha -1}}]
=\rho_A^{\frac{\beta}{2-\beta}}/\tr[\rho_A^{\frac{\beta}{2-\beta}}]$.
By duality~(e),
\begin{align}
\exp((1-\beta )I_\alpha^{\downarrow\downarrow}(A:B)_\rho )
&= \inf_{\substack{\sigma_A\in \mathcal{S}(A):\\ \rho_A\ll \sigma_A}} \widetilde{Q}_\beta(\rho_A\| \sigma_A^{-1})
\geq \inf_{\substack{\sigma_A\in \mathcal{S}(A):\\ \rho_A\ll \sigma_A}}Q_\beta(\rho_A\| \sigma_A^{-1})
\\
&= \inf_{\substack{\sigma_A\in \mathcal{S}(A):\\ \rho_A\ll \sigma_A}}\tr[\rho_A^\beta\sigma_A^{\beta-1}]
=\lVert \rho_A^\beta \rVert_{\frac{1}{2-\beta}}
=\tr[\rho_A^{\frac{\beta}{2-\beta}}]^{2-\beta}
\label{eq:prmi2-ac-3}\\
&=\widetilde{Q}_\beta(\rho_A\| \hat{\sigma}_A^{-1}) 
\geq \inf_{\substack{\sigma_A\in \mathcal{S}(A):\\ \rho_A\ll \sigma_A}} \widetilde{Q}_\beta(\rho_A\| \sigma_A^{-1}),
\end{align}
so all inequalities must be saturated. 
In~\eqref{eq:prmi2-ac-3}, we have used the variational characterization of the Schatten quasi-norms~\cite[Lemma~3.2]{tomamichel2016quantum}.
Therefore, 
$I_\alpha^{\downarrow\downarrow}(A:B)_\rho=2H_{\frac{\beta}{2-\beta}}(A)_\rho=2H_{\frac{1}{2\alpha -1}}(A)_\rho$.

Case 4: $\alpha\in \{0,1\}$. Then the assertion follows from the previous cases due to the continuity in $\alpha$~(o).
\end{proof}
\begin{proof}[Proof of (t)]
The assertion regarding $\sigma_A$ and $\tau_B$ can be verified by inserting $\sigma_A$ and $\tau_B$  into~\eqref{eq:prmi2-pure}. 
The rest of the assertion follows from~(s).
\end{proof}
\begin{proof}[Proof of (g)]
Let $\ket{\rho}_{ABC}\in ABC$ be such that $\tr_C[\proj{\rho}_{ABC}]=\rho_{AB}$.

Let $\alpha\in [0,\infty)$. Then $\frac{1}{\alpha}\in (0,\infty]$. 
By duality of the minimized generalized PRMI, see Proposition~\ref{prop:gen-prmi}~(d), 
\begin{subequations}\label{eq:prop-bounds-2ra}
\begin{align}
I^{\downarrow\downarrow}_\alpha (A:B)_\rho
\leq I_\alpha^{\downarrow} (\rho_{AB}\| \rho_A^0/r_A )
&=-\widetilde{D}_{\frac{1}{\alpha}}(\rho_{AC}\| (\rho_A^0/r_A)^{-1}\otimes \rho_C)
\\
&=2\log r_A - \widetilde{D}_{\frac{1}{\alpha}}(\rho_{AC}\|\rho_A^0/r_A \otimes \rho_C)
\leq 2\log r_A.
\end{align}
\end{subequations}
The last inequality follows from the non-negativity of the sandwiched divergence.

Let now $\alpha\in [\frac{1}{2},2]$. Let $\gamma\coloneqq \frac{1}{2\alpha-1}\in [\frac{1}{3},\infty]$. 

First, suppose $\spec(\rho_A)\subseteq \{0,1/r_A\}$ and $H(A|B)_\rho=-\log r_A$. 
Then $\rho_A=\rho_A^0/r_A$ and $\rho_{AC}=\rho_A\otimes \rho_C$.
By (s), $I_\alpha^{\downarrow\downarrow}(A:B)_\rho 
= 2H_{\gamma}(A)_\rho 
= 2\log r_A$.

Now, suppose $I_\alpha^{\downarrow\downarrow}(A:B)_\rho=2\log r_A$ instead. 
Then the inequalities in~\eqref{eq:prop-bounds-2ra} must be saturated, so 
$\widetilde{D}_{\frac{1}{\alpha}}(\rho_{AC}\|\rho_A^0/r_A\otimes \rho_C)=0$.
Since $\frac{1}{\alpha}\geq \frac{1}{2}$, it follows from the positive definiteness of the sandwiched divergence that $\rho_{AC}=\rho_A^0/r_A\otimes \rho_C$.
Hence, $\rho_A=\rho_A^0/r_A$, which implies that $\spec(\rho_A)\subseteq\{0,1/r_A\}$.
Therefore, $H(A|B)_\rho=-H(A|C)_\rho=-H(A)_\rho=-\log r_A$.

Let now $\alpha\in [0,\frac{1}{2})$.
By~(b) and the expression for pure states in~(t), we have 
\begin{align}
I_\alpha^{\downarrow\downarrow}(A:B)_\rho 
&\leq I_\alpha^{\downarrow\downarrow}(A:BC)_{\proj{\rho}}
=\frac{1}{1-\alpha}H_\infty(A)_\rho 
\leq \frac{1}{1-\alpha}\log r_A <2\log r_A.
\end{align}
\end{proof}
\begin{proof}[Proof of (u)]
Let $\alpha\in [0,\infty)$.
\begin{align}
I_\alpha^{\downarrow\downarrow}(A:B)_\rho
&=\inf_{\sigma_A\in \mathcal{S}(A)}
I_\alpha^{\downarrow}(\rho_{AB}\| \sigma_A)
\\
&=\inf_{\sigma_A\in \mathcal{S}(A)}
\inf_{\substack{\tau_B\in \mathcal{S}(B):\\ \exists (t_y)_{y\in \mathcal{Y}}\in [0,1]^{\times |\mathcal{Y}|}:\\ \tau_B=\sum\limits_{y\in \mathcal{Y}}t_y \proj{b_y}_B }} 
D_\alpha(\rho_{AB}\| \sigma_A\otimes \tau_B)
\label{eq:doubly-cc1}\\
&=\inf_{\substack{\tau_B\in \mathcal{S}(B):\\ \exists (t_y)_{y\in \mathcal{Y}}\in [0,1]^{\times |\mathcal{Y}|}:\\ \tau_B=\sum\limits_{y\in \mathcal{Y}}t_y \proj{b_y}_B }} 
I_\alpha^{\downarrow}(\rho_{AB}\| \tau_B)
\\
&=\inf_{\substack{\tau_B\in \mathcal{S}(B):\\ \exists (t_y)_{y\in \mathcal{Y}}\in [0,1]^{\times |\mathcal{Y}|}:\\ \tau_B=\sum\limits_{y\in \mathcal{Y}}t_y \proj{b_y}_B }} 
\inf_{\substack{\sigma_A\in \mathcal{S}(A):\\ \exists (s_x)_{x\in \mathcal{X}}\in [0,1]^{\times |\mathcal{X}|}:\\ \sigma_A=\sum\limits_{x\in \mathcal{X}}s_x \proj{a_x}_A }} 
D_\alpha(\rho_{AB}\| \sigma_A\otimes \tau_B)
=I_\alpha^{\downarrow\downarrow}(X:Y)_P
\label{eq:doubly-cc2}
\end{align}
Above, we have used Proposition~\ref{prop:gen-prmi}~(q) twice: for~\eqref{eq:doubly-cc1}, and for the first equality in~\eqref{eq:doubly-cc2}.
\end{proof}
\begin{proof}[Proof of (v)]
Let $\alpha\in [0,\infty)$. 
By (u), $I_\alpha^{\downarrow\downarrow}(A:B)_\rho=I_\alpha^{\downarrow\downarrow}(X:Y)_P$. 
As shown in~\cite[Lemma~11]{lapidoth2019two},
$I_\alpha^{\downarrow\downarrow}(X:Y)_P=\frac{\alpha}{1-\alpha}H_{\infty}(A)_\rho$ if $\alpha\in [0,\frac{1}{2}]$, and 
$I_\alpha^{\downarrow\downarrow}(X:Y)_P=H_{\frac{\alpha}{2\alpha -1}}(A)_\rho $ if $\alpha\in (\frac{1}{2},\infty)$.
The assertion regarding $\sigma_A$ and $\tau_B$ can be verified by inserting $\sigma_A$ and $\tau_B$ into~\eqref{eq:prmi2-copycc}. 
\end{proof}

\section{Proof of Theorem~\ref{thm:srmi2}}\label{proof:srmi2}

For the proof of additivity, Theorem~\ref{thm:srmi2}~(d), we make use of a lemma that follows from the equivalence of optimizers and fixed-points established above in Lemma~\ref{lem:fixed_point}. 
We first state and prove this lemma, after which we present the proof of Theorem~\ref{thm:srmi2}.

\subsection{Lemma for Theorem~\ref{thm:srmi2}~(d)}\label{app:lemma-fixed}
\begin{lem}[Multiplicativity from fixed-point property]\label{lem:fixedpoint}
Let $\rho_{AC}\in \mathcal{S}(AC),\rho'_{DF}\in \mathcal{S}(DF),\mu_C\in \mathcal{S}(C),\mu'_F\in \mathcal{S}(F)$. Then all of the following hold.
\begin{enumerate}[label=(\alph*)]
\item For any $\beta\in [\frac{1}{2},1)$
\begin{align}
&\inf\limits_{\sigma_{AD}\in \mathcal{S}_{\sim \rho_{A}\otimes \rho'_{D}}(AD)}
\widetilde{Q}_\beta(\rho_{AC}\otimes \rho'_{DF}\| \sigma_{AD}^{-1}\otimes \mu_{C}\otimes \mu'_{F})
\\
&=\inf\limits_{\sigma_{A}\in \mathcal{S}_{\sim\rho_A}(A)}
\widetilde{Q}_\beta(\rho_{AC}\| \sigma_{A}^{-1}\otimes \mu_{C})
\cdot\inf\limits_{\sigma'_{D}\in \mathcal{S}_{\sim \rho'_D}(D)}
\widetilde{Q}_\beta(\rho'_{DF}\| {\sigma'_{D}}^{-1}\otimes \mu'_{F}).
\end{align}
\item For any $\beta\in (1,2]$
\begin{align}
&\sup\limits_{\sigma_{AD}\in \mathcal{S}(AD)}
\widetilde{Q}_\beta(\rho_{AC}\otimes \rho'_{DF}\| \sigma_{AD}^{-1}\otimes \mu_{C}\otimes \mu'_{F})
\\
&=\sup\limits_{\sigma_{A}\in \mathcal{S}(A)}
\widetilde{Q}_\beta(\rho_{AC}\| \sigma_{A}^{-1}\otimes \mu_{C})
\cdot\sup\limits_{\sigma'_{D}\in \mathcal{S}(D)}
\widetilde{Q}_\beta(\rho'_{DF}\| {\sigma'_{D}}^{-1}\otimes \mu'_{F}).
\end{align}
\end{enumerate}
\end{lem}
\begin{proof}[Proof of (a)]
Let $\beta\in [\frac{1}{2},1)$.

Case 1: $\rho_{C}\perp\mu_C\lor \rho'_{F}\perp \mu_F'$. 
Then both sides of the equality are zero, so the assertion is trivially true.

Case 2: $\rho_{C}\not\perp\mu_C\land \rho'_{F}\not\perp \mu_F'$. 
Let $\gamma\coloneqq \frac{\beta -1}{\beta}\in [-1,0)$.

Let $X_{AC}\coloneqq \mu_C^{\frac{1-\beta}{2\beta}} \rho_{AC}^{\frac{1}{2}}$ and $X_A\coloneqq \tr_C[X_{AC}]$.
Furthermore, let $\hat{\sigma}_A\in \mathcal{S}_{\sim X_A}(A)$ be such that
\begin{equation}
\hat{\sigma}_A
=\frac{\tr_C[({\hat{\sigma}_A}^{\frac{\gamma}{2}}X_{AC}X_{AC}^\dagger {\hat{\sigma}_A}^{\frac{\gamma}{2}})^\beta ]}{\tr[({\hat{\sigma}_A}^{\frac{\gamma}{2}}X_{AC}X_{AC}^\dagger {\hat{\sigma}_A}^{\frac{\gamma}{2}})^\beta ]}.
\label{eq:proof-fixed-s1}
\end{equation}
Note that such a quantum state exists due to Lemma~\ref{lem:fixed_point}~(b).

Similarly, let $\tilde{X}_{DF}\coloneqq {\mu'_F}^{\frac{1-\beta}{2\beta}} {\rho'_{DF}}^{\frac{1}{2}}$ and 
$\tilde{X}_D\coloneqq \tr_F[\tilde{X}_{DF}]$, 
and let $\tilde{\sigma}_D\in \mathcal{S}_{\sim \tilde{X}_D}(D)$ be such that 
\begin{equation}
\tilde{\sigma}_D
=\frac{\tr_F[({\tilde{\sigma}_D}^{\frac{\gamma}{2}}{\tilde{X}_{DF}}\tilde{X}_{DF}^\dagger {\tilde{\sigma}_D}^{\frac{\gamma}{2}})^\beta ]}{\tr[({\tilde{\sigma}_D}^{\frac{\gamma}{2}}{\tilde{X}_{DF}}\tilde{X}_{DF}^\dagger {\tilde{\sigma}_D}^{\frac{\gamma}{2}})^\beta ]}.
\label{eq:proof-fixed-s2}
\end{equation}
\eqref{eq:proof-fixed-s1} and~\eqref{eq:proof-fixed-s2} imply that 
$\hat{\sigma}_A\otimes \tilde{\sigma}_D\in \mathcal{S}_{\sim X_A\otimes \tilde{X}_D}(AD)$ and
\begin{align} \label{eq:proof-fixed-s3}
\frac{\tr_{CF}[((\hat{\sigma}_A\otimes \tilde{\sigma}_D)^{\frac{\gamma}{2}}{X_{AC}\otimes \tilde{X}_{DF}} X_{AC}^\dagger\otimes \tilde{X}_{DF}^\dagger (\hat{\sigma}_A\otimes \tilde{\sigma}_D)^{\frac{\gamma}{2}})^\beta ]}{\tr[((\hat{\sigma}_A\otimes \tilde{\sigma}_D)^{\frac{\gamma}{2}}{X_{AC}\otimes \tilde{X}_{DF}} X_{AC}^\dagger\otimes \tilde{X}_{DF}^\dagger (\hat{\sigma}_A\otimes \tilde{\sigma}_D)^{\frac{\gamma}{2}})^\beta ]}
&=\hat{\sigma}_A\otimes \tilde{\sigma}_D.
\end{align}

We have 
\begin{align}
&\inf\limits_{\sigma_{AD}\in \mathcal{S}_{\sim \rho_A\otimes \rho_D'}(AD)}
\widetilde{Q}_\beta(\rho_{AC}\otimes \rho'_{DF}\| \sigma_{AD}^{-1}\otimes \mu_{C}\otimes \mu'_{F})
\\
&=\inf\limits_{\sigma_{AD}\in \mathcal{S}_{\sim X_A\otimes \tilde{X}_D}(AD)}
\tr[( X_{AC}^\dagger\otimes \tilde{X}_{DF}^\dagger \sigma_{AD}^\gamma X_{AC}\otimes \tilde{X}_{DF} )^\beta]
\\
&=\tr[( X_{AC}^\dagger\otimes \tilde{X}_{DF}^\dagger (\hat{\sigma}_A\otimes \tilde{\sigma}_D)^\gamma X_{AC}\otimes \tilde{X}_{DF} )^\beta]
\label{eq:proof-fixed-1}\\
&=\tr[(X_{AC}^\dagger \hat{\sigma}_A^\gamma X_{AC})^\beta]
\cdot 
\tr[(\tilde{X}_{DF}^\dagger \tilde{\sigma}_D^\gamma \tilde{X}_{DF})^\beta]
\\
&=\inf\limits_{\sigma_{A}\in \mathcal{S}_{\sim X_A}(A)}
\tr[(X_{AC}^\dagger \sigma_A^\gamma X_{AC})^\beta]
\cdot \inf\limits_{\sigma'_{D}\in \mathcal{S}_{\sim \tilde{X}_{D}}(D) }
\tr[(\tilde{X}_{DF}^\dagger {\sigma'_D}^\gamma \tilde{X}_{DF})^\beta]
\label{eq:proof-fixed-2}
\\
&=\inf\limits_{\sigma_{A}\in \mathcal{S}_{\sim \rho_A}(A)}
\widetilde{Q}_\beta(\rho_{AC}\| \sigma_{A}^{-1}\otimes \mu_{C})
\cdot \inf\limits_{\sigma'_{D}\in \mathcal{S}_{\sim \rho'_D}(D)}
\widetilde{Q}_\beta(\rho'_{DF}\| {\sigma'_{D}}^{-1}\otimes \mu'_{F}).
\end{align}
\eqref{eq:proof-fixed-1} holds because~\eqref{eq:proof-fixed-s3} allows us to employ Lemma~\ref{lem:fixed_point}~(b).
\eqref{eq:proof-fixed-2} holds because~\eqref{eq:proof-fixed-s1} and~\eqref{eq:proof-fixed-s2} allow us to employ Lemma~\ref{lem:fixed_point}~(b).
\end{proof}
\begin{proof}[Proof of (b)]
Let $\beta\in (1,2]$. 
Let $\gamma\coloneqq \frac{\beta -1}{\beta}\in (0,\frac{1}{2}]\subseteq (0,1)$. 
Then $0<\beta\leq \frac{\beta}{\beta-1}=\frac{1}{\gamma}$. 
One can then prove the assertion in a way analogous to~(a) 
by replacing the infima by suprema and 
employing Lemma~\ref{lem:fixed_point}~(a) instead of Lemma~\ref{lem:fixed_point}~(b).
\end{proof}

\subsection{Proof of Theorem~\ref{thm:srmi2}}\label{app:srmi2}

\begin{proof}[Proof of (a)]
This assertion follows from the symmetry of the definition of the doubly minimized SRMI in~\eqref{eq:srmi2} with respect to $A$ and $B$.
\end{proof}
\begin{proof}[Proof of (b), (f), (i), (k), (l), (p)] 
Since $\widetilde{I}_\alpha^{\downarrow\downarrow}(A:B)_\rho=\inf_{\sigma_A\in \mathcal{S}(A)}\widetilde{I}_\alpha^{\downarrow}(\rho_{AB}\| \sigma_A)$, these properties follow from the corresponding properties of the minimized generalized SRMI, see Proposition~\ref{prop:gen-srmi}.
\end{proof}
\begin{proof}[Proof of (c)]
\begin{align}
\widetilde{I}_\alpha^{\downarrow\downarrow}(A':B')_{V\otimes W\rho_{AB}V^\dagger\otimes W^\dagger}
&=\inf_{\sigma_{A'}\in \mathcal{S}(A')}\widetilde{I}_\alpha^\downarrow (V\otimes W\rho_{AB}V^\dagger\otimes W^\dagger\| \sigma_{A'})
\\
&=\inf_{\sigma_{A'}\in \mathcal{S}(A')}\widetilde{I}_\alpha^\downarrow (V\rho_{AB}V^\dagger \| \sigma_{A'})
=\widetilde{I}_\alpha^{\downarrow\downarrow}(A':B)_{V\rho_{AB}V^\dagger}
\label{eq:iso-dd1s}\\
&=\inf_{\tau_{B}\in \mathcal{S}(B)}\widetilde{I}_\alpha^\downarrow (V\rho_{AB}V^\dagger \| \tau_{B})
\\
&=\inf_{\tau_{B}\in \mathcal{S}(B)}\widetilde{I}_\alpha^\downarrow (\rho_{AB}\| \tau_{B})
=\widetilde{I}_\alpha^{\downarrow\downarrow}(A:B)_{\rho_{AB}}
\label{eq:iso-dd2s}
\end{align}
Above, we have used the invariance of the minimized generalized SRMI under local isometries, see Proposition~\ref{prop:gen-srmi}~(b), twice:
for the first equality in~\eqref{eq:iso-dd1s},
and for the first equality in~\eqref{eq:iso-dd2s}.
\end{proof}
\begin{proof}[Proof of (e)]
The assertions in~\eqref{eq:i-duality-1} and~\eqref{eq:i-duality1} follow from~(i) and the duality of the minimized generalized SRMI, see Proposition~\ref{prop:gen-srmi}~(d). 
It remains to prove the other two equalities. 

Let $\alpha\in [\frac{2}{3},1)\cup (1,\infty]$ and $\beta\coloneqq \frac{\alpha}{2\alpha -1}\in [\frac{1}{2},1)\cup (1,2]$.
Let us define the following function. 
\begin{align}
f:\mathcal{S}(A)\times \mathcal{S}(C)\rightarrow [0,\infty), \quad
(\sigma_A,\mu_C)\mapsto 
&\, \widetilde{Q}_\beta(\rho_{AC}\| \sigma_A^{-1}\otimes \mu_C)
\\
&=\tr[( (\rho_{AC}^{\frac{1}{2}}\mu_C^{\frac{1-\beta}{2\beta}}) \sigma_A^{\frac{\beta-1}{\beta}}  (\mu_C^{\frac{1-\beta}{2\beta}}\rho_{AC}^{\frac{1}{2}}) )^\beta]
\label{eq:sion-case1}\\
&=\tr[( (\rho_{AC}^{\frac{1}{2}}\sigma_A^{\frac{\beta -1}{2\beta}}) \mu_C^{\frac{1-\beta}{\beta}}  (\sigma_A^{\frac{\beta -1}{2\beta}}\rho_{AC}^{\frac{1}{2}}) )^\beta]
\label{eq:sion-case2}
\end{align} 

Case 1: $\alpha\in [\frac{2}{3},1)$. Then $\beta\in (1,2]$.
By Sion's minimax theorem~\cite{sion1958general}, 
\begin{equation}
\sup_{\sigma_A\in\mathcal{S}(A)} 
\inf_{\mu_C\in \mathcal{S}_{>0}(C)}
\widetilde{Q}_\beta (\rho_{AC}\| \sigma_A^{-1}\otimes\mu_C )
=\inf_{\mu_C\in \mathcal{S}_{>0}(C)} 
\sup_{\sigma_A\in\mathcal{S}(A)}
\widetilde{Q}_\beta (\rho_{AC}\| \sigma_A^{-1}\otimes\mu_C ).
\end{equation}
The conditions for applying Sion's minimax theorem are fulfilled: 
The set $\mathcal{S}(A)$ is compact and convex, and 
$\mathcal{S}_{>0}(C)$ is convex. 
For any fixed $\mu_C\in \mathcal{S}_{>0}(C)$, the function 
$\mathcal{S}(A)\rightarrow [0,\infty),\sigma_A\mapsto f(\sigma_A,\mu_C)$
is continuous~\cite{mueller2013quantum} and concave~\cite[Theorem~2.1(a)]{evert2022convexity} 
(see also~\cite{epstein1973remarks,carlen2008minkowski,hiai2013concavity})
since $\frac{\beta-1}{\beta}\in (0,\frac{1}{2}]\subseteq [0,1]$ and 
$0<\beta\leq\frac{\beta}{\beta -1}$, see~\eqref{eq:sion-case1}.
For any fixed $\sigma_A\in \mathcal{S}(A)$, the function 
$\mathcal{S}_{>0}(C)\rightarrow [0,\infty),\mu_C\mapsto f(\sigma_A,\mu_C)$ is continuous~\cite{mueller2013quantum}
and convex~\cite[Theorem~2.1(b)]{evert2022convexity} 
(see also~\cite{hiai2013concavity}) 
since $\frac{1-\beta}{\beta}\in [-\frac{1}{2},0)\subseteq [-1,0]$ and $\beta>0$, see~\eqref{eq:sion-case2}.
Therefore, Sion's minimax theorem can be applied. 
This proves~\eqref{eq:i-duality-2}.

Case 2: $\alpha\in (1,\infty]$. Then $\beta\in [\frac{1}{2},1)$.
By Sion's minimax theorem~\cite{sion1958general}, 
\begin{equation}
\inf_{\sigma_A\in \mathcal{S}_{\sim\rho_A}(A)}
\sup_{\mu_C\in \mathcal{S}(C)}
\widetilde{Q}_\beta (\rho_{AC}\| \sigma_A^{-1}\otimes\mu_C )
=\sup_{\mu_C\in \mathcal{S}(C)}
\inf_{\sigma_A\in \mathcal{S}_{\sim\rho_A}(A)}
\widetilde{Q}_\beta (\rho_{AC}\| \sigma_A^{-1}\otimes\mu_C ).
\end{equation}
The conditions for applying Sion's minimax theorem are fulfilled: 
The set $\mathcal{S}(C)$ is compact and convex, and 
$\mathcal{S}_{\sim\rho_A}(A)$ is convex. 
For any fixed $\mu_C\in \mathcal{S}(C)$, the function 
$\mathcal{S}_{\sim\rho_A}(A)\rightarrow [0,\infty),\sigma_A\mapsto f(\sigma_A,\mu_C)$
is continuous~\cite{mueller2013quantum} and convex~\cite[Theorem~2.1(b)]{evert2022convexity} since 
$\frac{\beta-1}{\beta}\in [-1,0)\subseteq [-1,0]$ and $\beta>0$, see~\eqref{eq:sion-case1}.
For any fixed $\sigma_A\in \mathcal{S}_{\sim\rho_A}(A)$, 
the function 
$\mathcal{S}(C)\rightarrow [0,\infty),\mu_C\mapsto f(\sigma_A,\mu_C)$ is continuous~\cite{mueller2013quantum}
and concave~\cite[Theorem~2.1(a)]{evert2022convexity} since 
$\frac{1-\beta}{\beta}\in (0,1]\subseteq [0,1]$ and $0<\beta\leq \frac{\beta}{1-\beta}$, see~\eqref{eq:sion-case2}. 
Therefore, Sion's minimax theorem can be applied.
This proves~\eqref{eq:i-duality2}.
\end{proof}
\begin{proof}[Proof of (d)]
By the definition of the doubly minimized SRMI in~\eqref{eq:srmi2}, it is evident that 
\begin{equation}
\widetilde{I}_\alpha^{\downarrow\downarrow}(AD:BE)_{\rho_{AB}\otimes \rho'_{DE}}
\leq \widetilde{I}_\alpha^{\downarrow\downarrow}(A:B)_{\rho_{AB}}+\widetilde{I}_\alpha^{\downarrow\downarrow}(D:E)_{\rho'_{DE}}
\end{equation}
for all $\alpha\in (0,\infty]$. 
It remains to prove that the opposite inequality holds for $\alpha\in [\frac{2}{3},\infty]$. 

Let $\alpha\in [\frac{2}{3},\infty]$.
Let $\ket{\rho}_{ABC}\in ABC$ be such that $\tr_C[\proj{\rho}_{ABC}]=\rho_{AB}$, and let 
$\ket{\rho'}_{DEF}\in DEF$ be such that $\tr_F[\proj{\rho'}_{DEF}]=\rho'_{DE}$.
Then $\tr_{CF}[\proj{\rho}_{ABC}\otimes \proj{\rho'}_{DEF}]=\rho_{AB}\otimes \rho'_{DE}$.

Case 1: $\alpha\in [\frac{2}{3},1)$. Let $\beta\coloneqq \frac{\alpha}{2\alpha -1}\in (1,2]$. 
Then
\begin{align}
&\widetilde{I}_\alpha^{\downarrow\downarrow}(AD:BE)_{\rho_{AB}\otimes \rho'_{DE}}
\\
&=-\frac{1}{\beta-1}\log 
\inf\limits_{\mu_{CF}\in \mathcal{S}_{>0}(CF)}
\sup\limits_{\sigma_{AD}\in \mathcal{S}(AD)}
\widetilde{Q}_\beta(\rho_{AC}\otimes \rho'_{DF}\| \sigma_{AD}^{-1}\otimes \mu_{CF})
\label{eq:add-0}
\\
&\geq -\frac{1}{\beta-1}\log 
\inf\limits_{\substack{\mu_{C}\in \mathcal{S}_{>0}(C), \\ \mu'_{F}\in \mathcal{S}_{>0}(F)}}
\sup\limits_{\sigma_{AD}\in \mathcal{S}(AD)}
\widetilde{Q}_\beta(\rho_{AC}\otimes \rho'_{DF}\| \sigma_{AD}^{-1}\otimes \mu_{C}\otimes \mu'_{F})
\label{eq:add-1}
\\
&=-\frac{1}{\beta-1}\log 
\inf\limits_{\substack{\mu_{C}\in \mathcal{S}_{>0}(C),\\ \mu'_{F}\in \mathcal{S}_{>0}(F)}}
\sup\limits_{\substack{\sigma_{A}\in \mathcal{S}(A),\\ \sigma'_{D}\in \mathcal{S}(D) }}
\widetilde{Q}_\beta(\rho_{AC}\| \sigma_{A}^{-1}\otimes \mu_{C})
\cdot
\widetilde{Q}_\beta(\rho'_{DF}\| {\sigma'_{D}}^{-1}\otimes \mu'_{F})
\label{eq:add-2}
\\
&=\widetilde{I}_\alpha^{\downarrow\downarrow}(A:B)_{\rho_{AB}}
+\widetilde{I}_\alpha^{\downarrow\downarrow}(D:E)_{\rho'_{DE}}.
\label{eq:add-4}
\end{align}
\eqref{eq:add-0} and~\eqref{eq:add-4} follow from~\eqref{eq:i-duality-2} in duality~(e). 
\eqref{eq:add-2} follows from Lemma~\ref{lem:fixedpoint}~(b).

Case 2: $\alpha\in (1,\infty]$. 
Let $\beta\coloneqq \frac{\alpha}{2\alpha -1}\in [\frac{1}{2},1)$. 
Then 
\begin{align}
&\widetilde{I}_\alpha^{\downarrow\downarrow}(AD:BE)_{\rho_{AB}\otimes \rho'_{DE}}
\\
&=-\frac{1}{\beta-1}\log 
\sup\limits_{\mu_{CF}\in \mathcal{S}(CF)}
\inf\limits_{\sigma_{AD}\in \mathcal{S}_{\sim \rho_A\otimes \rho'_D}(AD)}
\widetilde{Q}_\beta(\rho_{AC}\otimes \rho'_{DF}\| \sigma_{AD}^{-1}\otimes \mu_{CF})
\label{eq:add0}
\\
&\geq -\frac{1}{\beta-1}\log 
\sup\limits_{\substack{\mu_{C}\in \mathcal{S}(C), \\ \mu'_{F}\in \mathcal{S}(F)}}
\inf\limits_{\sigma_{AD}\in \mathcal{S}_{\sim \rho_{A}\otimes \rho'_{D}}(AD)}
\widetilde{Q}_\beta(\rho_{AC}\otimes \rho'_{DF}\| \sigma_{AD}^{-1}\otimes \mu_{C}\otimes \mu'_{F})
\label{eq:add1}
\\
&=-\frac{1}{\beta-1}\log 
\sup\limits_{\substack{\mu_{C}\in \mathcal{S}(C), \\ \mu'_{F}\in \mathcal{S}(F)}}
\inf\limits_{\substack{\sigma_{A}\in \mathcal{S}_{\sim\rho_A}(A),\\ \sigma'_{D}\in \mathcal{S}_{\sim\rho_D'}(D)}}
\widetilde{Q}_\beta(\rho_{AC}\| \sigma_{A}^{-1}\otimes \mu_{C})
\cdot 
\widetilde{Q}_\beta(\rho'_{DF}\| {\sigma'_{D}}^{-1}\otimes \mu'_{F})
\label{eq:add2}
\\
&=\widetilde{I}_\alpha^{\downarrow\downarrow}(A:B)_{\rho_{AB}}
+\widetilde{I}_\alpha^{\downarrow\downarrow}(D:E)_{\rho'_{DE}}.
\label{eq:add4}
\end{align}
\eqref{eq:add0} and~\eqref{eq:add4} follow from~\eqref{eq:i-duality2} in duality~(e). 
\eqref{eq:add2} follows from Lemma~\ref{lem:fixedpoint}~(a).

Case 3: $\alpha =1$. Then additivity follows from~(k).
\end{proof}
\begin{proof}[Proof of (j):~\eqref{eq:srmi2-omega1}, \eqref{eq:srmi2-omega2}] 
Let $\alpha \in [\frac{2}{3},\infty]$. 

We will now prove the first equality in~\eqref{eq:srmi2-omega1}. 
For any $n\in \mathbb{N}_{>0}$ 
\begin{subequations}\label{eq:proof-reg-r}
\begin{align}
\widetilde{I}_\alpha^{\downarrow\downarrow}(A:B)_\rho
&=\inf_{\substack{\sigma_A\in \mathcal{S}(A),\\ \tau_B\in \mathcal{S}(B)}} 
\frac{1}{n}\widetilde{D}_\alpha (\rho_{AB}^{\otimes n}\| \sigma_A^{\otimes n}\otimes \tau_B^{\otimes n})
\label{eq:r0s}\\
&\geq \inf_{\substack{\sigma_{A^n}\in \mathcal{S}_{\sym}(A^{\otimes n}),\\ \tau_{B^n}\in \mathcal{S}_{\sym}(B^{\otimes n}) }}
\frac{1}{n}\widetilde{D}_\alpha (\rho_{AB}^{\otimes n}\| \sigma_{A^n}\otimes \tau_{B^n})
\label{eq:r1s}\\
&\geq \frac{1}{n}\widetilde{D}_\alpha (\rho_{AB}^{\otimes n}\| \omega_{A^n}^n\otimes \omega_{B^n}^n)- \frac{\log g_{n,d_A}}{n}- \frac{\log g_{n,d_B}}{n}.
\label{eq:r2s}
\end{align}
\end{subequations}
\eqref{eq:r0s} follows from the additivity of the sandwiched divergence. 
\eqref{eq:r2s} follows from Remark~\ref{rem:universal-state}~(b). 
In the limit $n\rightarrow\infty$, the second and third term in~\eqref{eq:r2s} vanish due to Remark~\ref{rem:universal-state}~(b). 
Hence, 
$\widetilde{I}_\alpha^{\downarrow\downarrow}(A:B)_\rho
\geq \limsup_{n\rightarrow\infty}\frac{1}{n}\widetilde{D}_\alpha (\rho_{AB}^{\otimes n}\| \omega_{A^n}^n \otimes \omega_{B^n}^n )$.

On the other hand, for any $n\in \mathbb{N}_{>0}$
\begin{subequations}\label{eq:proof-reg-omega}
\begin{align}
\frac{1}{n}\widetilde{D}_\alpha (\rho_{AB}^{\otimes n}\| \omega_{A^n}^n\otimes \omega_{B^n}^n)
&\geq \inf_{\substack{\sigma_{A^n}\in \mathcal{S}_{\sym}(A^{\otimes n}),\\ \tau_{B^n}\in \mathcal{S}_{\sym}(B^{\otimes n}) }} 
\frac{1}{n}\widetilde{D}_\alpha (\rho_{AB}^{\otimes n}\| \sigma_{A^n}\otimes \tau_{B^n})
\label{eq:r4s}\\
&\geq \inf_{\substack{\sigma_{A^n}\in \mathcal{S}(A^n),\\ \tau_{B^n}\in\mathcal{S}(B^n)}} 
\frac{1}{n}\widetilde{D}_\alpha (\rho_{AB}^{\otimes n}\| \sigma_{A^n}\otimes \tau_{B^n})
= \frac{1}{n}\widetilde{I}_\alpha^{\downarrow\downarrow}(A^n:B^n)_{\rho^{\otimes n}}
=\widetilde{I}_\alpha^{\downarrow\downarrow}(A:B)_{\rho}.
\label{eq:r5s}
\end{align}
\end{subequations}
\eqref{eq:r4s} follows from Remark~\ref{rem:universal-state}~(a). 
\eqref{eq:r5s} follows from additivity~(d). 
It follows that 
$\liminf_{n\rightarrow\infty}\frac{1}{n}\widetilde{D}_\alpha (\rho_{AB}^{\otimes n}\| \omega_{A^n}^n \otimes \omega_{B^n}^n )\geq \widetilde{I}_\alpha^{\downarrow\downarrow}(A:B)_{\rho}$. 
This completes the proof of the first equality in~\eqref{eq:srmi2-omega1}. 
\sloppy

We will now prove the second equality in~\eqref{eq:srmi2-omega1}. 
For any $n\in \mathbb{N}_{>0}$
\begin{subequations}\label{eq:proof-regularization-p}
\begin{align}
\frac{1}{n}\widetilde{D}_\alpha(\rho_{AB}^{\otimes n}\| \omega_{A^n}^n\otimes \omega_{B^n}^n)
&\geq \frac{1}{n}\widetilde{D}_\alpha(\mathcal{P}_{\omega_{A^n}^n\otimes \omega_{B^n}^n}(\rho_{AB}^{\otimes n})\| \omega_{A^n}^n\otimes \omega_{B^n}^n)
\label{eq:r-1}\\
&\geq \frac{1}{n}\widetilde{D}_\alpha(\rho_{AB}^{\otimes n}\| \omega_{A^n}^n\otimes \omega_{B^n}^n) 
-\frac{2}{n}\log \lvert\spec (\omega_{A^n}^n\otimes \omega_{B^n}^n)\rvert
\label{eq:r-2}\\
&\geq \frac{1}{n}\widetilde{D}_\alpha(\rho_{AB}^{\otimes n}\| \omega_{A^n}^n\otimes \omega_{B^n}^n) 
-2(d_A-1)\frac{\log (n+1)}{n}-2(d_B-1)\frac{\log (n+1)}{n}.
\label{eq:r-3}
\end{align}
\end{subequations}
\eqref{eq:r-1} follows from the data-processing inequality for the sandwiched divergence.
\eqref{eq:r-2} follows from~\cite[Lemma~3]{hayashi2016correlation}.
\eqref{eq:r-3} holds because 
\begin{align}
\lvert\spec (\omega_{A^n}^n\otimes \omega_{B^n}^n)\rvert
&\leq \lvert\spec (\omega_{A^n}^n)\rvert\cdot \lvert\spec (\omega_{B^n}^n)\rvert
\leq (n+1)^{d_A-1}(n+1)^{d_B-1},
\end{align}
where the last inequality follows from Remark~\ref{rem:universal-state}~(e). 
Taking the limit $n\rightarrow\infty$ of~\eqref{eq:proof-regularization-p} implies the second equality in~\eqref{eq:srmi2-omega1}.

We will now prove~\eqref{eq:srmi2-omega2}. For any $n\in \mathbb{N}_{>0}$
\begin{align}
\widetilde{I}_\alpha^{\downarrow\downarrow}(A:B)_\rho
&\geq \inf_{\substack{\sigma_{A^n}\in \mathcal{S}_{\sym}(A^{\otimes n}),\\ \tau_B\in \mathcal{S}_{\sym}(B^{\otimes n})}}
\frac{1}{n}\widetilde{D}_\alpha(\rho_{AB}^{\otimes n}\| \sigma_{A^n}\otimes\tau_{B^n})
\label{eq:rr-1}\\
&\geq \inf_{\substack{\sigma_{A^n}\in \mathcal{S}(A^n),\\ \tau_B\in \mathcal{S}(B^n)}}
\frac{1}{n}\widetilde{D}_\alpha(\rho_{AB}^{\otimes n}\| \sigma_{A^n}\otimes\tau_{B^n})
=\frac{1}{n}\widetilde{I}_\alpha^{\downarrow\downarrow}(A^n:B^n)_{\rho^{\otimes n}}
=\widetilde{I}_\alpha^{\downarrow\downarrow}(A:B)_\rho.
\label{eq:rr-2}
\end{align}
\eqref{eq:rr-1} follows from~\eqref{eq:r1s}. 
\eqref{eq:rr-2} follows from additivity~(d). 
This proves the assertion in~\eqref{eq:srmi2-omega2}.
\end{proof}
\begin{proof}[Proof of (m)] 
The continuity of $\widetilde{I}_\alpha^{\downarrow\downarrow} (A:B)_\rho$ on $\alpha\in (0,1)$ and on $\alpha\in [1,\infty]$ follows from the continuity in $\alpha$ of the sandwiched divergence.  
It remains to prove left-continuity at $\alpha=1$. 
We have for any $n\in \mathbb{N}_{>0}$
\begin{align}
\widetilde{D}_1 (\rho_{AB}^{\otimes n}\| \omega_{A^n}^n\otimes\omega_{B^n}^n)
-\frac{\log g_{n,d_A}}{n}-\frac{\log g_{n,d_B}}{n}
\leq \lim_{\alpha\rightarrow 1^-}\widetilde{I}_\alpha^{\downarrow\downarrow}(A:B)_\rho
\leq \widetilde{I}_1^{\downarrow\downarrow}(A:B)_\rho.
\end{align}
The first inequality follows from~\eqref{eq:proof-reg-r}, 
and the second inequality follows from monotonicity in $\alpha$~(l). 
By taking the limit $n\rightarrow\infty$, it follows that 
$\lim_{\alpha\rightarrow 1^-}\widetilde{I}_\alpha^{\downarrow\downarrow}(A:B)_\rho= \widetilde{I}_1^{\downarrow\downarrow}(A:B)_\rho$ 
due to~\eqref{eq:srmi2-omega1} in~(j) and Remark~\ref{rem:universal-state}~(b). 
\end{proof}
\begin{proof}[Proof of (h)] 
$I_\alpha^{\downarrow\downarrow}(A:B)_\rho\geq \widetilde{I}_\alpha^{\downarrow\downarrow}(A:B)_\rho$ for all $\alpha\in (0,\infty)$ 
follows from~\eqref{eq:alt}.

Let $\ket{\rho}_{ABC}$ be such that $\tr_C[\proj{\rho}_{ABC}]=\rho_{AB}$. 
Let $\alpha\in (\frac{1}{2},\infty)$. 

Case 1: $\alpha\in (\frac{1}{2},1)$. Then $\frac{\alpha}{2\alpha-1}\in (1,\infty)$. 
By duality~(e), 
\begin{align}
\widetilde{I}_\alpha^{\downarrow\downarrow}(A:B)_\rho
\geq\inf_{\sigma_A\in \mathcal{S}_{\not\perp \rho_A}(A)} 
-\frac{1}{\frac{\alpha}{2\alpha -1}-1}\log \widetilde{Q}_{\frac{\alpha}{2\alpha -1}}(\rho_{AC}\| \sigma_A^{-1}\otimes \rho_C)
=I_{2-\frac{1}{\alpha}}^{\downarrow\downarrow}(A:B)_\rho.
\end{align}
The last equality follows from the duality of the doubly minimized PRMI, see Theorem~\ref{thm:prmi2}. 

Case 2: $\alpha\in (1,\infty)$. Then $\frac{\alpha}{2\alpha-1}\in (\frac{1}{2},1)$. 
By duality~(e), 
\begin{align}
\widetilde{I}_\alpha^{\downarrow\downarrow}(A:B)_\rho
\geq\inf_{\substack{\sigma_A\in \mathcal{S}(A):\\ \rho_A\ll\sigma_A }} 
-\frac{1}{\frac{\alpha}{2\alpha -1}-1}\log \widetilde{Q}_{\frac{\alpha}{2\alpha -1}}(\rho_{AC}\| \sigma_A^{-1}\otimes \rho_C)
=I_{2-\frac{1}{\alpha}}^{\downarrow\downarrow}(A:B)_\rho.
\end{align}
The last equality follows from the duality of the doubly minimized PRMI, see Theorem~\ref{thm:prmi2}. 

Case 3: $\alpha\in \{\frac{1}{2},1\}$. Then the assertion follows from the previous cases by continuity in $\alpha$~(m).
\end{proof}
\begin{proof}[Proof of (o)] 
Convexity is inherited from the sandwiched divergence because, according to the first equality in~\eqref{eq:srmi2-omega1} in~(j), $(\alpha-1)\widetilde{I}_\alpha^{\downarrow\downarrow}(A:B)_\rho$ 
is the pointwise limit of a sequence of functions that are convex in $\alpha$.
\end{proof}
\begin{proof}[Proof of (n)] 
Let us define the following two functions. 
\begin{align}
f:(1,\infty)\rightarrow \mathbb{R},\quad \alpha
&\mapsto \widetilde{I}_\alpha^{\downarrow\downarrow}(A:B)_\rho
\\
g:(1,\infty)\rightarrow \mathbb{R},\quad \alpha
&\mapsto (\alpha -1)\widetilde{I}_\alpha^{\downarrow\downarrow}(A:B)_\rho
\end{align}
By (o), $g$ is convex. Thus, the left and right derivative of $g$ exist at all points of its domain and 
$\frac{\partial}{\partial \alpha^-}g(\alpha)\leq \frac{\partial}{\partial \alpha^+}g(\alpha)$. 
We have $f(\alpha)=\frac{1}{\alpha-1}g(\alpha)$ for all $\alpha\in(1,\infty)$. 
Hence, for any $\alpha\in (1,\infty)$
\begin{subequations}\label{eq:proof-srmi2-diff0}
\begin{align}
\frac{\partial}{\partial \alpha^-}\widetilde{I}_\alpha^{\downarrow\downarrow}(A:B)_\rho
=\frac{\partial}{\partial \alpha^-}f(\alpha)
&=-\frac{1}{(\alpha-1)^2}g(\alpha)+\frac{1}{\alpha -1}\frac{\partial}{\partial \alpha^-}g(\alpha),
\\
&\leq  -\frac{1}{(\alpha-1)^2}g(\alpha)+\frac{1}{\alpha -1}\frac{\partial}{\partial \alpha^+}g(\alpha)
=\frac{\partial}{\partial \alpha^+}f(\alpha)
=\frac{\partial}{\partial \alpha^+}\widetilde{I}_\alpha^{\downarrow\downarrow}(A:B)_\rho.
\end{align}
\end{subequations}
Let $\alpha\in (1,\infty)$ and let 
$(\sigma_A,\tau_B)\in \argmin_{(\sigma_A',\tau_B')\in \mathcal{S}(A)\times\mathcal{S}(B)} 
\widetilde{D}_\alpha(\rho_{AB}\| \sigma_A'\otimes \tau_B')$ be arbitrary but fixed. 
Then 
\begin{subequations}\label{eq:proof-srmi2-diff}
\begin{align}
\frac{\partial}{\partial \alpha^+}\widetilde{I}_\alpha^{\downarrow\downarrow}(A:B)_\rho
&=\lim_{\varepsilon\rightarrow 0^+} \frac{1}{\varepsilon} (\widetilde{I}_{\alpha+\varepsilon}^{\downarrow\downarrow}(A:B)_\rho -\widetilde{I}_{\alpha}^{\downarrow\downarrow}(A:B)_\rho )
\\
&\leq \lim_{\varepsilon\rightarrow 0^+} \frac{1}{\varepsilon} (\widetilde{D}_{\alpha+\varepsilon}(\rho_{AB}\| \sigma_A\otimes \tau_B) -\widetilde{D}_{\alpha}(\rho_{AB}\| \sigma_A\otimes \tau_B) )
\\
&=\frac{\partial}{\partial \alpha}\widetilde{D}_{\alpha}(\rho_{AB}\| \sigma_A\otimes \tau_B)
\label{eq:proof-srmi2-d1}\\
&=\lim_{\varepsilon\rightarrow 0^-} \frac{1}{\varepsilon} (\widetilde{D}_{\alpha+\varepsilon}(\rho_{AB}\| \sigma_A\otimes \tau_B) -\widetilde{D}_{\alpha}(\rho_{AB}\| \sigma_A\otimes \tau_B) )
\label{eq:proof-srmi2-d2}\\
&\leq \lim_{\varepsilon\rightarrow 0^-} \frac{1}{\varepsilon} (\widetilde{I}_{\alpha+\varepsilon}^{\downarrow\downarrow}(A:B)_\rho -\widetilde{I}_{\alpha}^{\downarrow\downarrow}(A:B)_\rho )
=\frac{\partial}{\partial \alpha^-}\widetilde{I}_\alpha^{\downarrow\downarrow}(A:B)_\rho.
\end{align}
\end{subequations}
\eqref{eq:proof-srmi2-d1} and~\eqref{eq:proof-srmi2-d2} follow from the differentiability of the sandwiched divergence. 
\eqref{eq:proof-srmi2-diff0} and \eqref{eq:proof-srmi2-diff} imply that 
the left and right derivative of $f$ coincide, so $f$ is differentiable and~\eqref{eq:srmi2-diff} holds for any $\alpha \in (1,\infty)$.

Next, we will show that $f$ is \emph{continuously} differentiable. 
Since $g(\alpha)=(\alpha-1)f(\alpha)$, $g$ is the product of two differentiable functions, so $g$ is also differentiable.
By convexity of $g$, this implies that $g$ is continuously differentiable.
By the product rule, it follows that $f$ is continuously differentiable. 

It remains to prove the assertion regarding the right derivative of $\widetilde{I}_\alpha^{\downarrow\downarrow}(A:B)_\rho$ at $\alpha=1$. 
We have 
\begin{align}
\frac{1}{2}V(A:B)_\rho
&=\frac{\partial}{\partial \alpha^+}I_{\alpha}^{\downarrow\downarrow}(A:B)_\rho|_{\alpha=1}
\label{eq:srmi2-diff0}\\
&=\lim_{\varepsilon\rightarrow 0^+}\frac{1}{\varepsilon} (I_{2-\frac{1}{1+\varepsilon}}^{\downarrow\downarrow}(A:B)_\rho-I_{1}^{\downarrow\downarrow}(A:B)_\rho)
\label{eq:srmi2-diff1}\\
&\leq \liminf_{\varepsilon\rightarrow 0^+}\frac{1}{\varepsilon} (\widetilde{I}_{1+\varepsilon}^{\downarrow\downarrow}(A:B)_\rho-\widetilde{I}_{1}^{\downarrow\downarrow}(A:B)_\rho)
\label{eq:srmi2-diff2}\\
&\leq \limsup_{\varepsilon\rightarrow 0^+}\frac{1}{\varepsilon} (\widetilde{I}_{1+\varepsilon}^{\downarrow\downarrow}(A:B)_\rho-\widetilde{I}_{1}^{\downarrow\downarrow}(A:B)_\rho)
\label{eq:srmi2-diff3}\\
&\leq \limsup_{\varepsilon\rightarrow 0^+}\frac{1}{\varepsilon} (\widetilde{D}_{1+\varepsilon}(\rho_{AB}\|\rho_A\otimes \rho_B)-\widetilde{D}_{1}(\rho_{AB}\|\rho_A\otimes \rho_B) )
\label{eq:srmi2-diff4}\\
&=\frac{\mathrm{d}}{\mathrm{d} \alpha}\widetilde{D}_{\alpha}(\rho_{AB}\|\rho_A\otimes \rho_B)|_{\alpha=1}
=\frac{1}{2}V(A:B)_\rho.
\label{eq:srmi2-diff5}
\end{align}
\eqref{eq:srmi2-diff0} has been proved in Theorem~\ref{thm:prmi2}. 
\eqref{eq:srmi2-diff1} holds due to the chain rule; note that the function 
$(-1,\infty)\rightarrow\mathbb{R},\varepsilon\mapsto h(\varepsilon)\coloneqq 2-\frac{1}{1+\varepsilon}$ is such that $h(0)=1$, $h'(\varepsilon)=\frac{1}{(1+\varepsilon)^2}$ and $h'(0)=1$. 
\eqref{eq:srmi2-diff2} follows from~(h). 
\eqref{eq:srmi2-diff4} follows from~(k).
\eqref{eq:srmi2-diff5} follows from the differentiability of the sandwiched divergence, see Remark~\ref{rem:sandwiched-divergence}. 
\end{proof}
\begin{proof}[Proof of (j):~\eqref{eq:srmi2-omega3}] 
Let $t\in [0,\infty)$. 
By Taylor expansion of $\widetilde{I}_\alpha^{\downarrow\downarrow}(A:B)_\rho$ about $\alpha=1$, 
\begin{equation}\label{eq:proof-reg2-taylor}
\widetilde{I}_{1+\frac{t}{\sqrt{n}}}^{\downarrow\downarrow}(A:B)_\rho 
=I(A:B)_\rho+\frac{1}{2}V(A:B)_\rho \frac{t}{\sqrt{n}} +o\left(\frac{t}{\sqrt{n}}\right)
\end{equation}
in the limit where $n\rightarrow\infty$.
For the first term, we used~(k), and for the second term, we used~(n).

The combination of the Taylor expansion in~\eqref{eq:proof-reg2-taylor} with~\eqref{eq:proof-reg-omega} and~\eqref{eq:proof-regularization-p} implies that 
\begin{align}
&t\sqrt{n} \left(\frac{1}{n} D_{1+\frac{t}{\sqrt{n}}}(\mathcal{P}_{\omega_{A^n}^n\otimes \omega_{B^n}^n}(\rho_{AB}^{\otimes n})\| \omega_{A^n}^n\otimes \omega_{B^n}^n) 
- I(A:B)_\rho \right)
\\
&\geq \frac{t^2}{2}V(A:B)_\rho
+t\sqrt{n}\, o\left(\frac{t}{\sqrt{n}}\right)-\frac{2t}{\sqrt{n}}\log( (n+1)^{d_A-1}(n+1)^{d_B-1}).
\end{align}
Hence,  
$\liminf_{n\rightarrow\infty}t\sqrt{n}\left(\frac{1}{n} D_{1+\frac{t}{\sqrt{n}}}(\mathcal{P}_{\omega_{A^n}^n\otimes \omega_{B^n}^n}(\rho_{AB}^{\otimes n})\| \omega_{A^n}^n\otimes \omega_{B^n}^n) 
- I(A:B)_\rho \right)
\geq \frac{t^2}{2}V(A:B)_\rho$.

The combination of the Taylor expansion in~\eqref{eq:proof-reg2-taylor} with~\eqref{eq:proof-reg-r} and~\eqref{eq:proof-regularization-p} implies that 
\begin{align}
&t\sqrt{n} \left(\frac{1}{n} D_{1+\frac{t}{\sqrt{n}}}(\mathcal{P}_{\omega_{A^n}^n\otimes \omega_{B^n}^n}(\rho_{AB}^{\otimes n})\| \omega_{A^n}^n\otimes \omega_{B^n}^n)
- I(A:B)_\rho \right)
\\
&\leq \frac{t^2}{2}V(A:B)_\rho
+t\sqrt{n}\, o\left(\frac{t}{\sqrt{n}}\right)+\frac{t}{\sqrt{n}}\log (g_{n,d_A} g_{n,d_B}).
\end{align}
Hence,  
$\limsup_{n\rightarrow\infty}t\sqrt{n}\left(\frac{1}{n} D_{1+\frac{t}{\sqrt{n}}}(\mathcal{P}_{\omega_{A^n}^n\otimes \omega_{B^n}^n}(\rho_{AB}^{\otimes n})\| \omega_{A^n}^n\otimes \omega_{B^n}^n) 
- I(A:B)_\rho \right)
\leq \frac{t^2}{2}V(A:B)_\rho$ due to Remark~\ref{rem:universal-state}~(b).
\end{proof}
\begin{proof}[Proof of (q)]
Let $\alpha\in [\frac{1}{2},\infty)$.

Case 1: $\alpha\in (1,\infty)$. 
Let $\beta\coloneqq \frac{\alpha}{2\alpha -1}\in (\frac{1}{2},1)$. 
Let $\hat{\sigma}_A\coloneqq \rho_A^{\frac{\alpha}{3\alpha-2}}/\tr[\rho_A^{\frac{\alpha}{3\alpha-2}}]
=\rho_A^{\frac{\beta}{2-\beta}}/\tr[\rho_A^{\frac{\beta}{2-\beta}}]$. 
Then, 
\begin{align}
\exp((1-\beta)\widetilde{I}_\alpha^{\downarrow\downarrow}(A:B)_\rho )
&=\inf_{\sigma_A\in \mathcal{S}_{\sim \rho_A}(A)} 
\widetilde{Q}_\beta (\rho_{A}\| \sigma_A^{-1})
\label{eq:proof-ac1}\\
&\geq \inf_{\sigma_A\in \mathcal{S}_{\sim \rho_A}(A)} 
Q_\beta (\rho_{A}\| \sigma_A^{-1})
=\inf_{\substack{\sigma \in \mathcal{S}(A):\\ \rho_A\ll \sigma_A}}
\tr[\rho_A^\beta\sigma_A^{\beta-1}]
\label{eq:proof-ac2}\\
&=\big\lVert \rho_A^\beta\big\rVert_{\frac{1}{2-\beta}} 
=\exp((1-\beta )2H_{\frac{\beta}{2-\beta}}(A)_\rho)
=\exp((1-\beta )2H_{\frac{\alpha}{3\alpha -2}}(A)_\rho)
\label{eq:proof-ac3}\\
&=\widetilde{Q}_\beta (\rho_{A}\| \hat{\sigma}_A^{-1})
\geq \inf_{\sigma_A\in \mathcal{S}_{\sim\rho_A}(A)} 
\widetilde{Q}_\beta (\rho_{A}\| \sigma_A^{-1}).
\end{align}
\eqref{eq:proof-ac1} follows from duality~(e). 
\eqref{eq:proof-ac2} follows from~\eqref{eq:alt}.
\eqref{eq:proof-ac3} follows from the variational characterization of the Schatten quasi-norms~\cite[Lemma~3.2]{tomamichel2016quantum}.

Case 2: $\alpha\in (\frac{2}{3},1)$. 
Let $\beta\coloneqq \frac{\alpha}{2\alpha -1}\in (1,2)$. 
Let $\hat{\sigma}_A\coloneqq \rho_A^{\frac{\alpha}{3\alpha-2}}/\tr[\rho_A^{\frac{\alpha}{3\alpha-2}}]
=\rho_A^{\frac{\beta}{2-\beta}}/\tr[\rho_A^{\frac{\beta}{2-\beta}}]$. 
Then,
\begin{align}
\exp((1-\beta)\widetilde{I}_\alpha^{\downarrow\downarrow}(A:B)_\rho )
&=\sup_{\sigma_A\in \mathcal{S}(A)} \widetilde{Q}_\beta (\rho_{A}\| \sigma_A^{-1})
\label{eq:proof-ac-1}\\
&\leq \sup_{\sigma_A\in \mathcal{S}(A)} Q_\beta (\rho_{A}\| \sigma_A^{-1})
=\sup_{\sigma_A \in \mathcal{S}(A)}\tr[\rho_A^\beta\sigma_A^{\beta-1}]
\label{eq:proof-ac-2}
\\
&=\big\lVert \rho_A^\beta\big\rVert_{\frac{1}{2-\beta}}
=\exp((1-\beta )2H_{\frac{\beta}{2-\beta}}(A)_\rho)
=\exp((1-\beta )2H_{\frac{\alpha}{3\alpha -2}}(A)_\rho)
\label{eq:proof-ac-3}
\\
&= \widetilde{Q}_\beta (\rho_{A}\| \hat{\sigma}_A^{-1})
\leq \sup_{\sigma_A\in \mathcal{S}(A)} \widetilde{Q}_\beta (\rho_{A}\| \sigma_A^{-1}).
\end{align}
\eqref{eq:proof-ac-1} follows from duality~(e). 
\eqref{eq:proof-ac-2} follows from~\eqref{eq:alt}.
\eqref{eq:proof-ac-3} follows from the variational characterization of the Schatten norms~\cite[Lemma~3.2]{tomamichel2016quantum}.

Case 3: $\alpha\in (\frac{1}{2},\frac{2}{3}]$. 
Let $\beta\coloneqq \frac{\alpha}{2\alpha -1}\in [2,\infty)$. 
Let $\ket{\hat{\sigma}}_A$ be a unit eigenvector of $\rho_A$ corresponding to its largest eigenvalue.
Then, 
\begin{align}
\exp((1-\beta)\widetilde{I}_\alpha^{\downarrow\downarrow}(A:B)_\rho )
&=\sup_{\sigma_A\in \mathcal{S}(A)} \widetilde{Q}_\beta (\rho_{A}\| \sigma_A^{-1})
\label{eq:proof-ac_1}\\
&\leq \sup_{\sigma_A\in \mathcal{S}(A)} Q_\beta (\rho_{A}\| \sigma_A^{-1})
=\sup_{\sigma_A \in \mathcal{S}(A)}\tr[\rho_A^\beta\sigma_A^{\beta-1}]
\label{eq:proof-ac_2}\\
&=\sup_{\substack{\ket{\sigma}_A\in A:\\ \brak{\sigma}_A=1}}
\tr[\rho_A^\beta\proj{\sigma}_A]
= \lVert \rho_A^{\beta} \rVert_\infty
= \lVert \rho_A\rVert_\infty^{\beta}
=\exp(-\beta H_\infty(A)_\rho)
\label{eq:proof-ac_3}\\
&= \widetilde{Q}_\beta (\rho_{A}\| \proj{\hat{\sigma}}_A^{-1})
\leq \sup_{\sigma_A\in \mathcal{S}(A)}
\widetilde{Q}_\beta (\rho_{A}\| \sigma_A^{-1}).
\label{eq:proof-ac_4}
\end{align}
\eqref{eq:proof-ac_1} follows from duality~(e). 
\eqref{eq:proof-ac_2} follows from~\eqref{eq:alt}.
\eqref{eq:proof-ac_3} holds because $\beta-1\in [1,\infty)$.

Case 4: $\alpha\in \{\frac{1}{2},1\}$. 
The assertion follows from the other cases by continuity in $\alpha$~(m).
\end{proof}
\begin{proof}[Proof of (r)] 
Let $\alpha\in (0,\infty)$.

Case 1: $\alpha\in [\frac{1}{2},\infty)$. By (q), 
$\widetilde{I}_\alpha^{\downarrow\downarrow}(A:B)_{\proj{\rho}}=2H_{\frac{\alpha}{3\alpha -2}}(A)_\rho$
if $\alpha\in (\frac{2}{3},\infty)$, and 
$\widetilde{I}_\alpha^{\downarrow\downarrow}(A:B)_{\proj{\rho}}=\frac{\alpha}{1-\alpha} H_{\infty}(A)_\rho$
if $\alpha\in [\frac{1}{2},\frac{2}{3}]$. 
The assertion regarding $\sigma_A$ and $\tau_B$ can be verified by inserting $\sigma_A$ and $\tau_B$ into~\eqref{eq:srmi2-pure}.

Case 2: $\alpha\in (0,\frac{1}{2})$. Then $\frac{1-\alpha}{\alpha}\in (1,\infty)$. 
By the expression of the minimized generalized SRMI for pure states, see Proposition~\ref{prop:gen-srmi}~(q), 
\begin{align}
\exp\left( \frac{\alpha-1}{\alpha} \widetilde{I}_\alpha^{\downarrow\downarrow}(A:B)_{\proj{\rho}} \right)
=\sup_{\sigma_A\in \mathcal{S}(A)} \lVert \rho_A^{\frac{1}{2}} \sigma_A^{\frac{1-\alpha}{\alpha}}\rho_A^{\frac{1}{2}} \rVert_\infty
= \lVert \rho_A \rVert_\infty
=\exp(-H_\infty(A)_\rho).
\end{align}
The assertion regarding $\sigma_A$ and $\tau_B$ can be verified by inserting $\sigma_A$ and $\tau_B$ into~\eqref{eq:srmi2-pure}.
\end{proof}
\begin{proof}[Proof of (g)] 
Let $\ket{\rho}_{ABC}\in ABC$ be such that $\tr_C[\proj{\rho}_{ABC}]=\rho_{AB}$. 

Let $\alpha\in [\frac{2}{3},\infty]$. Then $\frac{\alpha}{3\alpha-2}\in [\frac{1}{3},\infty]$ and $\frac{\alpha}{2\alpha-1}\in [\frac{1}{2},2]$.
By (b) and the expression for pure states~(r), 
\begin{align}\label{eq:proof-srmi2-bounds}
\widetilde{I}_\alpha^{\downarrow\downarrow}(A:B)_\rho
\leq \widetilde{I}_\alpha^{\downarrow\downarrow}(A:BC)_{\proj{\rho}}
&=2H_{\frac{\alpha}{3\alpha -2}}(A)_\rho
\leq 2H_{1/3}(A)_\rho \leq 2\log r_A.
\end{align}

First, suppose $\spec(\rho_A)\subseteq \{0,1/r_A\}$ and $H(A|B)_\rho=-\log r_A$. 
Then $\rho_A=\rho_A^0/r_A$ and $\rho_{AC}=\rho_A\otimes \rho_C$. 
By (q), $\widetilde{I}_\alpha^{\downarrow\downarrow}(A:B)_\rho =2H_{\frac{\alpha}{3\alpha -2}}(A)_\rho=2\log r_A$.

Now, suppose $\widetilde{I}_\alpha^{\downarrow\downarrow}(A:B)_\rho=2\log r_A$ instead. 
Then the inequalities in~\eqref{eq:proof-srmi2-bounds} must be saturated. 
Hence, $H_{1/3}(A)_\rho=\log r_A$, which implies that $\spec(\rho_A)\subseteq\{0,1/r_A\}$.
By duality~(e), 
\begin{align}
\widetilde{I}_\alpha^{\downarrow\downarrow}(A:B)_\rho 
\leq -\widetilde{I}_{\frac{\alpha}{2\alpha -1}}^{\downarrow}(\rho_{AC}\| \rho_A^{-1})
&=-\min_{\mu_C\in \mathcal{S}(C)} \widetilde{D}_{\frac{\alpha}{2\alpha -1}} (\rho_{AC}\| \rho_A^{-1}\otimes \mu_C)
\\
&=2\log r_A-\min_{\mu_C\in \mathcal{S}(C)} \widetilde{D}_{\frac{\alpha}{2\alpha -1}} (\rho_{AC}\| \rho_A\otimes \mu_C) 
\leq 2\log r_A,
\label{eq:proof-srmi2-bounds-min}
\end{align}
where the last inequality follows from the non-negativity of the sandwiched divergence. 
Since  $\widetilde{I}_\alpha^{\downarrow\downarrow}(A:B)_\rho=2\log r_A$, the inequality in~\eqref{eq:proof-srmi2-bounds-min} must be saturated, so 
$\widetilde{D}_{\frac{\alpha}{2\alpha -1}} (\rho_{AC}\| \rho_A\otimes \mu_C) =0$ for some $\mu_C\in \mathcal{S}(C)$.
By positive definiteness of the sandwiched divergence, we have $\rho_{AC}=\rho_A\otimes \mu_C$.
Therefore, $H(A|B)_\rho=-H(A|C)_\rho=-H(A)_\rho=-\log r_A$.

Let now $\alpha\in [\frac{1}{2},\frac{2}{3})$. 
By~(b) and the expression for pure states~(r), 
\begin{align}\label{eq:proof-bounds-a}
\widetilde{I}_\alpha^{\downarrow\downarrow}(A:B)_\rho
\leq \widetilde{I}_\alpha^{\downarrow\downarrow}(A:BC)_{\proj{\rho}}
&=\frac{\alpha}{1-\alpha}H_\infty(A)_\rho
\leq\frac{\alpha}{1-\alpha}\log r_A
<2\log r_A .
\end{align}

Let now $\alpha\in (0,\frac{1}{2})$. By monotonicity in $\alpha$~(l) and~\eqref{eq:proof-bounds-a}, 
\begin{align}
\widetilde{I}_\alpha^{\downarrow\downarrow}(A:B)_\rho
\leq \widetilde{I}_{1/2}^{\downarrow\downarrow}(A:B)_\rho\leq H_\infty(A)_\rho\leq 2H_{1/3}(A)_\rho.
\end{align}
\end{proof}
\begin{proof}[Proof of (s)] 
Let $\alpha\in [\frac{1}{2},\infty]$. 
Let $(\sigma_A,\tau_B)\in\mathcal{S}(A) \times \mathcal{S}(B)$ be such that 
$\widetilde{I}_\alpha^{\downarrow\downarrow}(A:B)_\rho = \widetilde{D}_\alpha(\rho_{AB}\| \sigma_A\otimes \tau_B)$. 
Let 
\begin{equation}
\sigma_A'\coloneqq \sum_{x\in \mathcal{X}}\proj{a_x}_{A}\,\sigma_A\,\proj{a_x}_{A}\in \mathcal{S}(A), 
\qquad 
\tau_B'\coloneqq \sum_{y\in \mathcal{Y}}\proj{b_y}_{B}\,\tau_B\,\proj{b_y}_{B}\in \mathcal{S}(B).
\end{equation}
By the data-processing inequality for the sandwiched divergence, see Remark~\ref{rem:sandwiched-divergence},
\begin{align}
\widetilde{D}_\alpha(\rho_{AB}\| \sigma_A\otimes \tau_B)
&\geq 
\widetilde{D}_\alpha (\sum_{\substack{x\in \mathcal{X}, \\ y\in \mathcal{Y}}}
\proj{a_x,b_y}_{AB}\, \rho_{AB}\, \proj{a_x,b_y}_{AB}
\|  \sigma_A'\otimes\tau_B')
\\
&=\widetilde{D}_\alpha(\rho_{AB}\| \sigma_A'\otimes \tau_B')
=D_\alpha (\rho_{AB}\| \sigma_A'\otimes \tau_B')
\geq \widetilde{I}_\alpha^{\downarrow\downarrow}(A:B)_\rho.
\end{align}
Therefore,
\begin{align}
\widetilde{I}_\alpha^{\downarrow\downarrow}(A:B)_\rho 
=\min_{\substack{\sigma_A\in \mathcal{S}(A):\\ \exists (s_x)_{x\in \mathcal{X}}\in [0,1]^{\times |\mathcal{X}|}:\\ \sigma_A=\sum\limits_{x\in \mathcal{X}}s_x \proj{a_x}_A }}
\min_{\substack{\tau_B\in \mathcal{S}(B):\\ \exists (t_y)_{y\in \mathcal{Y}}\in [0,1]^{\times |\mathcal{Y}|}:\\ \tau_B=\sum\limits_{y\in \mathcal{Y}}t_y \proj{b_y}_B }} 
D_\alpha(\rho_{AB}\| \sigma_A\otimes \tau_B)
=I_\alpha^{\downarrow\downarrow}(X:Y)_P .
\end{align}
\end{proof}

\bibliographystyle{arxiv_fullname}
\bibliography{bibfile}

\end{document}

%% file: fig_prmi.tex
\begin{tikzpicture}
\begin{axis}[
    axis lines = left,
    xlabel = \(\alpha\),
    x label style={at={(axis description cs:1.06,0.173)}},
    ymin = 0,
    ytick distance=0.5,
    height = 0.35\textwidth,
    width = 0.6\textwidth,
    legend pos=south east,
    legend cell align={left},
    clip=false
]

%%%%% MI1
\addplot[thick,color=Magenta]
table[x=alpha,y=rmi0]{data_prmi.txt};
\addlegendentry{\(I_\alpha^{\uparrow\uparrow}(A:B)_{\proj{\rho}}\)}
\addplot[thick,color=blue]
table[x=alpha,y=rmi1]{data_prmi.txt};
\addlegendentry{\(I_\alpha^{\uparrow\downarrow}(A:B)_{\proj{\rho}}\)}
\addplot[thick,color=LimeGreen]
table[x=alpha,y=rmi2]{data_prmi.txt};
\addlegendentry{\(I_\alpha^{\downarrow\downarrow}(A:B)_{\proj{\rho}}\)}

%%%%% hlines
\addplot[color=LimeGreen,dashed,thick]
table[x=alpha,y=h1]{data_prmih.txt};
\addplot[color=blue,dashed,thick]
table[x=alpha,y=h2]{data_prmih.txt};
\addplot[color=Magenta,dashed,thick]
table[x=alpha,y=h3]{data_prmih.txt};
\addplot[color=lightgray,dashed,thick]
table[x=alpha,y=h4]{data_prmih.txt};
\addplot[color=LimeGreen,dashed,thick]
table[x=alpha,y=h5]{data_prmih.txt};
\addplot[color=blue,dashed,thick]
table[x=alpha,y=h6]{data_prmih.txt};

%%%%% text
\node[right, text=black] at (309,21) {$H_\infty(A)_\rho$};
\node[right, text=black] at (309,42) {$2H_\infty(A)_\rho$};
\node[right, text=black] at (309,66) {$2H_{3}(A)_\rho$};
\node[right, text=black] at (309,100) {$2H(A)_\rho$};
\node[right, text=black] at (309,138) {$2H_0(A)_\rho$};
\node[right, text=black] at (309,183) {$2H_{-1}(A)_\rho$};

\end{axis}
\end{tikzpicture}

%% file: fig_prmi_cc.tex
\begin{tikzpicture}
\begin{axis}[
    axis lines = left,
    xlabel = \(\alpha\),
    x label style={at={(axis description cs:1.06,0.173)}},
    ymin = 0,
    height = 0.35\textwidth,
    width = 0.6\textwidth,
    legend pos=south east,
    legend cell align={left},
    clip=false
]

%%%%% MI1
\addplot[thick,color=Magenta]
table[x=alpha,y=rmi0]{data_prmi_cc.txt};
\addlegendentry{\(I_\alpha^{\uparrow\uparrow}(A:B)_{\rho}=I_\alpha^{\uparrow\uparrow}(X:Y)_P\)}
\addplot[thick,color=blue]
table[x=alpha,y=rmi1]{data_prmi_cc.txt};
\addlegendentry{\(I_\alpha^{\uparrow\downarrow}(A:B)_{\rho}=I_\alpha^{\uparrow\downarrow}(X:Y)_P\)}
\addplot[thick,color=LimeGreen]
table[x=alpha,y=rmi2]{data_prmi_cc.txt};
\addlegendentry{\(I_\alpha^{\downarrow\downarrow}(A:B)_{\rho}=I_\alpha^{\downarrow\downarrow}(X:Y)_P\)}

%%%%% hlines
\addplot[color=blue,dashed,thick]
table[x=alpha,y=h1]{data_prmih_cc.txt};
\addplot[color=Magenta,dashed,thick]
table[x=alpha,y=h2]{data_prmih_cc.txt};
\addplot[color=lightgray,dashed,thick]
table[x=alpha,y=h3]{data_prmih_cc.txt};
\addplot[color=LimeGreen,dashed,thick]
table[x=alpha,y=h4]{data_prmih_cc.txt};
\addplot[color=blue,dashed,thick]
table[x=alpha,y=h5]{data_prmih_cc.txt};

%%%%% text
\node[right, text=black] at (309,218) {$H_\infty(A)_\rho$};
\node[right, text=black] at (309,380) {$H_2(A)_\rho$};
\node[right, text=black] at (309,490) {$H(A)_\rho$};
\node[right, text=black] at (309,590) {$H_{1/2}(A)_\rho$};
\node[right, text=black] at (309,690) {$H_{0}(A)_\rho$};

\end{axis}
\end{tikzpicture}

%% file: fig_srmi.tex
\begin{tikzpicture}
\begin{axis}[
    axis lines = left,
    xlabel = \(\alpha\),
    x label style={at={(axis description cs:1.06,0.173)}},
    ymin = 0,
    ymax = 2.15,
    height = 0.35\textwidth,
    width = 0.6\textwidth,
    legend pos=south east,
    legend cell align={left},
    clip=false
]

%%%%% MI1
\addplot[thick,color=blue]
table[x=alpha,y=rmi0]{data_srmi.txt};
\addlegendentry{\(\widetilde{I}_\alpha^{\uparrow\uparrow}(A:B)_{\proj{\rho}}\)}
\addplot[thick,color=LimeGreen]
table[x=alpha,y=rmi1]{data_srmi.txt};
\addlegendentry{\(\widetilde{I}_\alpha^{\uparrow\downarrow}(A:B)_{\proj{\rho}}\)}
\addplot[thick,color=RedOrange]
table[x=alpha,y=rmi2]{data_srmi.txt};
\addlegendentry{\(\widetilde{I}_\alpha^{\downarrow\downarrow}(A:B)_{\proj{\rho}}\)}

%%%%% hlines
\addplot[color=LimeGreen,dashed,thick]
table[x=alpha,y=h1]{data_srmih.txt};
\addplot[color=blue,dashed,thick]
table[x=alpha,y=h2]{data_srmih.txt};
\addplot[color=lightgray,dashed,thick]
table[x=alpha,y=h3]{data_srmih.txt};
\addplot[color=RedOrange,dashed,thick]
table[x=alpha,y=h4]{data_srmih.txt};
\addplot[color=LimeGreen,dashed,thick]
table[x=alpha,y=h5]{data_srmih.txt};
\addplot[color=blue,dashed,thick]
table[x=alpha,y=h6]{data_srmih.txt};

%%%%% text
\node[right, text=black] at (309,21) {$H_\infty(A)_\rho$};
\node[right, text=black] at (309,45) {$2H_\infty(A)_\rho$};
\node[right, text=black] at (309,100) {$2H(A)_\rho$};
\node[right, text=black] at (309,120) {$2H_{1/3}(A)_\rho$};
\node[right, text=black] at (309,141) {$2H_0(A)_\rho$};
\node[right, text=black] at (309,183) {$2H_{-1}(A)_\rho$};

\end{axis}
\end{tikzpicture}

%% file: properties.bbl
\begin{thebibliography}{10}

\bibitem{lapidoth2019two}
Amos Lapidoth and Christoph Pfister.
\newblock {Two Measures of Dependence}.
\newblock {\em Entropy}, 21(778), 2019.
\newblock
  \texttt{\href{http://dx.doi.org/10.3390/e21080778}{DOI:\,10.3390/e21080778}}.

\bibitem{sibson1969information}
Robin Sibson.
\newblock Information radius.
\newblock {\em Zeitschrift f{\"u}r Wahrscheinlichkeitstheorie und Verwandte
  Gebiete}, 14:149--160, 1969.
\newblock
  \texttt{\href{http://dx.doi.org/10.1007/BF00537520}{DOI:\,10.1007/BF00537520}}.

\bibitem{csiszar1995generalized}
Imre Csisz\'ar.
\newblock {Generalized Cutoff Rates and R\'enyi's Information Measures}.
\newblock {\em IEEE Transactions on Information Theory}, 41(1):26--34, 1995.
\newblock
  \texttt{\href{http://dx.doi.org/10.1109/18.370121}{DOI:\,10.1109/18.370121}}.

\bibitem{ho2015convexity}
Siu-Wai Ho and Sergio Verd\'u.
\newblock {Convexity/Concavity of Renyi Entropy and $\alpha$-Mutual
  Information}.
\newblock In {\em 2015 IEEE International Symposium on Information Theory
  (ISIT)}, pages 745--749, 2015.
\newblock
  \texttt{\href{http://dx.doi.org/10.1109/ISIT.2015.7282554}{DOI:\,10.1109/ISIT.2015.7282554}}.

\bibitem{verdu2015alpha}
Sergio Verd\'u.
\newblock {$\alpha$-Mutual Information}.
\newblock In {\em 2015 Information Theory and Applications Workshop (ITA)},
  2015.
\newblock
  \texttt{\href{http://dx.doi.org/10.1109/ITA.2015.7308959}{DOI:\,10.1109/ITA.2015.7308959}}.

\bibitem{verdu2021error}
Sergio Verd\'u.
\newblock {Error Exponents and $\alpha$-Mutual Information}.
\newblock {\em Entropy}, 23(2):199, 2021.
\newblock
  \texttt{\href{http://dx.doi.org/10.3390/e23020199}{DOI:\,10.3390/e23020199}}.

\bibitem{esposito2022sibsons}
Amedeo~Roberto Esposito, Adrien Vandenbroucque, and Michael Gastpar.
\newblock {On Sibson's $\alpha$-Mutual Information}.
\newblock In {\em 2022 IEEE International Symposium on Information Theory
  (ISIT)}, pages 2904--2909, 2022.
\newblock
  \texttt{\href{http://dx.doi.org/10.1109/ISIT50566.2022.9834428}{DOI:\,10.1109/ISIT50566.2022.9834428}}.

\bibitem{esposito2024sibsons}
Amedeo~Roberto Esposito, Michael Gastpar, and Ibrahim Issa.
\newblock {Sibson's $\alpha$-Mutual Information and its Variational
  Representations}, 2024.
\newblock
  \texttt{\href{http://dx.doi.org/10.48550/arXiv.2405.08352}{DOI:\,10.48550/arXiv.2405.08352}}.

\bibitem{tomamichel2018operational}
Marco Tomamichel and Masahito Hayashi.
\newblock {Operational Interpretation of R\'enyi Information Measures via
  Composite Hypothesis Testing Against Product and Markov Distributions}.
\newblock {\em IEEE Transactions on Information Theory}, 64(2):1064--1082,
  2018.
\newblock
  \texttt{\href{http://dx.doi.org/10.1109/TIT.2017.2776900}{DOI:\,10.1109/TIT.2017.2776900}}.

\bibitem{hoeffding1965asymptotically}
Wassily Hoeffding.
\newblock Asymptotically optimal tests for multinomial distributions.
\newblock {\em The Annals of Mathematical Statistics}, 36(2):369--401, 1965.
\newblock
  \texttt{\href{http://dx.doi.org/10.1214/AoMS/1177700150}{DOI:\,10.1214/AoMS/1177700150}}.

\bibitem{hoeffding1965probabilities}
Wassily Hoeffding.
\newblock On probabilities of large deviations.
\newblock {\em Proceedings of the Fifth Berkeley Symposium on Mathematical
  Statistics and Probability}, 5(1):203--219, 1967.

\bibitem{csiszar1971error}
Imre Csisz\'ar and Giuseppe Longo.
\newblock On the error exponent for source coding and for testing simple
  statistical hypotheses.
\newblock {\em Studia Scientiarum Mathematicarum Hungarica}, 6:181--191, 1971.

\bibitem{audenaert2008asymptotic}
Koenraad M.~R. Audenaert, Michael Nussbaum, Arleta Szko\l{}a, and Frank
  Verstraete.
\newblock {Asymptotic Error Rates in Quantum Hypothesis Testing}.
\newblock {\em Communications in Mathematical Physics}, 279(1):251--283, 2008.
\newblock
  \texttt{\href{http://dx.doi.org/10.1007/s00220-008-0417-5}{DOI:\,10.1007/s00220-008-0417-5}}.

\bibitem{blahut1974hypothesis}
Richard Blahut.
\newblock {Hypothesis Testing and Information Theory}.
\newblock {\em IEEE Transactions on Information Theory}, 20(4):405--417, 1974.
\newblock
  \texttt{\href{http://dx.doi.org/10.1109/TIT.1974.1055254}{DOI:\,10.1109/TIT.1974.1055254}}.

\bibitem{han1989strong}
Te~Sun Han and Kingo Kobayashi.
\newblock The strong converse theorem for hypothesis testing.
\newblock {\em IEEE Transactions on Information Theory}, 35(1):178--180, 1989.
\newblock
  \texttt{\href{http://dx.doi.org/10.1109/18.42188}{DOI:\,10.1109/18.42188}}.

\bibitem{nakagawa1993converse}
Kenji Nakagawa and Fumio Kanaya.
\newblock {On the Converse Theorem in Statistical Hypothesis Testing}.
\newblock {\em IEEE Transactions on Information Theory}, 39(2):623--628, 1993.
\newblock
  \texttt{\href{http://dx.doi.org/10.1109/18.212293}{DOI:\,10.1109/18.212293}}.

\bibitem{gupta2014multiplicativity}
Manish~K. Gupta and Mark~M. Wilde.
\newblock {Multiplicativity of Completely Bounded $p$-Norms Implies a Strong
  Converse for Entanglement-Assisted Capacity}.
\newblock {\em Communications in Mathematical Physics}, 334(2):867--887, 2015.
\newblock
  \texttt{\href{http://dx.doi.org/10.1007/s00220-014-2212-9}{DOI:\,10.1007/s00220-014-2212-9}}.

\bibitem{hayashi2016correlation}
Masahito Hayashi and Marco Tomamichel.
\newblock {Correlation detection and an operational interpretation of the
  R\'enyi mutual information}.
\newblock {\em Journal of Mathematical Physics}, 57(10):102201, 2016.
\newblock
  \texttt{\href{http://dx.doi.org/10.1063/1.4964755}{DOI:\,10.1063/1.4964755}}.

\bibitem{mckinlay2020decomposition}
Alexander McKinlay and Marco Tomamichel.
\newblock {Decomposition rules for quantum R\'enyi mutual information with an
  application to information exclusion relations}.
\newblock {\em Journal of Mathematical Physics}, 61(072202), 2020.
\newblock
  \texttt{\href{http://dx.doi.org/10.1063/1.5143862}{DOI:\,10.1063/1.5143862}}.

\bibitem{petz1986quasi}
D\'enes Petz.
\newblock Quasi-entropies for finite quantum systems.
\newblock {\em Reports on Mathematical Physics}, 23(1):57--65, 1986.
\newblock
  \texttt{\href{http://dx.doi.org/10.1016/0034-4877(86)90067-4}{DOI:\,10.1016/0034-4877(86)90067-4}}.

\bibitem{wilde2014strong}
Mark~M. Wilde, Andreas Winter, and Dong Yang.
\newblock {Strong Converse for the Classical Capacity of Entanglement-Breaking
  and Hadamard Channels via a Sandwiched R\'enyi Relative Entropy}.
\newblock {\em Communications in Mathematical Physics}, 331(2):593--622, 2014.
\newblock
  \texttt{\href{http://dx.doi.org/10.1007/s00220-014-2122-x}{DOI:\,10.1007/s00220-014-2122-x}}.

\bibitem{mueller2013quantum}
Martin M\"uller-Lennert, Fr\'ed\'eric Dupuis, Oleg Szehr, Serge Fehr, and Marco
  Tomamichel.
\newblock {On quantum R\'enyi entropies: A new generalization and some
  properties}.
\newblock {\em Journal of Mathematical Physics}, 54(12):122203, 2013.
\newblock
  \texttt{\href{http://dx.doi.org/10.1063/1.4838856}{DOI:\,10.1063/1.4838856}}.

\bibitem{hayashi2007error}
Masahito Hayashi.
\newblock Error exponent in asymmetric quantum hypothesis testing and its
  application to classical-quantum channel coding.
\newblock {\em Physical Review A}, 76(062301), 2007.
\newblock
  \texttt{\href{http://dx.doi.org/10.1103/PhysRevA.76.062301}{DOI:\,10.1103/PhysRevA.76.062301}}.

\bibitem{nagaoka2006converse}
Hiroshi Nagaoka.
\newblock {The Converse Part of The Theorem for Quantum Hoeffding Bound}, 2006.
\newblock
  \texttt{\href{http://dx.doi.org/10.48550/arXiv.quant-ph/0611289}{DOI:\,10.48550/arXiv.quant-ph/0611289}}.

\bibitem{mosonyi2014quantum}
Mil\'an Mosonyi and Tomohiro Ogawa.
\newblock {Quantum Hypothesis Testing and the Operational Interpretation of the
  Quantum R\'enyi Relative Entropies}.
\newblock {\em Communications in Mathematical Physics}, 334(3):1617--1648,
  2015.
\newblock
  \texttt{\href{http://dx.doi.org/10.1007/s00220-014-2248-x}{DOI:\,10.1007/s00220-014-2248-x}}.

\bibitem{mosonyi2015two}
Mil\'an Mosonyi and Tomohiro Ogawa.
\newblock {Two Approaches to Obtain the Strong Converse Exponent of Quantum
  Hypothesis Testing for General Sequences of Quantum States}.
\newblock {\em IEEE Transactions on Information Theory}, 61(12):6975--6994,
  2015.
\newblock
  \texttt{\href{http://dx.doi.org/10.1109/TIT.2015.2489259}{DOI:\,10.1109/TIT.2015.2489259}}.

\bibitem{burri2025prmisrmi2}
Laura Burri.
\newblock {Doubly minimized Petz and sandwiched R\'enyi mutual information:
  Operational interpretation from binary quantum state discrimination}, 2026.
\newblock
  \texttt{\href{http://dx.doi.org/10.48550/arXiv.2406.03213}{DOI:\,10.48550/arXiv.2406.03213}}.

\bibitem{kudlerflam2023renyi}
Jonah Kudler-Flam, Laimei Nie, and Akash Vijay.
\newblock R\'enyi mutual information in quantum field theory, tensor networks,
  and gravity.
\newblock {\em Journal of High Energy Physics}, 2024(6):195, 2024.
\newblock
  \texttt{\href{http://dx.doi.org/10.1007/JHEP06(2024)195}{DOI:\,10.1007/JHEP06(2024)195}}.

\bibitem{kudlerflam2023renyi1}
Jonah Kudler-Flam.
\newblock {R\'enyi Mutual Information in Quantum Field Theory}.
\newblock {\em Physical Review Letters}, 130(021603), 2023.
\newblock
  \texttt{\href{http://dx.doi.org/10.1103/PhysRevLett.130.021603}{DOI:\,10.1103/PhysRevLett.130.021603}}.

\bibitem{cheng2023simple}
Hao-Chung Cheng.
\newblock {Simple and Tighter Derivation of Achievability for Classical
  Communication Over Quantum Channels}.
\newblock {\em PRX Quantum}, 4:040330, 2023.
\newblock
  \texttt{\href{http://dx.doi.org/10.1103/PRXQuantum.4.040330}{DOI:\,10.1103/PRXQuantum.4.040330}}.

\bibitem{cheng2023tight}
Hao-Chung Cheng and Li~Gao.
\newblock {Tight One-Shot Analysis for Convex Splitting with Applications in
  Quantum Information Theory}, 2023.
\newblock
  \texttt{\href{http://dx.doi.org/10.48550/arXiv.2304.12055}{DOI:\,10.48550/arXiv.2304.12055}}.

\bibitem{cheng2024error}
Hao-Chung Cheng and Li~Gao.
\newblock {Error Exponent and Strong Converse for Quantum Soft Covering}.
\newblock {\em IEEE Transactions on Information Theory}, 70(5):3499--3511,
  2024.
\newblock
  \texttt{\href{http://dx.doi.org/10.1109/TIT.2023.3307437}{DOI:\,10.1109/TIT.2023.3307437}}.

\bibitem{cheng2025errorexponentsquantumpacking}
Hao-Chung Cheng and Po-Chieh Liu.
\newblock {Error Exponents for Quantum Packing Problems via An Operator Layer
  Cake Theorem}, 2025.
\newblock
  \texttt{\href{http://dx.doi.org/10.48550/arXiv.2507.06232}{DOI:\,10.48550/arXiv.2507.06232}}.

\bibitem{berta2021composite}
Mario Berta, Fernando G. S.~L. Brand\~{a}o, and Christoph Hirche.
\newblock {On Composite Quantum Hypothesis Testing}.
\newblock {\em Communications in Mathematical Physics}, 385(1):55--77, 2021.
\newblock
  \texttt{\href{http://dx.doi.org/10.1007/s00220-021-04133-8}{DOI:\,10.1007/s00220-021-04133-8}}.

\bibitem{zhai2023chain}
Yuan Zhai, Bo~Yang, and Zhengjun Xi.
\newblock {Chain rules for a mutual information based on R\'enyi zero-relative
  entropy}.
\newblock {\em Physical Review A}, 108(1):012413, 2023.
\newblock
  \texttt{\href{http://dx.doi.org/10.1103/PhysRevA.108.012413}{DOI:\,10.1103/PhysRevA.108.012413}}.

\bibitem{berta2011quantum}
Mario Berta, Matthias Christandl, and Renato Renner.
\newblock {The Quantum Reverse Shannon Theorem Based on One-Shot Information
  Theory}.
\newblock {\em Communications in Mathematical Physics}, 306(3):579--615, 2011.
\newblock
  \texttt{\href{http://dx.doi.org/10.1007/s00220-011-1309-7}{DOI:\,10.1007/s00220-011-1309-7}}.

\bibitem{datta2013oneshot}
Nilanjana Datta, Joseph~M. Renes, Renato Renner, and Mark~M. Wilde.
\newblock {One-Shot Lossy Quantum Data Compression}.
\newblock {\em IEEE Transactions on Information Theory}, 59(12):8057--8076,
  2013.
\newblock
  \texttt{\href{http://dx.doi.org/10.1109/TIT.2013.2283723}{DOI:\,10.1109/TIT.2013.2283723}}.

\bibitem{ciganovic2014smooth}
Nikola Ciganovi\'c, Normand~J. Beaudry, and Renato Renner.
\newblock {Smooth Max-Information as One-Shot Generalization for Mutual
  Information}.
\newblock {\em IEEE Transactions on Information Theory}, 60(3):1573--1581,
  2014.
\newblock
  \texttt{\href{http://dx.doi.org/10.1109/TIT.2013.2295314}{DOI:\,10.1109/TIT.2013.2295314}}.

\bibitem{berta2014identifying}
Mario Berta, Joseph~M. Renes, and Mark~M. Wilde.
\newblock {Identifying the Information Gain of a Quantum Measurement}.
\newblock {\em IEEE Transactions on Information Theory}, 60(12):7987--8006,
  2014.
\newblock
  \texttt{\href{http://dx.doi.org/10.1109/TIT.2014.2365207}{DOI:\,10.1109/TIT.2014.2365207}}.

\bibitem{anshu2017quantum}
Anurag Anshu, Vamsi~Krishna Devabathini, and Rahul Jain.
\newblock {Quantum Communication Using Coherent Rejection Sampling}.
\newblock {\em Physical Review Letters}, 119:120506, 2017.
\newblock
  \texttt{\href{http://dx.doi.org/10.1103/PhysRevLett.119.120506}{DOI:\,10.1103/PhysRevLett.119.120506}}.

\bibitem{scalet2021computablerenyi}
Samuel~O. Scalet, {\'A}lvaro~M. Alhambra, Georgios Styliaris, and J.~Ignacio
  Cirac.
\newblock {Computable R\'enyi mutual information: Area laws and correlations}.
\newblock {\em {Quantum}}, 5:541, 2021.
\newblock
  \texttt{\href{http://dx.doi.org/10.22331/q-2021-09-14-541}{DOI:\,10.22331/q-2021-09-14-541}}.

\bibitem{beigi2013sandwiched}
Salman Beigi.
\newblock {Sandwiched R\'enyi divergence satisfies data processing inequality}.
\newblock {\em Journal of Mathematical Physics}, 54(12), 2013.
\newblock
  \texttt{\href{http://dx.doi.org/10.1063/1.4838855}{DOI:\,10.1063/1.4838855}}.

\bibitem{leditzky2016strong}
Felix Leditzky, Mark~M. Wilde, and Nilanjana Datta.
\newblock {Strong converse theorems using R\'enyi entropies}.
\newblock {\em Journal of Mathematical Physics}, 57(8), 2016.
\newblock
  \texttt{\href{http://dx.doi.org/10.1063/1.4960099}{DOI:\,10.1063/1.4960099}}.

\bibitem{mosonyi2015coding}
Mil\'an Mosonyi.
\newblock {Coding Theorems for Compound Problems via Quantum R\'enyi
  Divergences}.
\newblock {\em IEEE Transactions on Information Theory}, 61(6):2997--3012,
  2015.
\newblock
  \texttt{\href{http://dx.doi.org/10.1109/TIT.2015.2417877}{DOI:\,10.1109/TIT.2015.2417877}}.

\bibitem{mosonyi2017strong}
Mil\'an Mosonyi and Tomohiro Ogawa.
\newblock {Strong Converse Exponent for Classical-Quantum Channel Coding}.
\newblock {\em Communications in Mathematical Physics}, 355(1):373--426, 2017.
\newblock
  \texttt{\href{http://dx.doi.org/10.1007/s00220-017-2928-4}{DOI:\,10.1007/s00220-017-2928-4}}.

\bibitem{li2022reliability}
Ke~Li and Yongsheng Yao.
\newblock {Reliability Function of Quantum Information Decoupling via the
  Sandwiched R\'enyi Divergence}.
\newblock {\em Communications in Mathematical Physics}, 405(160), 2024.
\newblock
  \texttt{\href{http://dx.doi.org/10.1007/s00220-024-05029-z}{DOI:\,10.1007/s00220-024-05029-z}}.

\bibitem{cheng2023quantum}
Mario Berta, Hao-Chung Cheng, and Li~Gao.
\newblock {Quantum Broadcast Channel Simulation via Multipartite Convex
  Splitting}.
\newblock {\em Communications in Mathematical Physics}, 406(36), 2025.
\newblock
  \texttt{\href{http://dx.doi.org/10.1007/s00220-024-05191-4}{DOI:\,10.1007/s00220-024-05191-4}}.

\bibitem{bluhm2023unified}
Andreas Bluhm, Angela Capel, Paul Gondolf, and Tim M\"obus.
\newblock {Unified Framework for Continuity of Sandwiched R\'enyi Divergences}.
\newblock {\em Annales Henri Poincar\'e}, 2024.
\newblock
  \texttt{\href{http://dx.doi.org/10.1007/s00023-024-01519-x}{DOI:\,10.1007/s00023-024-01519-x}}.

\bibitem{li2024operational}
Ke~Li and Yongsheng Yao.
\newblock {Operational Interpretation of the Sandwiched R\'enyi Divergence of
  Order 1/2 to 1 as Strong Converse Exponents}.
\newblock {\em Communications in Mathematical Physics}, 405(22), 2024.
\newblock
  \texttt{\href{http://dx.doi.org/10.1007/s00220-023-04890-8}{DOI:\,10.1007/s00220-023-04890-8}}.

\bibitem{pusz1975functional}
Wies{\l}aw Pusz and Stanis{\l}aw~L. Woronowicz.
\newblock Functional calculus for sesquilinear forms and the purification map.
\newblock {\em Reports on Mathematical Physics}, 8(2):159--170, 1975.
\newblock
  \texttt{\href{http://dx.doi.org/10.1016/0034-4877(75)90061-0}{DOI:\,10.1016/0034-4877(75)90061-0}}.

\bibitem{ando1987some}
Tsuyoshi Ando.
\newblock On some operator inequalities.
\newblock {\em Mathematische Annalen}, 279(1):157--159, 1987.
\newblock
  \texttt{\href{http://dx.doi.org/10.1007/BF01456197}{DOI:\,10.1007/BF01456197}}.

\bibitem{kubo1980means}
Fumio Kubo and Tsuyoshi Ando.
\newblock Means of positive linear operators.
\newblock {\em Mathematische Annalen}, 246(3):205--224, 1980.
\newblock
  \texttt{\href{http://dx.doi.org/10.1007/BF01371042}{DOI:\,10.1007/BF01371042}}.

\bibitem{renner2006security}
Renato Renner.
\newblock {Security of Quantum Key Distribution}, 2006.
\newblock
  \texttt{\href{http://dx.doi.org/10.48550/arXiv.quant-ph/0512258}{DOI:\,10.48550/arXiv.quant-ph/0512258}}.

\bibitem{christandl2009postselection}
Matthias Christandl, Robert K\"onig, and Renato Renner.
\newblock {Postselection Technique for Quantum Channels with Applications to
  Quantum Cryptography}.
\newblock {\em Physical Review Letters}, 102(2), 2009.
\newblock
  \texttt{\href{http://dx.doi.org/10.1103/PhysRevLett.102.020504}{DOI:\,10.1103/PhysRevLett.102.020504}}.

\bibitem{tomamichel2013hierarchy}
Marco Tomamichel and Masahito Hayashi.
\newblock {A Hierarchy of Information Quantities for Finite Block Length
  Analysis of Quantum Tasks}.
\newblock {\em IEEE Transactions on Information Theory}, 59(11):7693--7710,
  2013.
\newblock
  \texttt{\href{http://dx.doi.org/10.1109/TIT.2013.2276628}{DOI:\,10.1109/TIT.2013.2276628}}.

\bibitem{li2014second}
Ke~Li.
\newblock Second-order asymptotics for quantum hypothesis testing.
\newblock {\em The Annals of Statistics}, 42(1):171--189, 2014.
\newblock
  \texttt{\href{http://dx.doi.org/10.1214/13-AoS1185}{DOI:\,10.1214/13-AoS1185}}.

\bibitem{ohya1993quantum}
Masanori Ohya and D{\'e}nes Petz.
\newblock {\em Quantum Entropy and Its Use}.
\newblock Springer, Berlin, 1993.

\bibitem{nussbaum2009chernoff}
Michael Nussbaum and Arleta Szko\l{}a.
\newblock {The Chernoff lower bound for symmetric quantum hypothesis testing}.
\newblock {\em The Annals of Statistics}, 37(2):1040--1057, 2009.
\newblock
  \texttt{\href{http://dx.doi.org/10.1214/08-AoS593}{DOI:\,10.1214/08-AoS593}}.

\bibitem{lin2015investigating}
Simon~M. Lin and Marco Tomamichel.
\newblock {Investigating properties of a family of quantum R\'enyi
  divergences}.
\newblock {\em Quantum Information Processing}, 14(4):1501--1512, 2015.
\newblock
  \texttt{\href{http://dx.doi.org/10.1007/s11128-015-0935-y}{DOI:\,10.1007/s11128-015-0935-y}}.

\bibitem{tomamichel2016quantum}
Marco Tomamichel.
\newblock {\em {Quantum Information Processing with Finite Resources}}.
\newblock Springer, 2016.
\newblock
  \texttt{\href{http://dx.doi.org/10.1007/978-3-319-21891-5}{DOI:\,10.1007/978-3-319-21891-5}}.

\bibitem{datta2009min}
Nilanjana Datta.
\newblock {Min- and Max-Relative Entropies and a New Entanglement Monotone}.
\newblock {\em IEEE Transactions on Information Theory}, 55(6):2816--2826,
  2009.
\newblock
  \texttt{\href{http://dx.doi.org/10.1109/TIT.2009.2018325}{DOI:\,10.1109/TIT.2009.2018325}}.

\bibitem{frank2013monotonicity}
Rupert~L. Frank and Elliott~H. Lieb.
\newblock {Monotonicity of a relative R\'enyi entropy}.
\newblock {\em Journal of Mathematical Physics}, 54(12), 2013.
\newblock
  \texttt{\href{http://dx.doi.org/10.1063/1.4838835}{DOI:\,10.1063/1.4838835}}.

\bibitem{lieb1991inequalities}
Elliott~H. Lieb and Walter~E. Thirring.
\newblock {\em {Inequalities for the Moments of the Eigenvalues of the
  Schr\"odinger Hamiltonian and Their Relation to Sobolev Inequalities}}, pages
  135--169.
\newblock Springer Berlin Heidelberg, Berlin, Heidelberg, 1991.
\newblock
  \texttt{\href{http://dx.doi.org/10.1007/978-3-662-02725-7}{DOI:\,10.1007/978-3-662-02725-7}}.

\bibitem{araki1990inequality}
Huzihiro Araki.
\newblock {On an inequality of Lieb and Thirring}.
\newblock {\em Letters in Mathematical Physics}, 19:167--170, 1990.
\newblock
  \texttt{\href{http://dx.doi.org/10.1007/BF01045887}{DOI:\,10.1007/BF01045887}}.

\bibitem{bhatia1996matrix}
Rajendra Bhatia.
\newblock {\em {Matrix Analysis}}.
\newblock Springer, 1996.
\newblock
  \texttt{\href{http://dx.doi.org/10.1007/978-1-4612-0653-8}{DOI:\,10.1007/978-1-4612-0653-8}}.

\bibitem{cheng2022properties}
Hao-Chung Cheng, Li~Gao, and Min-Hsiu Hsieh.
\newblock {Properties of Noncommutative R\'enyi and Augustin Information}.
\newblock {\em {Communications in Mathematical Physics}}, 390:501--544, 2022.
\newblock
  \texttt{\href{http://dx.doi.org/10.1007/s00220-022-04319-8}{DOI:\,10.1007/s00220-022-04319-8}}.

\bibitem{bhatia2007positive}
Rajendra Bhatia.
\newblock {\em {Positive Definite Matrices}}.
\newblock Princeton University Press, 2007.

\bibitem{evert2022convexity}
Eric Evert, Scott McCullough, Tea \v{S}trekelj, and Anna Vershynina.
\newblock Convexity of a certain operator trace functional.
\newblock {\em Linear Algebra and its Applications}, 643:218--234, 2022.
\newblock
  \texttt{\href{http://dx.doi.org/10.1016/J.LAA.2022.02.033}{DOI:\,10.1016/J.LAA.2022.02.033}}.

\bibitem{epstein1973remarks}
Henri Epstein.
\newblock {Remarks on Two Theorems of E. Lieb}.
\newblock {\em Communications in Mathematical Physics}, 31:317--325, 1973.
\newblock
  \texttt{\href{http://dx.doi.org/10.1007/BF01646492}{DOI:\,10.1007/BF01646492}}.

\bibitem{carlen2008minkowski}
Eric~A. Carlen and Elliott~H. Lieb.
\newblock {A Minkowski Type Trace Inequality and Strong Subadditivity of
  Quantum Entropy II: Convexity and Concavity}.
\newblock {\em Letters in Mathematical Physics}, 83(2):107--126, 2008.
\newblock
  \texttt{\href{http://dx.doi.org/10.1007/s11005-008-0223-1}{DOI:\,10.1007/s11005-008-0223-1}}.

\bibitem{hiai2013concavity}
Fumio Hiai.
\newblock Concavity of certain matrix trace and norm functions.
\newblock {\em Linear Algebra and its Applications}, 439(5):1568--1589, 2013.
\newblock
  \texttt{\href{http://dx.doi.org/10.1016/J.LAA.2013.04.020}{DOI:\,10.1016/J.LAA.2013.04.020}}.

\bibitem{sion1958general}
Maurice Sion.
\newblock On general minimax theorems.
\newblock {\em Pacific Journal of Mathematics}, 8(1):171--176, 1958.
\newblock
  \texttt{\href{http://dx.doi.org/10.2140/PJM.1958.8.171}{DOI:\,10.2140/PJM.1958.8.171}}.

\end{thebibliography}
